\documentclass[pra,twocolumn,superscriptaddress,showpacs,nobibnotes,notitlepage,nofootinbib]{revtex4-2}

\usepackage[utf8]{inputenc}
\usepackage[margin=1in]{geometry} 
\usepackage{amsmath,amsthm,amssymb,hyperref}
\usepackage[dvipsnames]{xcolor}
\hypersetup{
    colorlinks=true,
    linkcolor=RoyalBlue,
    citecolor=magenta,
    filecolor=magenta,      
    urlcolor=blue,
    pdftitle={},
    pdfpagemode=FullScreen,
    }
\usepackage{adjustbox}
\usepackage{graphicx}
\usepackage{braket}
\usepackage[ruled,linesnumbered]{algorithm2e}
\usepackage{float}
\usepackage{thm-restate}
\usepackage{cleveref}
\usepackage{xcolor}
\usepackage{soul}
\usepackage{algorithm2e}
\usepackage{tikz}
\usetikzlibrary{quantikz2}

\newtheorem{definition}{Definition}

\newtheorem{corollary}{Corollary}
\newtheorem{proposition}{Proposition}

\newtheorem*{remark*}{Remark}
\usepackage{xcolor}
\usepackage{soul}
\usepackage{comment}
\usepackage{booktabs}
\usepackage{multirow}
\tikzset{
noisy_1/.style={fill=red!50,line
width=1pt,inner xsep=0pt,inner ysep=0pt}
}
\tikzset{
noisy_2/.style={fill=orange!50,line
width=1pt, inner xsep=0pt,inner ysep=0pt}
}

\SetCommentSty{mycommfont}

\usepackage{soul}
\usepackage[dvipsnames]{xcolor}
\sethlcolor{yellow}

\makeatletter
\newcommand\k@t[1]{{|{#1}\rangle}}
\renewcommand{\ket}[1]{\left| #1 \right\rangle}
\makeatother

{{\color{blue}}}

\begin{document}

\title{Efficient simulation of logical magic state preparation protocols}

\author{Samyak Surti}
\email{spsurti@ucdavis.edu}
\affiliation{Department of Computer Science, University of California, Davis, CA, 95616, USA }
\author{Lucas Daguerre}
\affiliation{Department of Physics and Astronomy, University of California, Davis, CA, 95616, USA}
\author{Isaac H. Kim}
\affiliation{Department of Computer Science, University of California, Davis, CA, 95616, USA }
\date{\today}

\begin{abstract}
 
Developing space- and time-efficient logical magic state preparation protocols will likely be an essential step toward building a large-scale fault-tolerant quantum computer. Motivated by this need, we introduce a scalable method for simulating logical magic state preparation protocols under the standard circuit-level noise model. When applied to protocols based on code-switching, magic state cultivation, and magic state distillation, our method yields a complexity polynomial in (i) the number of qubits and (ii) the nonstabilizerness, e.g., stabilizer rank or Pauli rank, of the target encoded magic state. The efficiency of our simulation method is rooted in a curious fact: every circuit-level Pauli error in these protocols propagates to a Clifford error at the end. This property is satisfied by a large family of protocols, including those that repeatedly measure a transversal Clifford that squares to a Pauli. We provide a proof-of-principle numerical simulation that prepares a magic state using such logical Clifford measurements. Our work enables practical simulation of logical magic state preparation protocols without resorting to approximations or resource-intensive state-vector simulations.

\end{abstract}

\maketitle


\tableofcontents

\section{Introduction}
\label{sec:intro}

\begin{figure*}[t!]
    \centering
    \includegraphics[width=\linewidth]{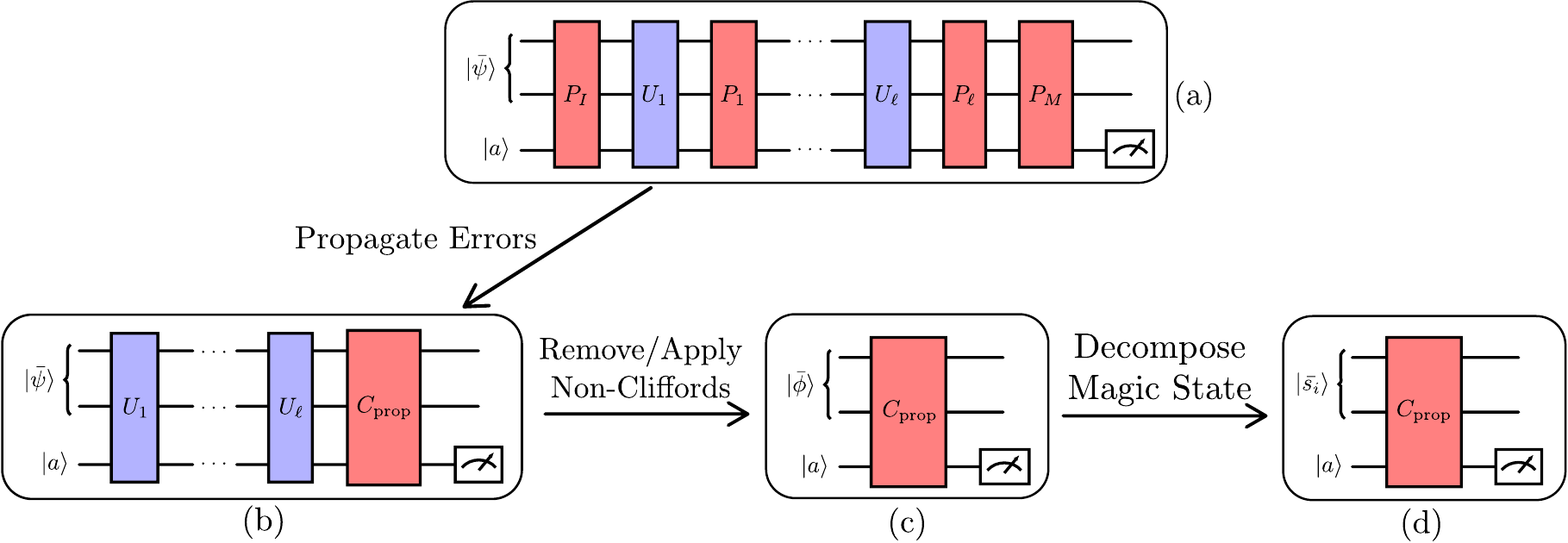}
    \caption{\textbf{Schematic representation of the simulation technique for a magic state preparation (MSP) protocol.} (a) This is the general form of the MSP protocol considered in this paper. Here $U_1,\ldots, U_\ell$ are layers of unitary gates that comprise the protocol. Circuit-level Pauli noise is denoted $P_1,\ldots, P_{\ell}$, $P_I$, and $P_M$, representing the errors associated with $U_1,\ldots, U_\ell$, initialization, and measurement, respectively. Upon measuring the ancilla register (denoted by $|a\rangle$), one either postselects or error-corrects the logical magic state. (b) For our circuit family, all circuit-level Pauli errors can be propagated to an end-of-circuit Clifford error $C_\mathrm{prop}$. (c) Because a noiseless circuit would have prepared a noiseless magic state, the original noisy protocol circuit can be recast as a Clifford error $C_\mathrm{prop}$ acting on a noiseless logical magic state $|\bar{\phi}\rangle \otimes |a\rangle= \Pi_i U_i(|\bar{\psi}\rangle \otimes |a\rangle)$. (d) To enable efficient simulation, the logical magic state is decomposed into a linear combination of stabilizer states $|\bar{s}_i\rangle$. The fidelity of the output noisy logical magic state can be obtained by simulating an ensemble of logical stabilizer states subjected to stochastic Clifford errors and measurements in the Pauli basis; see Sec.~\ref{sec:sim_technique} for more details.}
    \label{fig:simulation-steps}
\end{figure*}

In the textbook model of quantum computation, one often considers a continuous gate set consisting of arbitrary one- and two-qubit gates. However, to combat noise present in quantum computers, one must use quantum error correction~\cite{Shor1995} to encode quantum information in a logical subspace. This logical subspace must not only protect quantum information, but also support quantum computation with a set of fault-tolerant logical gates, i.e., logical gates that are robust to noise~\cite{aharonov1997fault}.

Such fault-tolerant gates are often classified into two categories: Clifford and non-Clifford gates. While there are many methods for implementing Clifford gates~\cite{dennis2002,Horsman_2012,Bombin2006,Cohen2022,Quintavalle2023partitioningqubits}, non-Clifford gates are more difficult to realize. To this end, non-Clifford gates are often carried out by preparing a special resource state called a \emph{magic state}~\cite{Bravyi_2005}. For many existing fault-tolerant quantum computing architectures, it is commonly assumed that preparing magic states is significantly more costly than Clifford gates \cite{Bravyi_2005,meier2012magic,Bravyi_2012,Jones_2013a,fowler_surface_2013,Jones_2013b,duclos_cianci_distillation_2013,Duclos_Cianci_2015,Campbell_2017,O_Gorman_2017,Haah2018codesprotocols,Gidney2019efficientmagicstate}. For this reason, quantum algorithm development had been rooted in a model in which non-Clifford gates are significantly more costly than their Clifford counterparts~\cite{Reiher2017}.

However, in recent years there has been significant progress in improving the efficiency of magic state preparation (MSP) protocols~\cite{Li_2015,Yoder_2017,Chao_2018a,Chao_2018b,Chamberland:2020axi,Itogawa_2025,Gidney:2024alh,Daguerre:2024gjd,sahay2025foldtransversalsurfacecodecultivation,claes2025cultivatingtstatessurface,Vaknin:2025pbp,Chen:2025imz}. A common feature of these protocols is carefully designed circuits that can efficiently detect errors during the preparation of a magic state. With this progress, the cost of preparing logical magic states is getting closer to that of preparing logical Clifford gates. Moreover, a series of experimental demonstrations for preparing logical magic states have been performed on various quantum devices~\cite{Lacroix:2024vls,Rodriguez:2024bhh,pogorelov2024experimentalfaulttolerantcodeswitching,Kim:2024vmw,Ye:2023hxg,Gupta:2023zei,Postler_2022,Anderson2021,dasu2025breakingmagicdemonstrationhighfidelity,Rosenfeld:2025xvf,Daguerre:2025boq}, some of which are based on these new protocols~\cite{Daguerre:2025boq}. These recent developments call for a careful study of more general MSP 
protocols, which will in turn inform future architectures for fault-tolerant quantum computers and algorithms. 

Unfortunately, understanding the exact performance of these MSP protocols is significantly more challenging than for protocols implementing Clifford gates. For the latter, classical simulation can be carried out using the standard Gottesman-Knill theorem~\cite{Gottesman:1998hu,Aaronson:2004xuh}. However, MSP protocols involve non-Clifford gates and, therefore, cannot be simulated this way. In principle, extended stabilizer simulation methods may be used to simulate these protocols~\cite{Bravyi:2016yhm,Bravyi:2016uqx,Bravyi:2018ugg,Bu:2019qed,Masot-Llima:2024doz,Zhang:2024jlz,Garner:2025xts,Aziz:2025cnu}. However, the complexity of these simulation methods generally scales exponentially with the number of non-Clifford gates. In conjunction, the number of non-Clifford gates used in these protocols scales with the code size, ultimately making the cost of simulation exponential in the number of qubits. To the best of our knowledge, there is no known efficient method for simulating these protocols.

In this paper, we propose an efficient and broadly applicable method to simulate MSP protocols. When applicable, our method yields a computational complexity scaling polynomially with the number of qubits and a quantity that quantifies the nonstabilizerness of the target encoded magic state, e.g., stabilizer rank~\cite{Bravyi:2018ugg} or Pauli rank~\cite{Bu:2019qed}. As a concrete example, our method can be used to efficiently simulate the magic state cultivation protocol~\cite{Gidney:2024alh}. Using a distance-$d$ color code~\cite{Bombin2006}, this protocol applies $\Omega(d^2)$ non-Clifford gates on $\Theta(d^2)$ qubits. As such, the simulation cost would naively scale exponentially in $d$. However, because the target encoded magic state has a stabilizer rank of $2$ independent of $d$, our method can enable the simulation of this protocol in time polynomial in $d$. 

More broadly, our method is applicable to (i) protocols based on a measurement of certain logical Clifford operator and (ii) protocols based on a transversal non-Clifford gate in the third level of the Clifford hierarchy. Note that there are many different ways of implementing these protocols in practice. They may be used at the logical level to perform magic state distillation~\cite{Bravyi_2005,meier2012magic}. Alternatively, one may opt to employ them at the physical level~\cite{Goto:2016gss,Chamberland_2018,Chamberland:2019ehl,Chamberland:2020axi,Itogawa_2025,Gidney:2024alh,Daguerre:2024gjd,Butt_2024}, using flag qubits to suppress errors~\cite{Yoder_2017,Chao_2018a,Chao_2018b}. Either way, the simulation cost ends up being polynomial in the number of qubits and the stabilizer/Pauli rank~\cite{Bravyi:2018ugg,Bu:2019qed} of the target encoded magic state.

The main insight that enables efficient simulation is the special structure of the propagated circuit-level error. More precisely, we observe that circuit-level Pauli errors in these protocols propagate to a Clifford error at the end of the protocol [Fig.~\ref{fig:simulation-steps}], despite the presence of many non-Clifford gates. In turn, noisy circuits containing non-Clifford gates can be effectively described by a noiseless logical magic state subjected to a Clifford error. We note that this fact was used in the simulation of protocols based on transversal non-Clifford gates~\cite{Daguerre:2024gjd}. However, for the protocols based on measurements of a logical Clifford, this fact is much less obvious.

To better formalize how errors propagate in these protocols, we develop a theory of \textit{Pauli-square-soot Cliffords} (PSCs). These are (non-Pauli) Clifford unitaries that square to a Pauli. Indeed, the majority of magic states considered in the literature can be viewed as $+1$-eigenstates of some PSC, hence motivating their study. We prove several mathematical facts concerning the structure of PSCs, which can be used to study error propagation through MSP protocols based on measuring logical PSCs in a systematic way. Additionally, we believe PSCs may be of interest independent of simulating MSP protocols.

Compared to existing simulation approaches, ours provide several advantages. Let us focus on cultivation-based protocols~\cite{Gidney:2024alh,sahay2025foldtransversalsurfacecodecultivation,claes2025cultivatingtstatessurface,Vaknin:2025pbp,Chen:2025imz}, which are relatively more nontrivial to simulate. In these works, the authors use the fidelity of preparing certain logical stabilizer states as a proxy for determining the fidelity of target logical magic states. For instance, Ref.~\cite{Gidney:2024alh} estimates fidelities for the logical $\ket{T} = T\ket{+}$ and $\ket{Y_+} = S\ket{+}$ state for the distance-$3$ codes using state-vector simulation. The ratio between the two is then used to deduce the fidelity of the logical $\ket{T}$ state at higher code distances. However, the assumption that this ratio remains approximately the same as a function of increasing code distance is a nontrivial one. By contrast, our method is free of such ambiguities. More recently, some works also proposed a simulation method for magic state cultivation protocols based on the ZX-calculus formalism~\cite{wan2025cuttingstabiliserdecompositionsmagic,wan2025simulatemagicstatecultivation}, whereby non-Clifford ZX-diagrams are decomposed into superpositions of a few Clifford ZX-diagrams \cite{Sutcliffe_2024}. However, the number of Clifford ZX-diagrams in this decomposition depend on the circuit-level error configuration, in contrast to our method which employs a superposition of two logical stabilizers states irrespective of the error configuration.

The rest of this paper is organized as follows. In Sec.~\ref{sec:preliminaries}, we introduce our notations and definitions. In Sec.~\ref{sec:example}, we describe a simple example that demonstrates the key ideas behind our approach. In Sec.~\ref{sec:controlled_clifford}, we introduce the notion of PSCs and describe their salient properties. In Sec.~\ref{sec:canonical_family}, we show that many existing MSP protocols in the literature, such as cultivation-like methods~\cite{Goto:2016gss,Chamberland:2019ehl,Chamberland:2020axi,Itogawa_2025,Gidney:2024alh,sahay2025foldtransversalsurfacecodecultivation,claes2025cultivatingtstatessurface,Vaknin:2025pbp}, logical magic state distillation~\cite{Bravyi_2005,meier2012magic,Bravyi_2012,Jones_2013a,fowler_surface_2013,Jones_2013b,duclos_cianci_distillation_2013,Duclos_Cianci_2015,Campbell_2017,O_Gorman_2017,Haah2018codesprotocols} and code-switching~\cite{Anderson:2014jvy,bombin2016dimensionaljumpquantumerror,Beverland2021,Butt_2024,Daguerre:2024gjd},  propagate circuit-level Pauli errors to an end-of-circuit Clifford error. In Sec.~\ref{sec:fid_method}, we describe our phase-insensitive method to estimate the fidelity of a target magic state. In Sec.~\ref{sec:sim_technique}, we provide a final overview of the simulation technique. In particular, as a proof-of-concept, we numerically simulate the protocol from Ref.~\cite{Goto:2016gss,Chamberland:2019ehl,Postler_2022} in Sec.~\ref{sub:sim-application-example}. We end with a discussion in Sec.~\ref{sec:discussion}.

\section{Preliminaries}
\label{sec:preliminaries}
We begin by introducing some useful definitions and formalisms. The $n$-qubit \textit{Pauli group}, denoted by $\mathcal{P}_n$, is generated by the $n$-fold tensor product of single-qubit Pauli matrices (including a global phase) $\mathcal{P}_n=\{\pm1,i\}\times\{I, X, Y,Z\}^{\otimes n}$. 

The \textit{Clifford hierarchy}~\cite{Gottesman_1999} is defined recursively for $k\geq 2$ as
\begin{equation}
    \mathcal{C}^{(k)}=\{V \in U(2^n): VPV^{\dagger} \in  \mathcal{C}^{(k-1)}\:,\:\forall P \in \mathcal{P}_n\}\:,
\end{equation}
where $\mathcal{C}^{(1)}=\mathcal{P}_n$ is the Pauli group, and $\mathcal{C}^{(2)}$ is the \textit{Clifford group}. Note that for $k\geq 3$, $\mathcal{C}^{(k)}$ does not form a group. Unitary gates that do not belong to $\mathcal{C}^{(2)}$ are \textit{non-Clifford gates}.

A \textit{Clifford circuit} is defined as a circuit of the following form. Without loss of generality, we consider an initial state prepared in a computational basis state $|0\rangle^{\otimes n}$. The gates in the circuit belong to the Clifford group $\mathcal{C}^{(2)}$ and measurements are performed in the Pauli basis. Moreover, \textit{Pauli stabilizer states} are states of the form $|\psi \rangle=U|0\rangle^{\otimes n}$ for $U\in \mathcal{C}^{(2)}$. These states are $+1$-eigenvectors of $UZ_iU^{\dagger} \in \mathcal{P}_n$ for $i=1,\dots,n$, where $Z_i$ acts nontrivially on the $i$-th qubit. The single-qubit Pauli stabilizer states are $|P \rangle \in \{|+\rangle,|-\rangle\}\cup \{|Y_{+}\rangle,|Y_{-}\rangle\}\cup \{|0\rangle,|1\rangle\}$, which are $\pm 1$ eigenvectors of $X$, $Y$ and $Z$, respectively. \textit{Controlled-gates} are unitary gates $C(U)$ (or $CU$) constructed from another unitary gate $U$ such that $C(U)=|0\rangle \langle 0| \otimes I + |1\rangle \langle 1| \otimes U$.

\textit{Magic states} are resource states that can be consumed to implement a non-Clifford gate~\cite{Bravyi_2005}, up to a Clifford correction. Importantly, a common approach for realizing a universal gate set is to augment the Clifford gate set with the preparation of such magic states. We note that a large class of magic states studied in the literature are $+1$-eigenstates of Clifford unitaries that we call PSC [Definition~\ref{def:square_root_pauli}]. These are Cliffords that square to a Pauli, whose salient features relevant to our work shall be described in Sec.~\ref{sec:controlled_clifford}.

For single-qubit magic states, the following examples are well-known:
\begin{equation}
\begin{aligned}
    |T\rangle &=T|+\rangle=\frac{1}{\sqrt{2}}(|0\rangle + e^{i\frac{\pi}{4}}|1\rangle), \\
    |H\rangle &=e^{-i\pi/8}SH|T\rangle\\
    &=\cos(\pi/8)|0\rangle+\sin(\pi/8)|1\rangle
    ,
    \label{eq:magic_states}
\end{aligned}
\end{equation}
which are +1-eigenstates of $TXT^{\dagger} = e^{-\frac{\pi}{4}i}SX$ and $H$, respectively. (We use the convention of $T=\text{diag}(1,e^{i\frac{\pi}{4}})$ and $S=\text{diag}(1,i)$).
For multi-qubit magic states, the following are standard examples
\begin{equation}
    |CCZ\rangle =CCZ|+\rangle^{\otimes 3}\quad,\quad |CS\rangle =CS|+\rangle^{\otimes 2}\:.
\end{equation}
Note that these are $+1$ eigenstates of $\{X_1CZ_{2,3}, X_2CZ_{1,3}, X_3CZ_{1,2} \}$ as well as $\{X_1S_2 CZ_{1,2}, X_2S_1 CZ_{1,2}\}$, respectively. These examples (except for $|H\rangle$) all have the form of $|V\rangle = V|+\rangle^{\otimes l}$, where $V$ is a diagonal gate in $V\in \mathcal{C}^{(3)}$~\cite{Bravyi:2018ugg}. These states are stabilized by the set of $C_i=VX_iV^{\dagger} \in \mathcal{C}^{(2)}$. Another example is $|CZ\rangle\propto\frac{I+CZ}{2}|+\rangle^{\otimes 2}$, which is a $+1$ eigenstate ~\cite{Gupta:2023zei,Davydova:2025ylx} of $CZ$.

All these examples can be written in the following form. Let $S_C=\langle U_i: U_i\text{ is PSC for }1\leq i\leq \ell \text{ and } [U_i,U_j]=0 \:\forall i,j\rangle$ be an abelian group generated by a commuting set of PSCs. These magic states can then be expressed as 
\begin{equation}
    |\phi\rangle \propto \sum_{c\in S_C} c|\Omega\rangle,\label{eq:eigen_PSC}
\end{equation}
where $|\Omega\rangle$ is a stabilizer state. It follows from the definition that $|\phi\rangle$ is a joint $+1$ eigenstate of all the elements of $S_C$.

The \textit{stabilizer rank}~\cite{Bravyi:2016uqx,Bravyi:2018ugg} of an $n$-qubit state $\ket{\phi}$ is the minimum number $q$ of stabilizer states $\ket{s_i}$ whose linear combination yields $\ket{\phi}$, with $\alpha_i \in \mathbb{C}$,
\begin{equation}
    \ket{\phi} = \sum_{i=1}^q \alpha_i \ket{s_i}\:.
\label{eq:stab_rak}
\end{equation}
From Eq.~\eqref{eq:eigen_PSC}, it follows that the stabilizer rank of a magic state is at most $|S_C|\leq 8^\ell$. For (possibly) mixed quantum states, a natural magic monotone to consider is \textit{Pauli rank}~\cite{Bu:2019qed}, which for an $n$-qubit state $\rho$ is the minimum number $p$ of nontrivial Paulis $P_i$ whose linear combination yields $\rho$,
\begin{equation}
    \rho = \frac{1}{2^n}I^{\otimes n}+\frac{1}{2^n}\sum_{i = 1}^{p} \beta_i P_i \:,
    \label{eq:rep_rho_Paulis}
\end{equation}
where $P_i \in \{ I,X, Y,Z\}^{\otimes n} \backslash \{I^{\otimes n}\}$ and $\beta_i=\text{Tr}(\rho P_i)$. Our definition of Pauli rank differs from~\cite{Bu:2019qed} by an additive factor of one since we only account for nontrivial Paulis in the expansion.

Lastly, we comment on the error model we consider. We will focus on arbitrary \textit{circuit-level} Pauli noise in which each quantum operation is modeled as a noiseless operation, preceded or followed by some stochastic Pauli error occurring with a given probability. For our simulations, we utilize the error model described in Appendix~\ref{app:error_model}.

\section{A Toy Example}
\label{sec:example}

In this section, we describe a toy example that demonstrates the key ideas behind our simulation method. This example is based on the $[[4,2,2]]$ code~\cite{meier2012magic} (also known as the $C_4$ code). Its stabilizer group is generated by $\{X_1X_2X_3X_4,Z_1Z_2Z_3Z_4\}$ and the logical Pauli operators are generated by $\{X_1X_2, X_1X_3, Z_1Z_2, Z_1Z_3\}$ [Fig.~\ref{fig:code422}].
\begin{figure}[!t]
    \centering
    \includegraphics[width=0.55\linewidth]{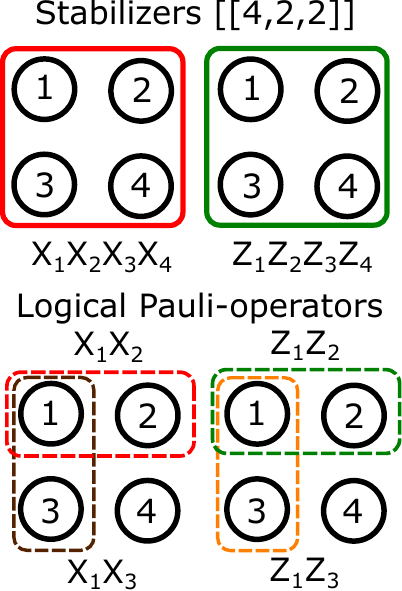}
    \caption{\textbf{Pictorial representation of the $\mathbf{[[4,2,2]]}$ code.} Each circle corresponds to a physical qubit with their respective labels. The stabilizer group generators (solid lines) and logical Pauli operators (dashed lines) are highlighted. }
    \label{fig:code422}
\end{figure}

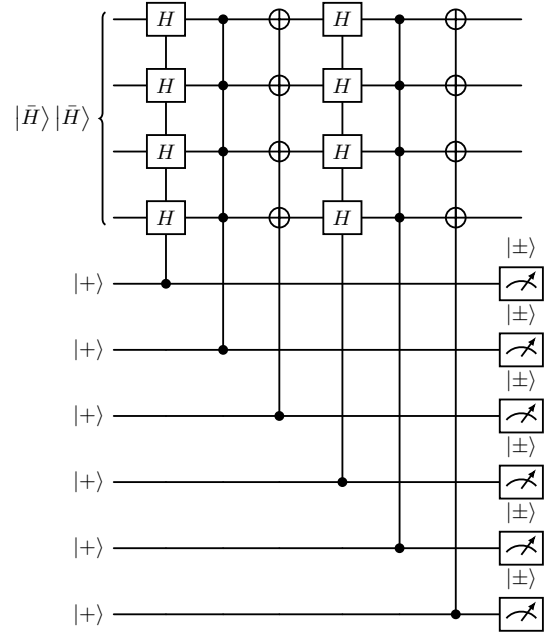
\begin{figure*}[!t]
    \centering
    \begin{adjustbox}{width=0.85\textwidth}
    \begin{quantikz}[column sep=0.25cm]
        \lstick[4]{$\ket{\bar{H}}\ket{\bar{H}}$} & \gate{H} & & & & \ctrl{0} &&&& \targ{} & & & & \gate{H} &&&& \ctrl{0} &&&& \targ{} &&&&\\
        & & \gate{H} & & & & \ctrl{0} &&&& \targ{} &&&& \gate{H} &&&& \ctrl{0} &&&& \targ{} &&& \\
        & & & \gate{H} & & & & \ctrl{0} &&&& \targ{} &&&& \gate{H} &&&& \ctrl{0} &&&& \targ{} &&\\
        & & & & \gate{H}  & & & & \ctrl{0} &&&& \targ{} &&&& \gate{H} &&&& \ctrl{0} &&&& \targ{} &\\ 
        \lstick{$\ket{+}$} & \ctrl{-4} & \ctrl{-3} & \ctrl{-2} & \ctrl{-1} &&&&&&&&&&&&&&&&&&&&& \meter{\ket{\pm}}\\
        \lstick{$\ket{+}$} &&&&& \ctrl{-5} & \ctrl{-4} & \ctrl{-3} & \ctrl{-2} &&&&&&&&&&&&&&&&& \meter{\ket{\pm}}\\
        \lstick{$\ket{+}$} &&&&&&&&& \ctrl{-6} & \ctrl{-5} & \ctrl{-4} & \ctrl{-3} &&&&&&&&&&&&& \meter{\ket{\pm}}\\
        \lstick{$\ket{+}$} &&&&&&&&&&&&& \ctrl{-7} & \ctrl{-6} & \ctrl{-5} & \ctrl{-4} &&&&&&&&& \meter{\ket{\pm}}\\
        \lstick{$\ket{+}$} &&&&&&&&&&&&&&&&& \ctrl{-8} & \ctrl{-7} & \ctrl{-6} & \ctrl{-5} &&&&& \meter{\ket{\pm}}\\
        \lstick{$\ket{+}$} &&&&&&&&&&&&&&&&&&&&& \ctrl{-9} & \ctrl{-8} & \ctrl{-7} & \ctrl{-6} & \meter{\ket{\pm}}\\
    \end{quantikz}
    \end{adjustbox}
    \caption{\textbf{A toy example protocol for preparing an encoded magic state $|\bar{H}\rangle |\bar{H}\rangle$ in the $[[4,2,2]]$ code.} The logical magic state is prepared, followed by two rounds of logical Clifford and stabilizer measurements.}
    \label{fig:422_example}
\end{figure*}
This code has two transversal Clifford gates, namely $H_1H_2H_3H_4$ and $S_1S_2S_3S_4$. We will focus on the former, whose action on the logical subspace is $(\bar{{H}}_1\otimes \bar{{H}}_2) \overline{{SWAP}}_{1,2}$, where $\bar{{H}}_k$, for $k\in \{1,2 \}$, is the logical Hadamard on the $k$'th qubit and $\overline{{SWAP}}_{1,2}$ is the logical swap gate. 

The $[[4,2,2]]$ code is a useful guiding example for understanding simulations of the protocols in Refs.~\cite{Itogawa_2025,Gidney:2024alh,sahay2025foldtransversalsurfacecodecultivation,claes2025cultivatingtstatessurface,Vaknin:2025pbp} under circuit-level noise. In these works, the authors studied a protocol to prepare a logical magic state using the following approach:
\begin{enumerate}
    \item Prepare an encoded magic state (in a manner that is not necessarily fault-tolerant).
    \item Measure a logical Clifford whose eigenstate is the encoded magic state. 
    \item Measure the stabilizers.
    \item Repeat Step 2 and 3 judiciously, and postselect on detecting no error.
\end{enumerate}
Because the $[[4,2,2]]$ code has a nontrivial transversal logical Clifford gate, we can envision a similar protocol using the same template. 

One toy example of such a circuit is shown in Figure~\ref{fig:422_example}. In this protocol, an encoded magic state $|\bar{H}\rangle|\bar{H}\rangle$ is first prepared, where $|\bar{H}\rangle$ is the $+1$ eigenstate of the Hadamard operator. Then, this state is verified using two rounds of logical Clifford and stabilizer measurements. We note that one should not use this circuit as is to prepare a logical magic state. One reason is that a weight-$1$ circuit-level error may propagate to a weight-$2$ error.\footnote{If needed, this issue can be dealt with by using a $|\text{CAT}\rangle$ state ancilla $|\text{CAT}\rangle=\frac{1}{\sqrt{2}}(|0\rangle^{\otimes 4}+|1\rangle^{\otimes 4})$~\cite{shor1997faulttolerantquantumcomputation}.} Another reason is that the logical Clifford measurement would not distinguish between $|\bar{H}\rangle|\bar{H}\rangle$ and $|-\bar{H}\rangle|-\bar{H}\rangle$, where $|-\bar{H}\rangle = \bar{Y}|\bar{H}\rangle$. Nonetheless, the purpose of this example is to explicitly demonstrate the error propagation rule in a simplified setup.

\begin{figure*}[!t]
    \centering
    \begin{adjustbox}{height=3cm}
    \begin{quantikz}
        \lstick[4]{$\ket{\bar{H}}\ket{\bar{H}}$} & \gate[style=red!70!white]{X} &\gate{H} & \phase{} & \targ{} & \gate{H} & \phase{} & \targ{} & \\
        && \gate{H} & \phase{} & \targ{} & \gate{H} & \phase{} & \targ{} &\\
        && \gate{H} & \phase{} & \targ{} & \gate{H} & \phase{} & \targ{} &\\
        && \gate{H} & \phase{} & \targ{} & \gate{H} & \phase{} & \targ{} &\\
        \lstick{$\ket{+}$} && \ctrl{-4} & & & & & &\meter{\ket{\pm}} \\
        \lstick{$\ket{+}$} && & \ctrl{-5} & & & & &\meter{\ket{\pm}} \\
        \lstick{$\ket{+}$} && & & \ctrl{-6} & & & &\meter{\ket{\pm}}\\
        \lstick{$\ket{+}$} && & & & \ctrl{-7} & & &\meter{\ket{\pm}} \\
        \lstick{$\ket{+}$} && & & & & \ctrl{-8} & &\meter{\ket{\pm}} \\
        \lstick{$\ket{+}$} && & & & & & \ctrl{-9} &\meter{\ket{\pm}}\\
    \end{quantikz}
    \end{adjustbox}
    \hspace{1cm}
    \begin{adjustbox}{height=3cm}
    \begin{quantikz}
        \lstick[4]{$\ket{\bar{H}}\ket{\bar{H}}$}  &\gate{H} & \phase{} & \targ{} & \gate{H} & \phase{} & \targ{} & \gate[style=red!70!white]{X} & & & & & & \gate[style=red!70!white]{Y} & \gate[style=red!70!white]{Y} & \\
        & \gate{H} & \phase{} & \targ{} & \gate{H} & \phase{} & \targ{} & & & & & & & & & \\
        & \gate{H} & \phase{} & \targ{} & \gate{H} & \phase{} & \targ{} & & & & & & & & & \\
        & \gate{H} & \phase{} & \targ{} & \gate{H} & \phase{} & \targ{} & & & & & & & & & \\
        \lstick{$\ket{+}$} & \ctrl{-4} & & & & & & \gate[style=red!70!white]{S} & \ctrl[style=red!70!white]{1} & \ctrl[style=red!70!white]{2} & \ctrl[style=red!70!white]{3} & \ctrl[style=red!70!white]{4} & \ctrl[style=red!70!white]{5} & \ctrl[style=red!70!white]{-4}  & & \meter{\ket{\pm}} \\
        \lstick{$\ket{+}$} & & \ctrl{-5} & & & & & \gate[style=red!70!white]{Z} & \phase[style=red!70!white]{} & & & & & & & \meter{\ket{\pm}} \\
        \lstick{$\ket{+}$} & & & \ctrl{-6} & & & & & & \phase[style=red!70!white]{} & & & & & & \meter{\ket{\pm}}\\
        \lstick{$\ket{+}$} & & & & \ctrl{-7} & & & \gate[style=red!70!white]{S} & \ctrl[style=red!70!white]{1} & \ctrl[style=red!70!white]{2} & \phase[style=red!70!white]{} & & & & \ctrl[style=red!70!white]{-7} & \meter{\ket{\pm}} \\
        \lstick{$\ket{+}$} & & & & & \ctrl{-8} & & \gate[style=red!70!white]{Z} & \phase[style=red!70!white]{} & & & \phase[style=red!70!white]{} & & & & \meter{\ket{\pm}} \\
        \lstick{$\ket{+}$} & & & & & & \ctrl{-9} & & & \phase[style=red!70!white]{} & & & \phase[style=red!70!white]{} & & & \meter{\ket{\pm}}\\
    \end{quantikz}
    \end{adjustbox}
    \caption{\textbf{Example of $X$-error propagation in  $\ket{\bar{H}}\ket{\bar{H}}$-state measurement protocol on $[[4,2,2]]$ code.} Left: Injection of a Pauli $X$-error. Right: The $X$-error propagates to a Clifford error at the end of the circuit. To better visualize the propagated error, we compress the otherwise sequentially-applied controlled-gates from Fig.~\ref{fig:422_example}.}
    \label{fig:422_propagation}
\end{figure*}

The simulation of Fig.~\ref{fig:422_example} in the absence of circuit-level noise is trivial. Because the initial state is stabilized by the logical Clifford and the stabilizers, the state obtained at the end of the protocol will be $|\bar{H}\rangle|\bar{H}\rangle$. Thus, the main question is how to simulate the circuit in Fig.~\ref{fig:422_example} in the presence of circuit-level noise without using a state-vector simulator. 

Our central claim is that in the presence of arbitrary circuit-level Pauli errors, the state resulting from this protocol is $\mathcal{C}_m|\bar{H}\rangle|\bar{H}\rangle$, for some stabilizer operation $\mathcal{C}_m$ that can be readily computed from the circuit-level Pauli noise. Here, $\mathcal{C}_m$ may include a Clifford gate, but more generally, a projection onto a $\pm 1$ eigenspace of some Pauli operator. Because stabilizer operations do not increase the stabilizer rank,  the stabilizer rank of the noisy state is, for every possible circuit-level Pauli error, at most the stabilizer rank of the noiseless final state. This is the main observation that will generalize to the protocols in Refs.~\cite{Bombin_2007,bombin2016dimensionaljumpquantumerror,Beverland2021,Daguerre:2024gjd,Daguerre:2025boq,Goto:2016gss,Chamberland:2019ehl,Chamberland:2020axi,Itogawa_2025,Gidney:2024alh,sahay2025foldtransversalsurfacecodecultivation,claes2025cultivatingtstatessurface,Vaknin:2025pbp}.

Before we proceed, let us comment on why we are assuming only Pauli errors during the preparation of the $|\bar{H}\rangle|\bar{H}\rangle$ state, prior to the verification step. We envision preparing the logical state $|\bar{H}\rangle|\bar{H}\rangle$ by applying a Clifford encoding circuit to a noisy physical $|H\rangle|H\rangle$ state. The key point is that one can randomly apply $H$ at the physical level, decohering the state in the eigenbasis of $H$. Then, the only possible error acting on this state is the $Y$ error, which is a Pauli. This error, when propagated through a Clifford, remains as a Pauli error; see Appendix~\ref{sub:initialization} for a generalization of this fact to other physical magic states.

We now explain how circuit-level errors propagate in Fig.~\ref{fig:422_example}. Given we are considering circuits with many non-Clifford gates, the main complication we must deal with is how individual circuit-level Pauli errors propagate through such circuits. If such an error only encounters a \textit{single} non-Clifford gate in $\mathcal{C}^{(3)}$, it simply propagates to a Clifford error, which is manageable. However, if this Pauli error encounters two or more non-Clifford gates, it may propagate to a non-Clifford error, which is harder to deal with.

Surprisingly though, in the circuit shown in Fig.~\ref{fig:422_example}, any circuit-level Pauli error \textit{always} propagates to a Clifford error. Let us see why this is the case. As an example, consider a circuit-level $X$-error injected at one of the data qubits at the beginning, which is then propagated to an error at the end of the circuit [Fig.~\ref{fig:422_propagation}]. The propagated error can be computed by sequentially passing it through the gates in the circuit. To that end, the following circuit identities will be useful:
\begin{align}
\begin{adjustbox}{width=0.8\columnwidth}
    \begin{quantikz}
&  & \ctrl{1} & \\ &\gate{X} & \gate{H} &  \end{quantikz} =
\begin{quantikz}
 & \ctrl{1} &  & \ctrl{1} & \gate{S} & \\ & \gate{H} & \gate{X} & \gate{Y}& &
\end{quantikz},
\end{adjustbox} \label{eq:nonclifford_identity_1}\\
\begin{adjustbox}{width=0.65\columnwidth}
\begin{quantikz}
& \ctrl{2} &  & \\ & & \ctrl{1} &  \\
& \gate{U}&\gate{V} &
\end{quantikz}
=
\begin{quantikz}
&  & \ctrl{2} & \ctrl{1} & \\
&\ctrl{1} &  &  \phase{} &\\
& \gate{V} &\gate{U} & &
\end{quantikz}.
\end{adjustbox} \label{eq:nonclifford_identity_2}
\end{align}
if $U$ and $V$ anticommute, i.e., $\{ U, V\}=0$. Repeatedly applying Eq.~\eqref{eq:nonclifford_identity_1} and~\eqref{eq:nonclifford_identity_2}, 
in addition to using the following identity,
\begin{equation}
\begin{adjustbox}{width=0.6\columnwidth}
\begin{quantikz}
& \gate{S} & \ctrl{1} & \\ & & \gate{U} &  \end{quantikz} = \begin{quantikz}
& \ctrl{1} & \gate{S} & \\ &  \gate{U} &   &\end{quantikz},
\end{adjustbox} 
\label{eq:nonclifford_identity_7}
\end{equation}
one can verify that the error propagates to a Clifford error at the end [Fig.~\ref{fig:422_propagation}]. The previous error propagation identities (and those found later in this section) can be easily proven as follows: Suppose we want to propagate a gate $U$ through $C(V)$. Then, we can evaluate the identity $C(V)UC(V^{\dagger})(|b\rangle \otimes I)$, for computational basis states $|b\rangle\in\{|0\rangle,|1\rangle\}$, to compute the propagated error in a case-by-case manner. In particular, this method will indicate whether controlled-gates appear in the propagated error.

The $X$ errors injected on the data qubits at other times also propagate to a Clifford error for a similar reason. Furthermore, $Z$ errors injected on the data qubits also propagate to a Clifford error, a fact that follows from an analog of Eq.~\eqref{eq:nonclifford_identity_1},
\begin{equation}
    \begin{adjustbox}{width=0.8\columnwidth}
    \begin{quantikz}
&  & \ctrl{1} & \\ &\gate{Z} & \gate{H} &  \end{quantikz} =
\begin{quantikz}
 & \ctrl{1} &  & \ctrl{1} & \gate{S^{\dagger}} & \\ & \gate{H} & \gate{Z} & \gate{Y}& &
\end{quantikz}.
\end{adjustbox} 
\label{eq:nonclifford_identity_3}
\end{equation}
Because $X$ and $Z$ generate $Y$, any circuit-level Pauli error injected on the data qubit propagates to a Clifford error at the end of the circuit.\footnote{For multiple errors, one can simply propagate each error in reverse order. Because each propagated error is a Clifford, so is their composition.}

Now let us consider the errors injected on the syndrome qubits. By construction, $Z$ errors do not propagate. On the other hand, $X$ errors when propagated can introduce $H$ errors

\begin{equation}
\begin{adjustbox}{width=0.6\columnwidth}
\begin{quantikz}
& \gate{X} & \ctrl{1} & \\ & & \gate{H} &  \end{quantikz} = \begin{quantikz}
& \ctrl{1} & \gate{X} & \\ &  \gate{H} &  \gate{H} &\end{quantikz}.
\end{adjustbox} 
\label{eq:nonclifford_identity_6}
\end{equation}

Given the previous circuit identity introduces $H$ errors, we must also analyze how such errors propagate. Note the following: 
\begin{equation}
     \begin{adjustbox}{width=0.75\columnwidth}
    \begin{quantikz}
&  & \ctrl{1} & \\ &\gate{H} & \phase{} &  \end{quantikz} =
\begin{quantikz}
 & \ctrl{1} &  & \ctrl{1} & \gate{S} & \\ & \phase{} & \gate{H} & \gate{Y}& &
\end{quantikz},
\end{adjustbox} 
\label{eq:nonclifford_identity_4}
\end{equation}
\begin{equation}
     \begin{adjustbox}{width=0.75\columnwidth}
    \begin{quantikz}
&  & \ctrl{1} & \\ &\gate{H} & \targ{} &  \end{quantikz} =
\begin{quantikz}
 & \ctrl{1} &  & \ctrl{1} & \gate{S^\dag} & \\ & \targ{} & \gate{H} & \gate{Y}& &
\end{quantikz}.
\end{adjustbox} 
\label{eq:nonclifford_identity_5}
\end{equation}
Thus, propagating $H$ errors introduce additional gates, namely a controlled-$Y$ and phase gates. The phase gate remains in the ancilla register without affecting data qubits.  Remarkably, the controlled-$Y$ can be again propagated to the end as a Clifford error, using Eq.~\eqref{eq:nonclifford_identity_2}. Therefore, any circuit-level Pauli error injected on the syndrome qubits also propagates to a Clifford error at the end of the circuit. To summarize, \emph{any} circuit-level Pauli error propagates to a Clifford error at the end of the circuit.

Thus, prior to measurement, we have a state of the following form:
\begin{equation}
    C_\mathrm{prop}(E)(|\bar{H}\rangle|\bar{H}\rangle |+\rangle^{\otimes 6}),
\end{equation}
where $C_\mathrm{prop}(E)$ is a Clifford that depends on the error configuration $E$. Then, we can proceed by simulating each of the $6$ measurements. To that end, note that $|\bar{H}\rangle|\bar{H}\rangle$ is a linear combination of at most $q=4$ stabilizer states (as in Eq.(\ref{eq:stab_rak})):
\begin{equation}
    |\bar{H}\rangle|\bar{H}\rangle = \sum_{i=1}^4 \alpha_i |\bar{s}_i\rangle,
\end{equation}
where $\{ |\bar{s}_i\rangle: i=1, \ldots, 4\}$ is a set of stabilizer states and $\alpha_i$ complex coefficients. The protocol can be simulated by computing the probability of measuring $\pm 1$ for each individual measurement and calculating the postmeasurement state. This state can be obtained from a postselection of the stabilizer states in some fixed basis (say $|0\rangle, |1\rangle$), all of which can be done using the standard extended stabilizer simulation method~\cite{Bravyi:2018ugg}. Finally, one can apply a Pauli (or even Clifford, if necessary) correction based on these measurement outcomes and compute the overlap with the state $|\bar{H}\rangle|\bar{H}\rangle$. Crucially, the overall cost of the simulation depends only on the stabilizer rank of the final state, not the number of non-Clifford gates used in the protocol. We review this stabilizer rank-based simulation approach in Appendix~\ref{app:stab_rank_sim}.

One can also express the ideal logical magic state and its noisy counterparts as density matrices 
\begin{equation}
    \bar{\rho} = \left(|\bar{H}{\rangle \langle}\bar{H}|\right)^{\otimes 2}\:,
\end{equation}
and
\begin{equation}
    \hat{\rho}_m = \mathcal{C}_m(E)\left[\left(|\bar{H}{\rangle\langle}\bar{H}|\right)^{\otimes 2}\right]\:,
\end{equation}
respectively, where $\mathcal{C}_m(E)$ is a stabilizer operation that depends on the error $E$ and the measurement outcome $m$. After applying a decoding operation $\mathcal{D}$ that maps $\hat{\rho}_m$ to a state in the code subspace $\bar{\rho}_N=\mathcal{D}(\hat{\rho}_m )$ (up to a potential logical error), we can compute the fidelity $F = \text{Tr}(\bar{\rho}\bar{\rho}_N)$ in the following way. Note the following decomposition (as in~Eq.(\ref{eq:rep_rho_Paulis})): 
\begin{equation}
    \left(|\bar{H}{\rangle \langle}\bar{H}|\right)^{\otimes 2} = \frac{1}{4}\bar{I}^{\otimes 2}+ \frac{1}{4}\sum_{i = 1}^8 \beta_i \bar{P}_i\:,
\end{equation}
where $\bar{P}_i$ are nontrivial 2-qubit Pauli operators and $\beta_i$ are $p=8$ real non-negative coefficients. Thus, the fidelity can be computed from expectation values of Pauli logical operators since $\langle \bar{P}_i\rangle_{\bar{\rho}_N}=\text{Tr}(\bar{P}_i \bar{\rho}_N)$. 
In Sec.~\ref{sec:fid_method}, we provide further detail about how these expectation values can be estimated by running $8\times 2^2$ stabilizer circuits for each $\bar{P}_i$ (and more generally, $p\times 2^k$ circuits for $k$ the number of logical qubits and $p$ the Pauli rank).

This example may raise more questions than answers as the way in which a circuit-level Pauli error propagates to an end-of-circuit Clifford error may seem mysterious. A natural question is whether there is a more general theory that underlies this observation. This is the topic of Sec.~\ref{sec:controlled_clifford}.

\clearpage

\section{Structure of Controlled-Clifford Unitaries}
\label{sec:controlled_clifford}

In this section, we provide a more general account for the circuit identities used in Sec.~\ref{sec:example}. This analysis will provide a more systematic understanding for why circuit-level Pauli noise propagates to a Clifford error in MSP protocols.

The following are a few key circuit identities that drive our discussion
\begin{equation}
\begin{adjustbox}{height=0.5cm}
    \begin{quantikz}
& &  & \ctrl{1} & \\ & \qwbundle{n} &\gate{V} & \gate{U} &  &\end{quantikz}
\end{adjustbox}
=
\begin{adjustbox}{height=0.5cm}
\begin{quantikz}
 & & \ctrl{1} & & \ctrl{1} & \\ 
 & \qwbundle{n} & \gate{U}  & \gate{V} & \gate{UVU^{\dagger}V^{\dagger}}& 
\end{quantikz},
\end{adjustbox} 
\label{eq:identity_general_1}
\end{equation}

\begin{equation}
\begin{adjustbox}{height=0.6cm}
\begin{quantikz}
    & \gate{X} & \ctrl{1} & \\
 &\qwbundle{n} & \gate{U} &
\end{quantikz}
\end{adjustbox}
=
\begin{adjustbox}{height=0.6cm}
\begin{quantikz}
   & & \ctrl{1} & \ctrl{1} & \gate{X} & \\
   &\qwbundle{n} & \gate{U} & \gate{U^{\dagger}{}^2} & \gate{U} &
\end{quantikz}
\end{adjustbox},
\label{eq:identity_general_2}
\end{equation}
where $U$ and $V$ are arbitrary $n$-qubit unitaries.

We will also use Eq.~\eqref{eq:nonclifford_identity_2}, noting that it applies to any $n$-qubit unitaries $U$ and $V$ insofar as $\{U,V\}=0$. (If $[U,V]=0$, the two controlled unitaries commute.) Equations~\eqref{eq:identity_general_1} and~\eqref{eq:identity_general_2} aim to model propagation of errors in the circuit. We envision the unitary $V$ in Eq.~\eqref{eq:identity_general_1} being a Pauli or a Clifford for the following reasons. A Pauli error may occur at that exact location or may have been propagated from a previous step. A Clifford error may emerge from a Pauli error if propagated through a non-Clifford gate from a previous step. We only consider Pauli $X$ errors on the control [Eq.~\eqref{eq:identity_general_2}] as Pauli $Z$ errors on the control commute past $C(U)$. Thus, the propagation of Pauli $Y$ errors can be inferred from the propagation behavior of Pauli $X$ errors. One may wonder if there is any need to consider a Clifford error on the control qubit. We will see that this case never occurs (up to innocuous diagonal gates) insofar as the circuit-level noise model is Pauli; see Appendix~\ref{app:error_model} for more details.

For reasons we make clear in Sec.~\ref{subsec:pauli-root_clifford}, we will choose $U$ to be a Clifford which is a \emph{square root} of a Pauli [Definition~\ref{def:square_root_pauli}]. 

\subsection{Cliffords that square to Paulis (PSC)}
\label{subsec:pauli-root_clifford}
We study a special class of Cliffords that square to a Pauli and elucidate its properties. The main results of this section concern complete characterization of these Cliffords from two different perspectives: controlled-unitary [Lemma~\ref{lemma:psc_hierarchy}] and its canonical form [Proposition~\ref{prop:psc_normal_form}]. We begin by introducing a formal definition.
\begin{definition}
A unitary is Pauli-Square-Root Clifford (PSC) if it is a (non-Pauli) Clifford that squares to a Pauli.
    \label{def:square_root_pauli}
\end{definition}

PSCs are special in that they satisfy a number of extra properties that are not universally true for all Cliffords. For instance, controlled versions of Cliffords are not necessarily in the third level of the Clifford hierarchy, with the Clifford $HS$ being one such counterexample~\cite{anderson2025controlledgatescliffordhierarchy}. However, PSCs are precisely those unitaries whose controlled version lie in the third level of the Clifford hierarchy.
\begin{restatable}{lemma}{LemmaPSCHierarchy}
    Let $U$ be a unitary. $U$ is a PSC if and only if 
    controlled-$U$ is in the third level of the Clifford hierarchy $\mathcal{C}^{(3)}$.
    \label{lemma:psc_hierarchy}
\end{restatable}

Moreover, there is a canonical form for any PSC, which will turn out to be convenient for proving the circuit identities in Sec.~\ref{subsec:psc_circuit_identities}. 
\begin{restatable}{proposition}{PropositionPSCNormalForm}
    A unitary $U$ is a PSC if and only if it can be written in the following form:
    \begin{equation}
     U = \alpha P \exp\left(\frac{\pi i}{4}\sum_{j=1}^mQ_j \right),\label{eq:standard_form}
    \end{equation}
    where $P$ is a Pauli and $\{Q_j: j=1, \ldots, m\}$ is a commuting set of Paulis $[Q_i,Q_j]=0$, and $\alpha$ is a complex number such that $\alpha^8=1$.
    \label{prop:psc_normal_form}
\end{restatable}

It is possible to check these statements individually for well-known PSCs, such as $H, CZ, X\otimes CZ, TST^{\dagger}$. However, the advantage of these statements is that those properties can be verified by simply checking if a given Clifford squares to a Pauli, something that can be done much more straightforwardly in many cases.

The proofs of Lemma~\ref{lemma:psc_hierarchy} and Proposition~\ref{prop:psc_normal_form} can be found in Appendix~\ref{app:proof_sec4}.

\subsection{Circuit identities}
\label{subsec:psc_circuit_identities}
In this section, we derive circuit identities involving PSCs that generalize the circuit identities discussed in Sec.~\ref{sec:example}. These follow straightforwardly from the canonical form [Eq.~\eqref{eq:standard_form}]. For details, see Appendix~\ref{app:proof_sec4}.

\begin{restatable}{proposition}{PropositionPSCPropagationCaseOne}
    Let $C$ and $P$ be a PSC and a Pauli respectively, supported on $t$ qubits
        \begin{equation}
            \begin{adjustbox}{height=0.65cm}
                \begin{quantikz}
            &  & & \ctrl{1} & \\ 
            & \qwbundle{t} & \gate{P} & \gate{C} &\end{quantikz} =
            \begin{quantikz}
             & & \ctrl{1} &  & \ctrl{1} & \\ 
             & \qwbundle{t} & \gate{C} & \gate{P} & \gate{Q}& 
            \end{quantikz},
            \end{adjustbox} \label{eq:propagation_psc}
        \end{equation}
        where $Q$ is a $t$-qubit Pauli that either commutes or anticommutes with $C$. Furthermore, $\text{Supp}(Q) \subseteq \text{Supp}(C)$.
    \label{prop:psc_propagation_case1}
\end{restatable}

\begin{restatable}{proposition}{PropositionCliffordProp}
    Let $C$ and $P$ be a PSC and a Pauli respectively, supported on $t$ qubits
    \begin{equation}
        \begin{adjustbox}{height=0.65cm}
    \begin{quantikz}
    &  & & \ctrl{1} & \\ 
    &\qwbundle{t} & \gate{C} & \gate{P} &\end{quantikz} =
    \begin{quantikz}
     & & \ctrl{1} &  & \ctrl{1} & \\ 
     & \qwbundle{t} & \gate{P} & \gate{C} & \gate{Q}& 
    \end{quantikz},
    \end{adjustbox} 
    \label{eq:propagation_psc_2}
    \end{equation}
    where $Q$ is a $t$-qubit Pauli that either commutes or anticommutes with $C$. Furthermore, $\text{Supp}(Q) \subseteq \text{Supp}(C)$.
\label{prop:Cliff_prop}
\end{restatable}

The proof of Proposition~\ref{prop:psc_propagation_case1} uses the circuit identity from Eq.~\eqref{eq:propagation_psc}, in addition to the canonical form of PSCs [Proposition~\ref{prop:psc_normal_form}]. The details of the proof can be found in Appendix~\ref{app:proof_sec4}. Proposition~\ref{prop:Cliff_prop} is proved using a similar argument.

\begin{restatable}{theorem}{TheoremPropagationPSC}
    Let $C$ be a PSC on $t$ qubits and $P$ be a $(t+1)$-qubit Pauli
    \begin{equation}
        \begin{adjustbox}{height=0.7cm}
    \begin{quantikz}
        & & \gate[2]{P}  & \ctrl{1} & \\ 
        & \qwbundle{t} & & \gate{C} & \end{quantikz} =
        \begin{quantikz}
         & & \ctrl{1} &  \gate[2]{P'} & \ctrl{1} & & \\
         & \qwbundle{t} & \gate{C} &  & \gate{Q}& \gate{C^b} &
        \end{quantikz},
        \end{adjustbox} 
        \label{eq:propagation_psc_full}
    \end{equation}
    for some $(t+1)$-qubit Pauli $P'$,  $t$-qubit Pauli $Q$, and $b\in \{0,1 \}$, where $Q$ either commutes or anticommutes with $C$. Furthermore, $\text{Supp}(Q) \subseteq \text{Supp}(C)$. 
    \label{thm:propagation_psc}
\end{restatable}
\begin{proof}
    If $P$ does not contain an $X$ or $Y$ on the control qubit, the claim follows immediately from Proposition~\ref{prop:psc_propagation_case1}. If the Pauli on the control qubit is $X$, using Eqs.~\eqref{eq:identity_general_2} and \eqref{eq:propagation_psc}, we obtain:

\begin{equation}
     \begin{adjustbox}{height=0.65cm}
    \begin{quantikz}
& & \gate{X}  & \ctrl{1} & \\ 
& \qwbundle{t} & \gate{P} & \gate{C} &\end{quantikz} =
\begin{quantikz}
 & & \ctrl{1} & \gate{X}  & \ctrl{1} & & \\ 
 & \qwbundle{t} & \gate{C} & \gate{P Q'^{\dagger}C^2} & \gate{Q'^{\dagger} C^{2}}& \gate{C} &
\end{quantikz},
\end{adjustbox} 
    \end{equation}
where $Q'=CPC^{\dagger}P^{\dagger}$ is a $t$-qubit Pauli that either commutes or anticommutes with $C$ [Proposition~\ref{prop:psc_propagation_case1}]. In particular, $[C^{2}, Q^{'\dagger}]=0$. Therefore, $C$ and $Q=Q'^{\dagger} C^{2}$ either commutes or anticommutes. Moreover, since $C$ is a PSC, $Q'^{\dagger} C^{2}$ and $P Q'^{\dagger}C^2$ are Paulis, concluding the proof for this case. This argument also applies to the case in which the Pauli on the control qubit is $Y$.

\end{proof}

Theorem~\ref{thm:propagation_psc} is important because it establishes a nontrivial structural property of a Pauli error propagating through a controlled-Clifford. Since the controlled-Clifford is in the third level of the Clifford hierarchy [Lemma~\ref{lemma:psc_hierarchy}], the error may propagate to a Clifford error. However, the non-Pauli part of the error consists of two components that have nontrivial algebraic relations with the original controlled-Clifford. 

For instance, because $Q$ and $C$ in Eq.~\eqref{eq:propagation_psc_full} either commute or anticommute, we have
\begin{equation}
    \begin{adjustbox}{height=0.5cm}
\begin{quantikz}
& & \ctrl{1} & \ctrl{1}  & \\ 
& \qwbundle{t} & \gate{Q}&\gate{C} &
\end{quantikz}
\end{adjustbox}
=
\begin{cases}
 \begin{adjustbox}{height=0.65cm}
\begin{quantikz}
& & \ctrl{1} & \ctrl{1}  &  \gate{Z} &\\ 
& \qwbundle{t} & \gate{C}&\gate{Q} & &
\end{quantikz}
\end{adjustbox} \,\, &\text{if} \,\, \{Q,C \}=0, \\[20pt]
\begin{adjustbox}{height=0.5cm}
\begin{quantikz}
& & \ctrl{1} & \ctrl{1}  & \\ 
& \qwbundle{t} & \gate{C}&\gate{Q} &
\end{quantikz}
\end{adjustbox} \,\, &\text{if} \,\, [Q,C]=0.
\end{cases}
\label{eq:same-target-same-control-commutation}
\end{equation}
and 
\begin{equation}
    \begin{adjustbox}{height=0.65cm}
\begin{quantikz}
& & \ctrl{2} & &  & \\ & & & \ctrl{1} &  \\
& \qwbundle{t} & \gate{Q}&\gate{C} &
\end{quantikz}
\end{adjustbox}
=
\begin{cases}
 \begin{adjustbox}{height=0.65cm}
\begin{quantikz}
&  & & \ctrl{2} & \ctrl{1} & \\
& &\ctrl{1} &  &  \phase{} &\\
& \qwbundle{t} & \gate{C} &\gate{Q} & &
\end{quantikz}
\end{adjustbox} \,\, &\text{if} \,\, \{Q,C \}=0, \\[20pt]
\begin{adjustbox}{height=0.65cm}
\begin{quantikz}
& & & \ctrl{2} &  & \\ &
&\ctrl{1} &  &   &\\
& \qwbundle{t} & \gate{C} &\gate{Q} & &
\end{quantikz}
\end{adjustbox} \,\, &\text{if} \,\, [Q,C]=0.
\end{cases}
\label{eq:same-target-diff-control-commutation}
\end{equation}
Therefore, if the controlled-$Q$ encounters another controlled-$C$ gate later, they commute up to a potential ${CZ}$ gate on the controls. Similarly, because $Q$ is a Pauli, it will also commute with any other controlled-Pauli up to a potential ${CZ}$ gate on the controls. Moreover, if we were to propagate $C$ in Eq.~\eqref{eq:propagation_psc_full} through a controlled-Pauli, we only get an additional controlled-Pauli whose underlying Pauli either commutes or anticommutes with $C$ [Proposition~\ref{prop:Cliff_prop}]. 

\section{Error propagation in magic state preparation protocols}
\label{sec:canonical_family}

In this section, we summarize families of MSP protocols for which arbitrary circuit-level Pauli errors propagate through them to an end-of-circuit Clifford error. We consider protocols based on measuring logical PSCs [Sec.~\ref{subsec:logical-clifford-meas}], magic state distillation (MSD) at the logical level [Sec.~\ref{subsec:msd_non_Cliff}] and protocols based on the implementation of a transversal non-Clifford gate [Sec.~\ref{subsec:transv-non-Cliff}]. Before diving into any further details, we would like to highlight that error propagation for protocols from Sec.~\ref{subsec:logical-clifford-meas} is in general more involved than for protocols from Secs.~\ref{subsec:msd_non_Cliff} and~\ref{subsec:transv-non-Cliff}. This is due to the complex error propagation properties of logical PSC measurements [Sect.~\ref{sec:controlled_clifford}]. We also note that, compared to Sec.~\ref{subsec:logical-clifford-meas}, the discussions in Secs~\ref{subsec:msd_non_Cliff} and~\ref{subsec:transv-non-Cliff} do not rely heavily on the properties of PSCs. Nonetheless, we discuss these protocols for completeness.

We note that all protocols discussed in this section can be equivalently recast with measurements postponed to the end. For convenience, we stick to this convention throughout this section. One \textit{can} simulate midcircuit measurements and ancilla resets explicitly without having to postpone them to the end. However, in this case, the propagated error also includes projections by stabilizer states, which are more cumbersome (though still efficient) to simulate. For instance, the MSP protocol of Ref.~\cite{Davydova:2025ylx} requires midcircuit measurements and feed-forward Pauli corrections to guarantee fault-tolerance. We leave analysis of error propagation through mid-circuit measurements and ancilla resets to future work.

\subsection{Canonical family 1: logical PSC measurement}
\label{subsec:logical-clifford-meas}

\begin{figure*}[!t]
    \centering
    \includegraphics[width=0.75\textwidth]{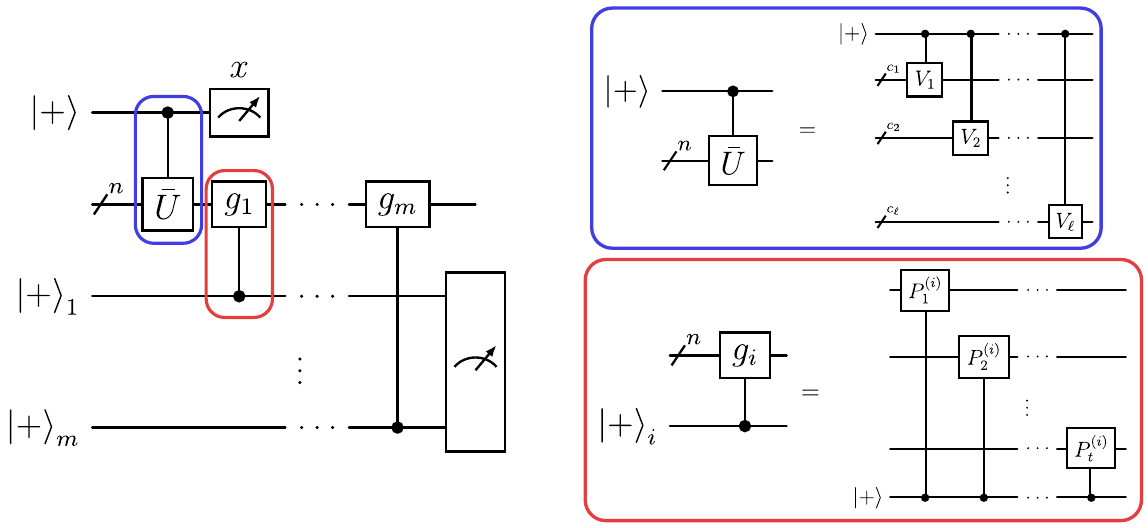}
    \caption{\textbf{Non-FT standard measurement gadget.} Measurement of a transversal logical Pauli-square-root Clifford $\bar{U}=\bigotimes_{i=1}^\ell V_i$, for Pauli-square-root Cliffords $V_i$, is followed by single-ancilla measurements of Pauli stabilizer generators $\{g_1,...,g_{m}\}$, where $g_i=\bigotimes_{j=1}^{t}P_j^{(i)}$ for single-qubit Paulis $P_j^{(i)}$.}
    \label{fig:non-ft-can-meas}
\end{figure*}

We describe a general framework for a family of protocols preparing logical magic states that are $+1$-eigenstates of logical PSC operators. Such protocols interleave layers of logical PSC measurements with stabilizer measurements~\cite{Goto:2016gss,Chamberland:2019ehl,Chamberland:2020axi,Itogawa_2025,Gidney:2024alh,sahay2025foldtransversalsurfacecodecultivation,claes2025cultivatingtstatessurface,Vaknin:2025pbp}. We show that the error propagation observed from the toy example of Sec.~\ref{sec:example} can be generalized to these cases.

As a warm-up, we first consider a non-fault-tolerant (non-FT) version 
of PSC measurement protocols. We then discuss a FT version using $|\text{CAT}\rangle$ states 
(where $|\text{CAT}\rangle=\frac{1}{\sqrt{2}}(|0\rangle^{\otimes s}+|1\rangle^{\otimes s})$ for some natural number $s$). Additionally, we also briefly discuss the use of flag qubits \cite{Chao_2018a,Chao_2018b,Chamberland_2018,Chamberland:2019ehl,Chamberland:2020axi} for achieving fault-tolerance in lieu of $|\text{CAT}\rangle$ states to improve ancilla overhead.

Our discussion will center around logical PSC gates, defined as follows.

\begin{definition}[Tranversal PSC gate] A logical unitary gate $\bar{U}$ is a transversal PSC if it can be written as a tensor product of PSC gates $V_i$ with disjoint support,
\begin{equation}
    \bar{U}=\bigotimes_{i=1}^\ell V_i\:.
\end{equation}
\label{def:transversal-psc}
\end{definition}

Note that we are using a liberal definition of transversality in that $V_i$ may act on more than one qubit. We make this choice in order to include logical Cliffords that are tensor products of multi-qubit PSCs. For example, in the folded surface code~\cite{Moussa:2016kgp}, the full set of logical Cliffords can be realized by a transversal PSC gate, wherein some $V_i$ act on two qubits instead of just one.

In the rest of this section, we focus on protocols that measure a transversal PSC $\bar{U}= \bigotimes_{i=1}^\ell V_i$ with the following additional properties:\footnote{Note that eigenstates of PSCs are not automatically magic states. For instance, the eigenstates of the $S$ gate are $\ket{0}$ and $\ket{1}$.  Therefore, checking whether a measurement of a logical PSC yields a magic state must be done on a case-by-case basis.}
\begin{enumerate}
    \item Every $V_i$ acts on at most $O(1)$ qubits.
    \item The PSC $\bar{U}$ has order $2$, i.e., $\bar{U}^2=I$.
\end{enumerate}
Let us briefly justify these choices. Without the first property, $\bar{U}$ may spread low-weight errors to undesirable high-weight errors. Indeed, to the best of our knowledge, all known protocols that measure a logical Clifford obey this assumption. The second assumption is imposed due to the fact that many existing protocols for preparing logical magic states measure PSCs of order $2$.\footnote{However, exceptions do exist, e.g., Ref.~\cite{claes2025cultivatingtstatessurface}. The Clifford being measured in this work is not a PSC.} We also note that generalizing to the case where $\bar{U}$ is of order $4$ is straightforward. Given $\bar{U}$ is a PSC, its square is a logical Pauli. The eigenphase of $\bar{U}$ can be determined by performing quantum phase estimation with $\bar{U}$, which involves applying both controlled-$\bar{U}$ and controlled-$\bar{U}^2$. The propagation of errors we discuss generalizes to this setting in a straightforward way. We also note that cases where $\bar{U}^4 = -I$ (due to a global phase) are redundant as this does not change $U$'s eigenstate. 

Now we introduce various protocols for measuring logical PSCs and prove that circuit-level Pauli errors propagate to an end-of-circuit Clifford error [Theorem~\ref{thm:standard-protocol-error-prop}] with polynomially-many one- and two-qubit gates. To provide some context, recall the MSP protocol discussed in Sec.~\ref{sec:example}. Here, a logical PSC and stabilizers of the $[[4,2,2]]$ code were measured twice (non-fault-tolerantly) in an alternating fashion. We envision generalizing this idea to protocols that repeatedly measure an arbitrary subset of a code's stabilizers and an available logical PSC, in an arbitrary order.

\begin{figure*}[!t]
    \centering
    \includegraphics[width=0.75\textwidth]{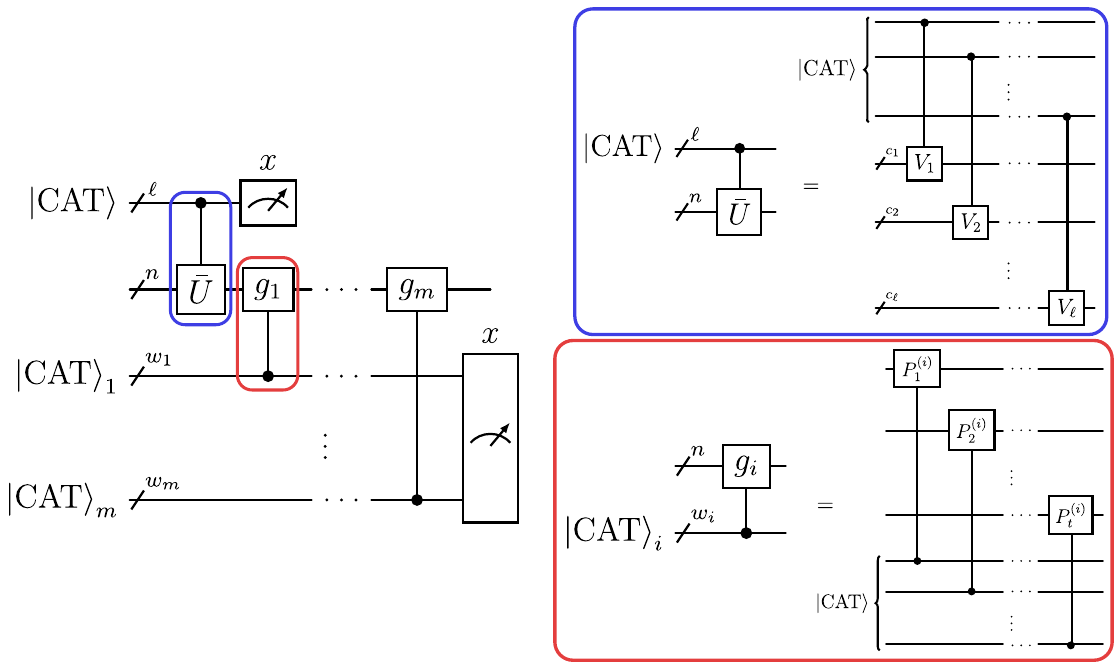}
    \caption{\textbf{Shor-style measurement gadget}. Measurement of a transversal logical Pauli-square-root Clifford $\bar{U}=\bigotimes_{i=1}^l V_i$ for Pauli-square-root Cliffords $V_i$ controlled by a weight-$\ell$ $|\text{CAT}\rangle$ state 
    is followed by Shor-style measurements of Pauli stabilizer generators $\{g_1,...,g_{m}\}$ where $g_i=\bigotimes_{j=1}^{t}P_j^{(i)}$ for single-qubit Paulis $P_j^{(i)}$ using weight-$w_i$ $|\text{CAT}\rangle$ states.}
    \label{fig:FT-canonical-meas-circ}
\end{figure*}

More formally, consider measuring a subset of stabilizers $S_M = \{ g_1, \ldots, g_m\}\subseteq S$, which can be specified in terms of a parity-check matrix $H$. We define sparsity parameters $w_c$ and $w_q$ associated with $H$ as the maximum row and column weights of this matrix, respectively. In particular, these are the maximum weight of the elements of $S_M$ and the maximum number of elements in $S_M$ acting nontrivially on a single-qubit, respectively. 

Given a logical PSC $\bar{U}$, we consider three types of protocols for measuring $\bar{U}$:
\begin{enumerate}
    \item \textbf{Standard protocol.} A protocol is \textit{standard} if every stabilizer/PSC is measured in a non-fault-tolerant manner with single-ancilla qubits; see Fig.~\ref{fig:non-ft-can-meas}.
    \item \textbf{Shor-style protocol.} A protocol is \textit{Shor-style} if every stabilizer/PSC is measured fault-tolerantly with $|\text{CAT}\rangle$ states; see Fig.~\ref{fig:FT-canonical-meas-circ}.
    \item \textbf{Flag-based.} A protocol is \textit{flag-based} if the stabilizer/PSC measurements involve flag qubits; see Fig.~\ref{fig:CH-meas-circ} for an example measuring $\bar{H}$ on the Steane code.
\end{enumerate}

For each type of protocol, we show that circuit-level Pauli errors propagate to Clifford errors with sufficiently-small gate descriptions, in effect enabling efficient simulation (since it is difficult to systematically cover all flag-based protocols, we only provide a general remark and an example.) At a high-level, we consider propagating a circuit-level error originating at a data or ancilla qubit and give a coarse upper-bound on the one- and two-qubit gate count of the resulting Clifford error. In particular, we show this upper-bound to be polynomial in the number of rounds of stabilizer and logical PSC measurements, the sparsity parameters $w_q$ and $w_c$ of parity-check matrix $H$, and the number of PSC tensor factors $V_i$ in $\bar{U}$ [Definition~\ref{def:transversal-psc}].

\begin{figure*}[t]
    \centering
    \includegraphics[width=0.9\linewidth]{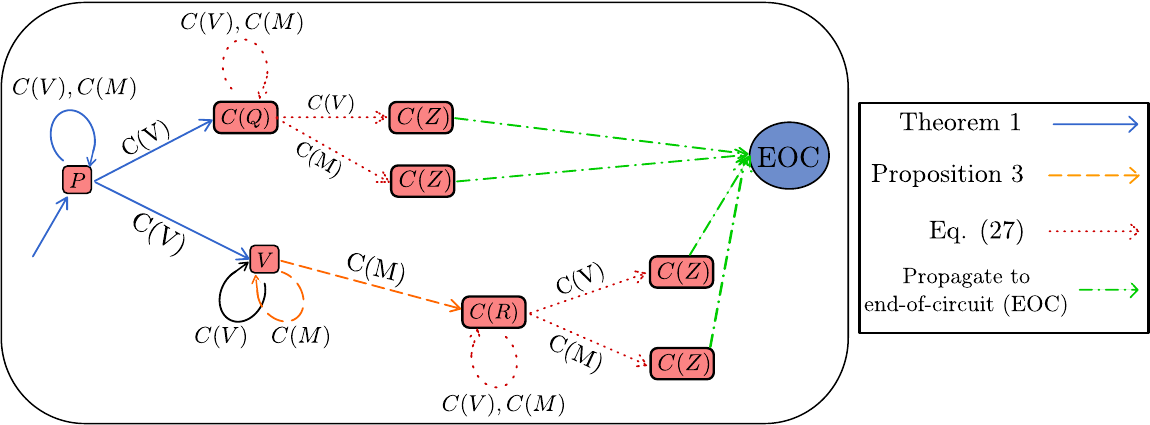}
    \caption{\textbf{Error propagation rules for canonical Pauli-square-root Clifford measurement protocols}. Here, $P$ is a Pauli and $V$ is a Pauli-square-root Clifford. $C(R)$ and $C(Q)$ are controlled-Paulis where $Q$ and $R$ either commute or anticommute with $V$. This diagram only shows nontrivial propagation of errors. For instance, an error (vertices) and a gate (edges) may have disjoint supports, in which case the error will remain the same. Such propagations are not included in this diagram.}
    \label{fig:propagation_rules_visual}
\end{figure*}

\begin{restatable}{theorem}{StandardProtocolErrorProp}
    Let $w_c$ and $w_q$ be the maximum row and column weight of a parity check matrix $H$. Consider a standard protocol consisting of $r = r_S + r_L$ rounds of stabilizer ($r_S$) and PSC ($r_L$) measurements. Any single-qubit circuit-level error propagates through this protocol to a Clifford error consisting of $O\left(r^2 w_q(w_c + \ell)\right)$ one- and two-qubit Clifford gates, for $\ell$ the number of PSC tensor factors $V_i$ in $\bar{U}$.
    \label{thm:standard-protocol-error-prop}
\end{restatable}

We provide an in-depth proof of Theorem~\ref{thm:standard-protocol-error-prop} in Appendix~\ref{app:proof_theo2}. Given Theorem~\ref{thm:standard-protocol-error-prop}, we can straightforwardly extrapolate to an average-case upper-bound on the number of gates in a Clifford error due to propagating \textit{all} sampled circuit-level errors. In particular, if a protocol circuit has $M$ error locations, the one- and two-qubit gate count of the total propagated Clifford error scales as $O\left(Mr^2w_q(w_c + \ell)\right)$.

\begin{corollary}
    Consider a standard protocol consisting of $r = r_S + r_L$ rounds of stabilizer ($r_S$) and PSC ($r_L$) measurements, with $M$ error locations. Any error configuration on these $M$ error locations propagates through this protocol to an end-of-circuit Clifford error consisting of $O\left(Mr^2w_q(w_c + \ell)\right)$ one- and two-qubit Clifford gates.
    \label{corr:standard-protocol-full-error-prop}
\end{corollary}

Additionally, the result from Theorem~\ref{thm:standard-protocol-error-prop} naturally upper-bounds the size of a Clifford error propagated through a Shor-style protocol. For instance, a circuit-level error on a stabilizer measurement ancilla qubit will spread to at most a \textit{single} data qubit rather than $w_c$ qubits due to transversality of the stabilizer measurement. Likewise, an error on a PSC measurement ancilla qubit would also only spread to at most \textit{single} data qubit instead of $\ell$ disjoint subsystems. As a result, the number of gates in a Clifford error propagated through a Shor-style protocol, originating from an ancilla qubit, would no longer have a dependence on $w_c$ or $\ell$. Hence we obtain the following result:

\begin{corollary}
    Consider a Shor-style protocol \footnote{Corollary 2 holds similarly for Steane and Knill-type error correction.} consisting of $r = r_S + r_L$ rounds of stabilizer ($r_S$) and PSC ($r_L$) measurements, with $M'$ error locations. Any error configuration on these $M'$ error locations propagates through this protocol to a Clifford error consisting of  $O\left(M'r^2w_q\right)$ one- and two-qubit Clifford gates.
    \label{corr:shor-style-protocol-full-error-prop}
\end{corollary}

One practical downside of using $|\text{CAT}\rangle$ states is its sizable ancilla overhead. Additionally, each $|\text{CAT}\rangle$ state has to be prepared and verified prior to performing measurements to maintain fault-tolerance \cite{shor1997faulttolerantquantumcomputation}. If postselection is being employed, this would mean incurring a much lower acceptance rate to prepare high-fidelity magic states.

As such in many practical scenarios, it is helpful to use flag qubits~\cite{Yoder_2017,Chao_2018a,Chao_2018b,Chamberland_2018,Chamberland:2020axi}. Flag qubit-based syndrome measurement schemes are engineered to detect when low-weight (correctable) faults propagate to faults with weight larger than $\lfloor (d-1)/2 \rfloor$, for code distance $d$. More precisely, a circuit is called a \textit{$t$-flag circuit} \cite{Chamberland_2018}, for $t = \lfloor (d-1)/2 \rfloor$, if all $v \leq t$ circuit faults that propagate to an error weight $> v$ are detected. Hence, such a $t$-flag circuit would require \textit{at most} $t$ ancilla qubits to detect an adversarial spread of errors.

Once again, just like the standard and Shor-style FT protocols, circuit-level Pauli noise propagates to end-of-circuit Clifford errors with sufficiently-few one- and two-qubit gates. This is because flag-based methods only introduce additional controlled-Pauli gates and at most an extra $t$ ancilla qubits per PSC measurement and stabilizer check. Indeed, we demonstrate a concrete example of such a flag-based PSC measurement protocol in Sec.~\ref{sub:sim-application-example}.

In conjunction with the above results, we note that rules for error propagation can be conveniently packaged into a directed graph structure. Recall that logical PSC measurement protocols are comprised entirely of controlled-PSC gates $C(V)$ and controlled-Pauli gates $C(M)$. To this end, Fig.~\ref{fig:propagation_rules_visual} specifies a complete set of recursive error propagation rules for circuit-level Pauli noise under these gates and shows that \textit{only certain} Clifford errors are propagated through to the end of a PSC measurement-based protocol. Thus, the propagated error is necessarily Clifford.

A motivation for Fig.~\ref{fig:propagation_rules_visual} is as follows: starting with Pauli errors (denoted by $P$), propagation through PSC measurements is given by Theorem~\ref{thm:propagation_psc}. Subsequently, propagation of the resulting controlled-Pauli errors through future PSC and stabilizer measurements is given by Eq.~\eqref{eq:same-target-diff-control-commutation}, incurring $C(Z)$ gates between ancilla qubits. Since we consider a circuit model in which ancillas are never reused, these $C(Z)$ errors can trivially be commuted to the end of the circuit. Finally, PSC errors (denoted $V$) can be commuted past further stabilizer measurements by Proposition~\ref{prop:Cliff_prop}.

\subsection{Canonical family 2: magic state distillation}
\label{subsec:msd_non_Cliff}

Magic state distillation (MSD) protocols~\cite{Bravyi_2005,meier2012magic,Bravyi_2012,Jones_2013a,fowler_surface_2013,Jones_2013b,duclos_cianci_distillation_2013,Duclos_Cianci_2015,Campbell_2017,O_Gorman_2017,Haah2018codesprotocols} based on stabilizer codes also propagate circuit-level Pauli errors to an end-of-circuit Clifford error. In light of~\cite{Zheng:2024dpx}, all known MSD protocols can be recast using the stabilizer reduction formalism~\cite{Campbell:2009jko}. Indeed, given an $[[n,k,d]]$ stabilizer code, a decoder $U_D$ consisting of Clifford operations acts on a $k$-qubit logical state $|\bar{\phi}\rangle$ as
\begin{equation}
    U_D|\bar{\phi}\rangle=|\phi\rangle \otimes |0\rangle^{\otimes (n-k)}\:.
\end{equation}
An $n$-to-$k$ MSD protocol for such a code can be engineered by non-fault-tolerantly preparing $n$ copies of the magic state as input $\rho_{in}=\rho_i^{\otimes n}$, followed by the application of the decoder $U_D$ and by a final round of $n-k$ measurements in the $Z$ basis. Upon postselecting the $+1$ measurement outcomes, the protocol outputs a distilled $k$-qubit magic state $\rho_{out}=\text{Tr}_{anc}(U_D\tilde{\rho}U_D^{\dagger})$ with higher fidelity, where $\tilde{\rho}\propto \Pi \rho_{in} \Pi$ and $\Pi$ is the projector onto the codespace.

MSD protocols described this way do not apply any non-Clifford gates during the distillation process. Instead, one only applies Clifford gates to non-Pauli stabilizer inputs. Thus, all circuit-level Pauli noise propagates to a Pauli error at the end. The fidelity of this noisy magic state can be estimated using extended stabilizer simulation, as we explain in more detail in Sec.~\ref{sec:fid_method}. Importantly, the simulation cost is polynomial in the stabilizer rank/Pauli rank of the \textit{output} logical magic state. Note that in many protocols of practical interest, e.g., $15$-to-$1$ protocol~\cite{Bravyi_2005}, the stabilizer/Pauli rank of the output is significantly smaller than that of the input. Therefore, this observation may lead to faster simulation of MSD protocols.

\subsection{Canonical family 3: transversal non-Clifford gate}
\label{subsec:transv-non-Cliff}

Another family of protocols that propagate circuit-level Pauli errors to an end-of-circuit Clifford error are those that rely on quantum error correcting (QEC) codes admitting a transversal non-Clifford gate from the third level of the Clifford hierarchy. Examples of such codes are 3D color codes~\cite{Bombin_2007,Bombin:2015tpp,Kubica:2014jue} and 3D surface codes~\cite{Vasmer:2019ovd}, which admit transversal $\bar{T}$ and $\overline{CCZ}$ gates, respectively. In general, to realize universal quantum computation, such codes are not used alone as several no-go theorems restrict the set of available fault-tolerant, logical gates~\cite{Eastin:2009tem,Bravyi:2012rnv}. Instead, one can use a \textit{pair} of codes with a complementary, fault-tolerant, and universal set of logical gates to circumvent these restrictions. Namely, this is done by switching back and forth between them in a technique known as code-switching~\cite{Anderson:2014jvy,bombin2016dimensionaljumpquantumerror,Beverland2021,Butt_2024,Daguerre:2024gjd,Daguerre:2025boq}. For instance, MSP protocols can use 3D color codes to prepare a magic $|\bar{T}\rangle=\bar{T}|\bar{+}\rangle$ state fault-tolerantly, which can then be teleported to a 2D color code~\cite{Daguerre:2024gjd} using a transversal entangling gate between both codes~\cite{Sullivan:2023qsg,Heussen:2024vtv} [Fig.~\ref{fig:code_switching}].

\begin{figure}[!h]
    \begin{adjustbox}{scale=0.9}
    \begin{quantikz}[column sep = 0.4cm, font=\large, wire types = {q,q}]
            \lstick{3D code: $\:\:|\bar{+}\rangle$}  & & \gate{\bar{T}} & \ctrl{1} & &\meter[1] {\bar{X}} \wire[d][1]{c}\\
            \lstick{2D code: $\:\:|\bar{0}\rangle$} &  &  & \targ{} &&  \gate{\bar{Z}}&\rstick{$|\bar{T}\rangle$} 
    \end{quantikz}\:.
    \end{adjustbox}
\caption{\textbf{Code-switching protocol between 3D and 2D color codes to prepare logical $|\bar{T}\rangle$ states~\cite{Daguerre:2024gjd}}. The logical non-Clifford $\bar{T}$ gate is transversal in the 3D color code.} \label{fig:code_switching} 
\end{figure}
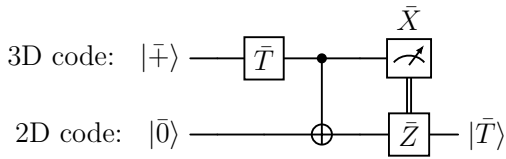

The structure of propagated Clifford errors for the family of MSP protocols described in Fig.~\ref{fig:code_switching} is simple. First, in a circuit-level noise model, the fault-tolerant preparation of the noisy $|\bar{+}\rangle$ state solely produces an initialization Pauli error frame $P_I$ such that $|\bar{+}\rangle=P_I|\bar{+}\rangle_{\text{noiseless}}$, for a noiseless plus state $|\bar{+}\rangle_{\text{noiseless}}$. In binary representation, $P_I \sim \bigotimes_{i=1}^n X_i^{a_i}Z_i^{b_i}$ for $a_i,b_i \in \{0,1\}$ (up to a global phase). Moreover, the transversal $\bar{T}$ gate in 3D color codes is such that
\begin{equation}
    \bar{T}=\left(\bigotimes_{i\in A}T_i\right)\left(\bigotimes_{i\in A^{C}}T_i^{\dagger}\right),
\end{equation}
for a disjoint partition $M=A \cup A^{C}$ of the set of physical qubits $M$ of the 3D color code~\cite{Bombin:2015tpp,Kubica:2014jue}. The Clifford error frame $C_\mathrm{prop}(P_I)$ is obtained upon propagating the Pauli error frame $P_I$ through $\bar{T}$. Indeed,
\begin{equation}
    \bar{T}|\bar{+}\rangle=(\bar{T}P_I\bar{T}^{\dagger})\bar{T}|\bar{+}\rangle_{\text{noiseless}}=C_\mathrm{prop}(P_I)|\bar{T}\rangle_{\text{noiseless}}\:,
    \label{eq:T_gate_3D}
\end{equation}
where
\begin{equation}
    C_\mathrm{prop}(P_I)=\bar{T}P_I\bar{T}^{\dagger}=\left(\bigotimes_{i\in A}S_i^{a_i}\right)\left(\bigotimes_{i\in A^{C}}S_i^{b_i\dagger}\right)P_I\:.\label{eq:CP_transversal}
\end{equation}
Hence, Eq.~\eqref{eq:T_gate_3D} shows that a logical $\bar{T}$ gate acting on a noisy $|\bar{+}\rangle$ state can be effectively described by a noiseless encoded logical magic state $|\bar{T}\rangle_{\text{noiseless}}$ subjected to a propagated Clifford error (\ref{eq:CP_transversal}).

\section{Fidelity estimation}
\label{sec:fid_method}

In the previous section, we showed how circuit-level Pauli noise propagates to a Clifford error on the final state [Fig.~\ref{fig:simulation-steps}(b)] for various MSP protocols. In this section, we give a detailed account on how to use the propagated Clifford error to estimate the fidelity of the output noisy logical magic state.

First, note that once errors are propagated to the end of the circuit, the noiseless circuit containing non-Clifford gates can freely act on the initial logical state to prepare a \textit{noiseless} logical magic state. Hence, the circuit can be effectively described as a noiseless logical magic state followed by the propagated Clifford error [Fig.~\ref{fig:simulation-steps}(c)]. In order to estimate the fidelity of such a noisy magic state using stabilizer simulations, we must circumvent the inability of such simulators to encode a non-Pauli stabilizer state.

One approach is to decompose the magic state as a linear combination of stabilizer states, namely, its stabilizer rank decomposition [Eq.~\eqref{eq:stab_rak}], and evolve each stabilizer state independently. We describe this method, which was originally proposed in Ref.~\cite{Bravyi:2018ugg}, in further detail in Appendix~\ref{app:stab_rank_sim}. This method keeps track of the global phase of stabilizer states, which is used to determine the relative phases needed in computing the overlaps between two such linear combinations of stabilizer states.

A second approach we consider relies on decomposing the magic state density matrix into a linear combination of Pauli matrices, namely its Pauli rank decomposition [Eq.~\eqref{eq:rep_rho_Paulis}] (which can be further decomposed into a linear combination of stabilizer state projectors)~\cite{Bu:2019qed}. This technique allows for simulating MSP protocols with a stabilizer simulator without the need to keep track of global phases. As such, it is readily applicable to commonly-used simulation tools such as \texttt{Stim}~\cite{gidney2021stim}. Although the technique was originally introduced in Ref.~\cite{Daguerre:2024gjd} for the special case of $|T\rangle$ states, here we provide a general framework suitable for arbitrary multi-qubit magic states.

We now describe how to efficiently estimate the fidelity $F = \mathrm{Tr}\left(\bar{\rho}\bar{\rho}_N\right)$ of a noisy $k$-qubit logical state $\bar{\rho}_N$ with respect to a target logical pure state $\bar{\rho}$. The noisy logical state is defined as  
\begin{equation}
    \bar{\rho}_N = \mathcal{E}_L(\bar{\rho}) = (\mathcal{D} \circ \mathcal{E} \circ \mathcal{U})(\bar{\rho})\:,
    \label{eq:noisy-logical-state}
\end{equation}
for an effective logical noise channel  $\mathcal{E}_L$. The channel $\mathcal{E}_L$ is composed of the encoding channel $\mathcal{U}$ which maps a logical $k$-qubit state onto an $n$-qubit state, a circuit-level noise channel $\mathcal{E}$ acting on an $n$-qubit state, and the decoding channel $\mathcal{D}$ that maps the noisy $n$-qubit state back onto a $k$-qubit logical state (up to a potential logical error). 

Since $k$-qubit Pauli operators form a basis for $2^k \times 2^k$ complex-valued matrices, one can re-express the magic state $\bar{\rho}$ as a linear combination of them as in Eq.~(\ref{eq:rep_rho_Paulis}),
\begin{equation}
    \bar{\rho} = \frac{1}{2^k}\bar{I}^{\otimes k} + \frac{1}{2^k}\sum_{i=1}^{p} \beta_i\bar{P}_i\:,
    \label{eq:pauli-rank-decomposition}
\end{equation}
where $\bar{P}_i \in \{\bar{I}, \bar{X}, \bar{Y},\bar{Z}\}^{\otimes k}\setminus\{\bar{I}^{\otimes k}\}$  and $\beta_i=\text{Tr}(\bar{\rho} \bar{P}_i)$. The number of nontrivial Paulis $p$ is known as the \textit{Pauli rank} of $\bar{\rho}$~\cite{Bu:2019qed}. Thus, by linearity of trace, it follows that 
\begin{equation}
    F = \frac{1}{2^k} + \frac{1}{2^k}\sum_{i=1}^{p}\beta_i\langle \bar{P}_i \rangle_{\bar{\rho}_N}\:,
    \label{eq:fid_magic}
\end{equation}
for expected values $\langle \bar{P}_i \rangle_{\bar{\rho}_N} = \mathrm{Tr}(\bar{P}_i \bar{\rho}_N)$. Recall that our goal is to efficiently estimate the fidelity $F$. Additionally, recall that we cannot specify nonstabilizer input states $\bar{\rho}_N$ in a stabilizer simulator. As such, we further decompose $\bar{\rho}_N$ into a linear combination of stabilizer state projectors. In particular, we can insert Eq.~\eqref{eq:pauli-rank-decomposition} into Eq.~\eqref{eq:noisy-logical-state} and then further decompose each logical Pauli $\bar{P}_i$ into such a linear combination. First, we write each $k$-qubit Pauli $\bar{P}_i$ as a tensor product of single-qubit Paulis
\begin{equation}
    \bar{P}_i = \bigotimes_{j=1}^k P_i^{(j)}\:,
\end{equation}
where $P_i^{(j)} \in \{\bar{I},\bar{X},\bar{Y},\bar{Z}\}$. Each single-qubit Pauli $P_i^{(j)}$ is thus either the identity or a non-convex combination of Pauli stabilizer states
\begin{equation}
    \bar{X}=|\bar{+}\rangle \langle \bar{+}|-|\bar{-}\rangle \langle \bar{-}|\:,
\end{equation}
\begin{equation}
    \bar{Y}=|\bar{Y}_+\rangle \langle \bar{Y}_+|-|\bar{Y}_-\rangle \langle \bar{Y}_-|\:,
\end{equation}
\begin{equation}
    \bar{Z}=|\bar{0}\rangle \langle \bar{0}|-|\bar{1}\rangle \langle \bar{1}|\:,
\end{equation}
\begin{equation}
    \bar{I} = |\bar{s}\rangle \langle \bar{s}|+|\bar{s}^{\perp}\rangle \langle \bar{s}^{\perp}|\:, \ \ \ \ \langle \bar{s}|\bar{s}^\perp \rangle = 0\:,
\end{equation}
for $|\bar{s}\rangle$ a Pauli stabilizer state (for instance, $\ket{\bar{s}} = \ket{\bar{+}}$ and $\ket{\bar{s}^\perp} = \ket{\bar{-}}$). Therefore, a nontrivial Pauli $\bar{P}_i$ can be written as a \textit{non-convex} combination of tensor products of Pauli stabilizer states 
\begin{equation}
    \bar{P}_i=\sum_{l=1}^{2^k}(-1)^{\gamma_l} \left(\bigotimes_{j=1}^k |\bar{s}_l^{(j)}\rangle \langle \bar{s}_l^{(j)}|\right)=\sum_{l=1}^{2^k}(-1)^{\gamma_l}\Pi^i_l\:,
    \label{eq:Pauli_comb_stab}
\end{equation}
where $\gamma_l \in \{0,1\}$, $|\bar{s}^{(j)}_l\rangle \langle \bar{s}^{(j)}_l|$ is a single logical-qubit stabilizer state projector acting on the $j$th logical qubit and $\Pi^i_l$ is a $k$-qubit stabilizer state projector, both corresponding to the $l$th term of the decomposition associated to $\bar{P}_i$. Hence, by the linearity of trace and the logical noise channel $\mathcal{E}_L$, 
\begin{equation}
    \langle \bar{P}_i \rangle_{\bar{\rho}_N}=\frac{1}{2^k}\sum_{j=1}^p\sum_{l=1}^{2^k}(-1)^{\gamma_l} \langle \bar{P}_i\rangle_{\mathcal{E}_L(\Pi^j_l)}\:,
\end{equation}
where $\langle \bar{P}_i\rangle_{\mathcal{E}_L(\Pi_l^j)}=\text{Tr}(\mathcal{E}_L(\Pi^j_l)\bar{P}_i)$. Finally, the fidelity (\ref{eq:fid_magic}) can be written as
\begin{equation}
    F = \frac{1}{2^k}+\frac{1}{4^{k}}\sum_{i,j = 1}^{p} \sum_{l=1}^{2^k}\beta_i\beta_j(-1)^{\gamma_l} \langle \bar{P}_i\rangle_{\mathcal{E}_L(\Pi^j_l)}\:.
\label{eq:fid_final}
\end{equation}
Thus, to estimate the fidelity of the noisy logical magic state $\bar{\rho}_N$, it suffices to run $O(p^22^k)$ different logical stabilizer circuits subjected to a propagated Clifford error. This is quadratic in the Pauli rank and exponential in the number of logical qubits. 

As a concrete example, consider the $k=1$ magic state $|\bar{T}\rangle \langle \bar{T}|$ with Pauli rank $p=2$, which can be written as
\begin{equation}
    |\bar{T}\rangle \langle \bar{T}|=\frac{1}{2}\bar{I}+\frac{1}{2\sqrt{2}}(\bar{X}+\bar{Y})\:.
\end{equation}
The fidelity (\ref{eq:fid_final}) has the following expression~\cite{Daguerre:2024gjd}:
\begin{equation}
    F_{|\bar{T}\rangle \langle \bar{T}|}=\frac{1}{2}+\frac{1}{8}\Delta \:,
\end{equation}
where
\begin{equation}
    \begin{split}
    \Delta=&\langle \bar{X}\rangle_{|\widetilde{+}\rangle}-\langle \bar{X}\rangle_{|\widetilde{-}\rangle}+\langle \bar{X}\rangle_{|\widetilde{Y_+}\rangle}-\langle \bar{X}\rangle_{|\widetilde{Y_-}\rangle}\\
    &+\langle \bar{Y}\rangle_{|\widetilde{+}\rangle}-\langle \bar{Y}\rangle_{|\widetilde{-}\rangle}+\langle \bar{Y}\rangle_{|\widetilde{Y_+}\rangle}-\langle \bar{Y}\rangle_{|\widetilde{Y_-}\rangle} \:.
    \end{split}
\end{equation}
In the last equation, the lower index $|\widetilde{s}\rangle$ stands for $\mathcal{E}_L(|\bar{s}\rangle \langle \bar{s}|)$, where $|\bar{s}\rangle \in \{|\bar{+}\rangle,|\bar{-}\rangle,|\bar{Y}_{+}\rangle,|\bar{Y}_{-}\rangle\}$.

\section{Simulation overview}
\label{sec:sim_technique}

Having described the families of MSP protocols that propagate circuit-level Pauli errors to an end-of-circuit Clifford error in Sec.~\ref{sec:canonical_family} and a fidelity estimation method in Sec.~\ref{sec:fid_method}, in this section we provide a holistic overview of our simulation technique (including a summary in Sec.~\ref{subsec:summary_sim_techn}). In Sec.~\ref{sub:sim-application-example}, we present a proof-of-principle numerical simulation of a MSP protocol that measures a transversal Clifford operator. Finally, in Sec.~\ref{subsec:time-space-complexity}, we analyze the space and time complexity of our method.

\subsection{Summary}
\label{subsec:summary_sim_techn}

Below, we provide an executive summary of our simulation technique. \newline

\noindent \textbf{Protocol definition:} MSP protocol, specified as a sequence of one- and two-qubit unitaries $\{U_i\}$ with $U_i \in \mathcal{C}^{(2)}$~or~$\mathcal{C}^{(3)}$, acting on a logical input state $|\bar{\psi}\rangle$ to produce a logical magic output state $|\bar{\phi}\rangle$\footnote{For protocols based on measurements of a logical PSC, the encoded state $|\bar{\psi}\rangle$ is equal to the output logical magic state $|\bar{\phi}\rangle$ stabilized by the PSC. For code-switching protocols described in this article, $|\bar{\psi}\rangle=|\bar{+}\rangle$ which differs from the the output magic state $|\bar{\phi}\rangle$.} such that $\prod_i U_i(\ket{\bar{\psi}} \otimes \ket{a}) = \ket{\bar{\phi}} \otimes \ket{a}$ for an ancilla register $|a\rangle$. Deferred projective measurements in the Pauli basis $\{M\}$ [Fig.~\ref{fig:simulation-steps}].\newline

\noindent \textbf{Simulation input:} MSP protocol. Circuit-level noise model of strength $p$, accounting for initialization $P_I$, gate-specific $\{P_i\}$ and measurement $P_M$ errors~[Appendix~\ref{app:error_model}]. Total number of Monte Carlo shots $N_{\text{tot}}$.\newline

\noindent \textbf{Simulation Output:} Fidelity $F=\langle \bar{\phi}|\bar{\rho}_N |\bar{\phi}\rangle=\text{Tr}(\bar{\rho}\bar{\rho}_N)$ of the output logical noisy magic state $\bar{\rho}_N$ with respect to the noiseless logical $\bar{\rho}=|\bar{\phi}\rangle \langle \bar{\phi}|$ magic state. Acceptance rate $p_{\text{acc}}=N_{\text{acc}}/N_{\text{tot}}$, for a number of accepted (preselected) shots $N_{\text{acc}}$.\newline

\begin{figure*}[!t]
    \centering
    \includegraphics[width=0.85\linewidth]{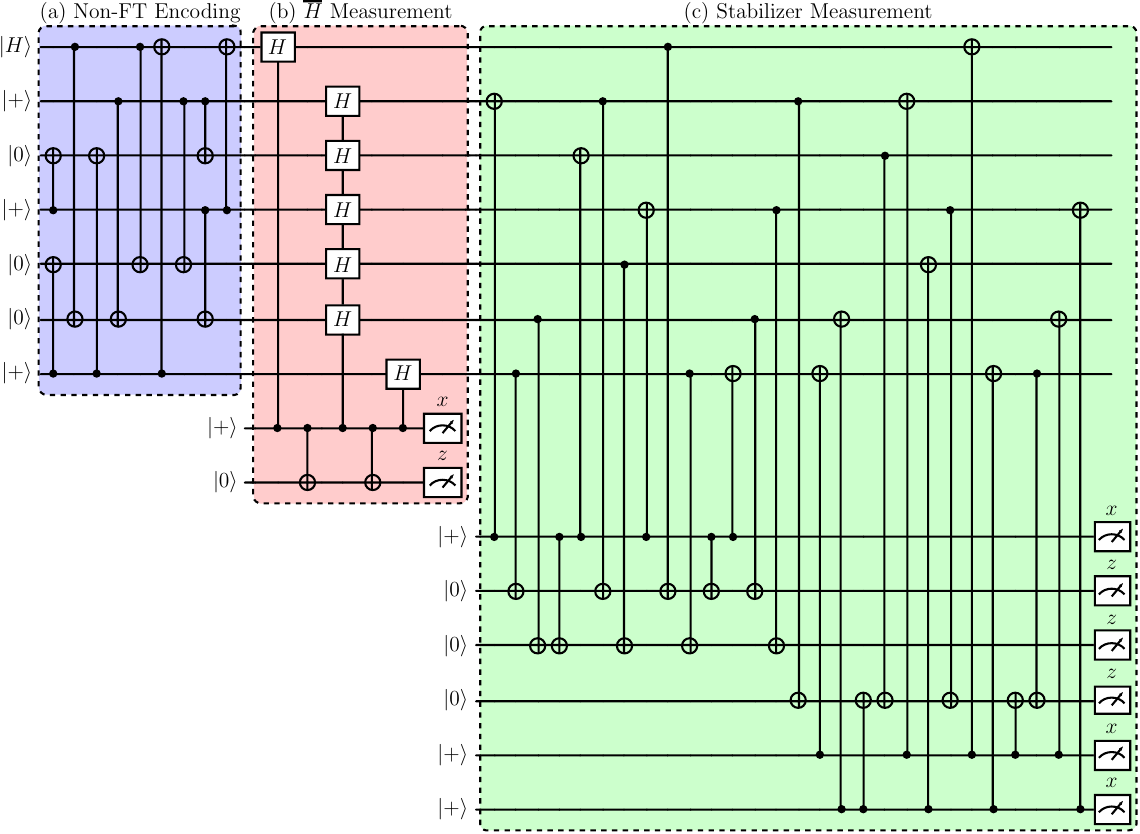}
    \caption{\textbf{Magic $|\bar{H}\rangle$ state preparation protocol on the [[7,1,3]] Steane code}. (a) Implements a unitary non-fault-tolerant encoding of the magic state $|\bar{H}\rangle$. (b) Measures the $\bar{H}$ operator using a 1-flag fault-tolerant gadget. Once again for ease of visualization, we compress otherwise sequentially-applied $C(H)$ gates into a single timeslice. (c) Performs a final round of fault-tolerant stabilizer measurements. Steps (b) and (c) correspond to a flag-based measurement with $(r_L,r_S)=(1,1)$. For the purposes of simulation, measurements are deferred to the end of the circuit and shots are postselected on $+1$ outcomes. See Figs.~\ref{fig:stim-vs-cirq-numerical-results1} and~\ref{fig:stim-vs-cirq-numerical-results2} for numerical results from simulating this protocol.}
    \label{fig:CH-meas-circ}
\end{figure*}

\noindent \textbf{Simulation Method:} 
\begin{itemize}
    \item \textbf{Error sampling [Fig.~\ref{fig:simulation-steps}(a)]:} Sample circuit-level Pauli errors at every location; see Appendix~\ref{app:error_model} for the error model. Of $N_{\text{tot}}$ samples, a fraction $f$ [Eq.~\eqref{eq:non_triv_Cliff}] of them contain nontrivial errors (non-identity Pauli strings).
    
    \item \textbf{Error propagation [Fig.~\ref{fig:simulation-steps}(b)]:} If nontrivial circuit-level Pauli errors are sampled, propagate these errors to an end-of-circuit Clifford error $C_\mathrm{prop}$ [Fig.~\ref{fig:simulation-steps}(b)].\footnote{In practice, it is enough to propagate Clifford errors up the last $U_i$ that belongs to $\mathcal{C}^{(3)}$.} 

    Upon error propagation, circuits containing non-Clifford gates can be effectively described by a noiseless logical magic state $|\bar{\phi}\rangle \otimes |a\rangle$ subjected to a Clifford error $C_\mathrm{prop}$ [Fig.~\ref{fig:simulation-steps}(c)].

    \item \textbf{State decomposition [Fig.~\ref{fig:simulation-steps}(d)]:} Decompose the noiseless logical magic state into a form that enables stabilizer simulation. For instance, $\bar{\rho}=|\bar{\phi}\rangle \langle \bar{\phi}|$ may be decomposed into a linear combination of logical Pauli matrices [Eq.~\eqref{eq:rep_rho_Paulis}], which can be further decomposed as a linear combination of logical stabilizer projectors. Alternatively, $|\bar{\phi}\rangle$  may be decomposed into a linear combination of stabilizer states [Eq.~\eqref{eq:stab_rak}].

    \item \textbf{Stabilizer-based simulation:} Apply the Clifford error $C_{\text{prop}}$ to a stabilizer state appearing in the aforementioned decomposition. Simulate measuring the $|a\rangle$ register, sampling over the postselected state. The logical error rate can be expressed as a linear combination of random variables that can be computed from such states and the noiseless magic state $|\bar{\phi}\rangle$; see Sec.~\ref{sec:fid_method} and Appendix~\ref{app:stab_rank_sim} for details. Take the sum of each expected value and renormalize by the acceptance rate $p_{acc}$, which yields the logical error probability.

\end{itemize}

\subsection{Numerical results: $|\bar{H}\rangle$ state preparation on the Steane code}
\label{sub:sim-application-example}

\begin{figure*}[!ht]
  \centering
  \begin{tabular}{cc}
  \includegraphics[width=0.5\textwidth]{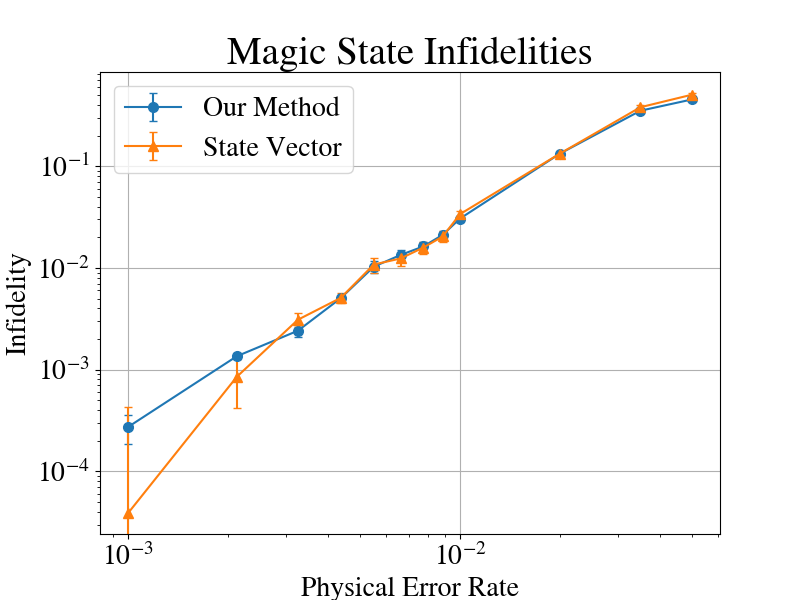} \label{fig:stim-cirq-infidelities}
  &
    \includegraphics[width=0.5\textwidth]{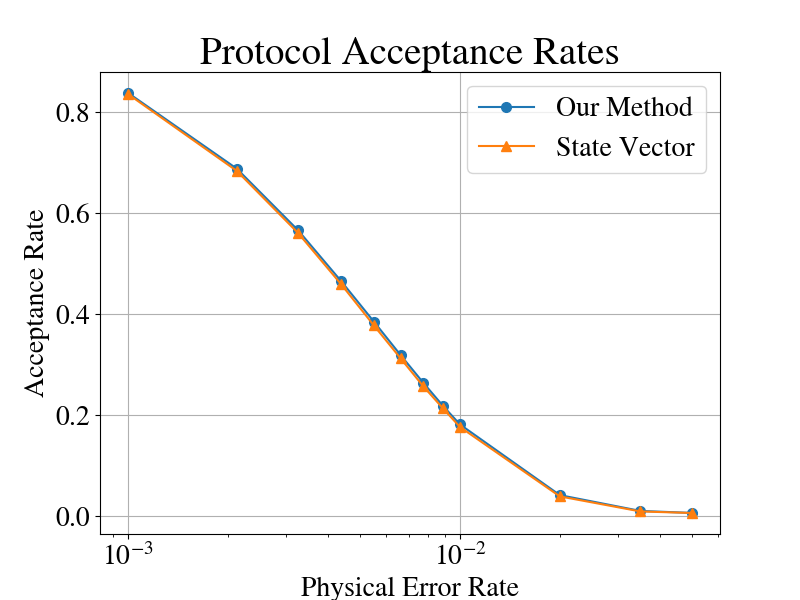} 
  \label{fig:stim-cirq-acceptance-rates} \\
    
    \small (a) & \small (b)
  \end{tabular}

    \caption{\textbf{Numerical simulations for the magic state $|\bar{H}\rangle$ preparation protocol on the $[[7,1,3]]$ Steane code under uniform circuit-level noise with physical error rate $p$}.
    Numerical simulations were performed with our phase-insensitive technique for $\bar{\rho}=|\bar{H}\rangle \langle \bar{H}| $ using the stabilizer-based simulator Stim~\cite{gidney2021stim} (blue) and for $|\bar{H}\rangle$  using the state-vector simulator Cirq~\cite{CirqDevelopers_2025} (orange). (a) Infidelity of the magic state as a function of the physical error rate $p$. At lower error rates, results from state-vector simulations slightly deviate from results from our technique due to finite sampling effects. (b) Acceptance rate as a function of the physical error rate $p$. We observe that the numerical results coincide for both simulation techniques. A comparison of the runtime for both simulation techniques can be found in Fig.~\ref{fig:stim-vs-cirq-numerical-results2}. Error bars correspond to one standard deviation. The number of shots employed for the stabilizer-based simulations was $N_{tot}\in[10^5,10^7]$ and $N_{tot}\in [10^5,10^6]$ for the state-vector simulations.}
   \label{fig:stim-vs-cirq-numerical-results1}
\end{figure*}

\begin{figure}[!ht]
  \centering
     \includegraphics[width=\linewidth]{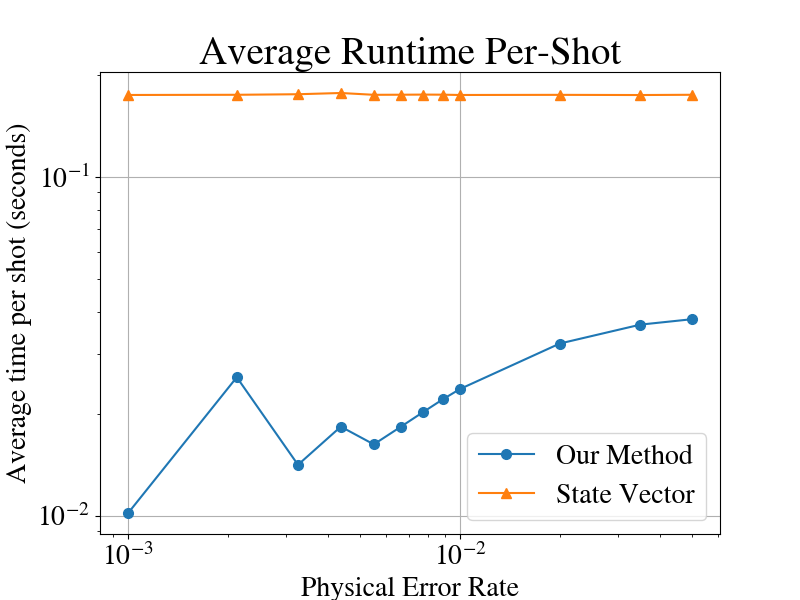}
    \label{fig:stim-cirq-avg-time-per-shot}
   
   \caption{\textbf{Average time per-simulation shot (in seconds) as a function of physical error rate $p$ for the magic state $|\bar{H}\rangle$ preparation protocol on the $[[7,1,3]]$ Steane code}. We observe an approximately $10\times$ improvement on the runtime of our technique (blue) compared to state-vector simulations (orange) for $p=10^{-3}$.}
   \label{fig:stim-vs-cirq-numerical-results2}
\end{figure}

As a proof-of-principle demonstration, we consider an MSP protocol~\cite{Goto:2016gss,Chamberland:2019ehl,Postler_2022} that prepares an $|\bar{H}\rangle$ state [Eq.(\ref{eq:magic_states})] stabilized by the Hadamard operator $\bar{H}$ in the [[7,1,3]] Steane code~\cite{Steane_1997}; see Fig.~\ref{fig:CH-meas-circ} for the circuit description of this protocol. This protocol is composed of a non-FT encoding of $|\bar{H}\rangle$ followed by a flag-based measurement of $\bar{H}$ with $r_L= 1$ PSC and $r_S=1$ stabilizer measurement rounds, respectively. We compare our phase-insensitive simulation technique with a state-vector simulation and achieve comparable results with sizable runtime improvements; see Figs.~\ref{fig:stim-vs-cirq-numerical-results1} and~\ref{fig:stim-vs-cirq-numerical-results2}). The slight difference in fidelity at low error rates is due to finite sampling effects.

We make a few remarks about our simulation. First, our phase-insensitive simulations rely on an explicitly-defined distribution for the propagated Clifford errors $C_\mathrm{prop}$ under circuit-level noise; see Appendix~\ref{app:prep_H} for details. Since obtaining an analytic formula for the propagated Clifford error is tedious in general, we provide an algorithm for systematically propagating circuit-level errors in Appendix~\ref{app:error_prop_alg}.

Second, in the simulation, on average only a small fraction of the simulation samples will have a nontrivial circuit-level error configuration. Let $M$ be the number of error locations. The expected fraction of samples that will contain an error $f$ is 
\begin{equation}
    f = 1- (1-p)^M \approx Mp \:,
    \label{eq:non_triv_Cliff}
\end{equation}
for some small error probability $p$. For instance, if $M=100$ and $p=10^{-3}$, one would expect 10\% of the samples to have nontrivial error configurations. The remaining samples need not require further processing.

Third, in \texttt{Stim}~\cite{gidney2021stim}, propagating errors that occur while encoding a logical magic state (prior to measuring them) can be done using the $\texttt{peek\_pauli\_flips()}$ function in the $\texttt{FlipSimulator()}$ class. More precisely, for fidelity estimation [Sec.~\ref{sec:fid_method}], one would plug in a set of stabilizer states corresponding to the stabilizer/Pauli rank decomposition of a magic state. For each such stabilizer state, error propagation can be performed with these functions, using the error model associated to the magic state [Appendix~\ref{app:error_model}].

Lastly, the main practical bottleneck of this approach, both in numerical simulations and hardware experiments involving magic states, is the sampling cost associated to the computation of Pauli expectation values necessary to estimate the fidelity. This is because these expectations have variance $O(1)$ since magic states are not eigenstates of logical Pauli operators. Indeed, the fidelity $F$ of magic states can be estimated as $F=1-O(1/\sqrt{N_\mathrm{tot}})$ for a total number of sampled shots $N_\mathrm{tot}$, which is in stark contrast with the fidelity of Pauli eigenstates $F= 1-O(1/N_\mathrm{tot})$~\cite{Mahler:2013mdg} \footnote{Recently,~\cite{Daguerre:2025boq} proposed a method to overcome this constraint by giving a lower bound $F\geq 1-O(1/N_\mathrm{tot})$ with quadratic improvement in sampling complexity at the expense of storing \textit{two copies} of the logical magic state and performing a logical Bell sampling; see also Ref.~\cite{Lee:2025lcs}.}. The phase-sensitive method can overcome this bottleneck in principle [Appendix~\ref{app:stab_rank_sim}], but we leave the implementation of this idea for future work.

For numerical examples of simulations for code-switching, we refer to~\cite{Daguerre:2024gjd} for a distance $3$ realization of Fig.~\ref{fig:code_switching}.

\subsection{Complexity}
\label{subsec:time-space-complexity}

We discuss the simulation complexity of our methods. Let $n_\mathrm{tot}$ be the total number of data and ancilla qubits the protocol acts on. We assume the protocol is defined by a measurement schedule consisting of $r$ rounds of measurements of either logical PSC or stabilizer measurements. Stabilizer measurements are defined in terms of a parity check matrix $H$, whose maximum row and column weights are $w_c$ and $w_q$, respectively. Finally, let $M$ denote the number of locations on which a Pauli error can occur.

We make the following simplifying assumptions. First, we focus on the standard protocol for measurements [Sec.~\ref{subsec:logical-clifford-meas}]. Second, we bound the number of tensor factors appearing in the logical PSC [Definition~\ref{def:transversal-psc}] with $n_{\mathrm{tot}}$. A tighter bound on the complexity may be obtained by noting that (i) the number of tensor factors of the logical PSC is strictly smaller than $n_{\mathrm{tot}}$ and (ii) the gate complexity of the propagated Clifford error in fault-tolerant protocols is lower [Corollary~\ref{corr:shor-style-protocol-full-error-prop}]. We do not discuss such details here.

Let us now discuss the time complexity of the two methods, the phase-insensitive method based on the Pauli rank decomposition [Sec.~\ref{sec:fid_method}] and the phase-sensitive method based on the stabilizer rank decomposition [Appendix~\ref{app:stab_rank_sim}] (for the space complexity, see Appendix~\ref{app:error_prop_alg} for more details). For both methods, each sample in a Monte Carlo simulation can be obtained by propagating a randomly sampled circuit-level error configuration to an end-of-circuit Clifford error. A conservative bound on the time complexity of obtaining this Clifford error is $O(Mr^3w_q^2n_\mathrm{tot}^2)$ [Appendix~\ref{app:error_prop_alg}].
Furthermore, each such propagated error consists of at most $O(Mr^2 w_qn_{\mathrm{tot}})$ one- and two-qubit Clifford gates [Corollary~\ref{corr:standard-protocol-full-error-prop}].

One approach to obtain the Clifford error is by sequentially pushing sampled circuit-level errors through the protocol circuit from left-to-right to obtain an end-of-circuit Clifford error [Appendix~\ref{app:error_prop_alg}]. A second approach, which we assume here, is to \textit{individually} propagate circuit-level errors at each of the $M$ error locations to a corresponding end-of-circuit Clifford error. These Clifford errors are then combined (in an appropriate order) to obtain the full end-of-circuit Clifford error. For either approach, this error can then be expressed as either a combined Clifford tableau (in the phase-insensitive method) or as a sequence of one- and two-qubit gates (in the phase-sensitive method). 

Subsequently, we describe the complexity of obtaining Monte Carlo samples for each of the two simulation methods. Note, these complexities do not include the one-time cost for constructing a look-up table that takes in as input an error location and a circuit-level error at that location and returns an end-of-circuit propagated Clifford error. In principle, such a look-up table can be constructed on a case-by-case basis using propagation rules as defined in Fig.~\ref{fig:propagation_rules_visual}. This look-up table is to be contrasted with the look-up table defined in Fig.~\ref{fig:propagation_alg_visual}, which computes propagation on a per-gate level.

\textbf{Phase-insensitive method}: For the Pauli rank-based method [Sec.~\ref{sec:fid_method}], it is convenient to store each of the propagated errors in the standard stabilizer tableau representation~\cite{Aaronson:2004xuh}. Combining these for the $M$ error locations would incurs an additional $O(Mn_{\mathrm{tot}}^3)$ time cost. This is the complexity needed for each sample. 

Note that for this method, there are additional parameters needed to specify the complexity. Let $p$ be the Pauli rank of the $k$-qubit logical magic state. To compute the fidelity of the noisy logical magic state, we require computing $p$ Pauli expectation values, each of which requires simulating $O(p2^k)$ Clifford circuits (comprised of the encoding circuit, the end-of-circuit Clifford error, and final measurements) on stabilizer state inputs. Therefore, to estimate the fidelity up to an additive error $\epsilon$, it suffices to have $O(p^22^k/\epsilon^2)$ samples. Thus the overall time cost is $O(p^22^kMn_{\mathrm{tot}}^3/\epsilon^2)$.

\textbf{Phase-sensitive method}: For the stabilizer rank-based method [Appendix~\ref{app:stab_rank_sim}], it is be convenient to store each of the propagated errors as a sequence of one- and two-qubit Clifford gates. The number of gates in the combined end-of-circuit Clifford error is $O(Mr^2w_qn_{\mathrm{tot}})$. Moreover, updating the stabilizer rank decomposition for each of these gates takes time at most $O(n_{\mathrm{tot}}^2)$~\cite{Bravyi:2018ugg}. Thus the complexity needed for each sample is $O(Mr^2w_qn_{\mathrm{tot}}^3)$. 

Let $q$ be the stabilizer rank of the logical magic state. To estimate the logical error rate up to an additive precision of $\epsilon$, it suffices to have $O(1/\epsilon)$ samples [Appendix~\ref{app:stab_rank_sim}]. For each such sample, the needed computation boils down to (i) updating $q$ stabilizer states using the aforementioned method and (ii) computing the overlap with $q$ stabilizer states. It is straightforward to see that the former costs $O(qMr^2w_qn_{\mathrm{tot}}^3)$. For the latter, note that each overlap can be computed in $O(n_{\textrm{tot}}^3)$ time~\cite[Lemma 3, 4]{Bravyi:2018ugg}. Thus, the overall time cost becomes $O(q^2 n_{\mathrm{tot}}^3)$. In the $q=O(1)$ regime, the overall time complexity is $O(Mr^2w_qn_{\mathrm{tot}}^3/\epsilon)$.

\textbf{Comparison}: These two methods have differing complexities and are incomparable at the moment. Nonetheless, the key point is that these methods have time and space complexity polynomial in the parameters that define the MSP protocol (except for the exponential dependence on $k$ for the phase-insensitive case). This is in stark contrast with expensive state-vector simulation, whose cost scales exponentially with $n_{\mathrm{tot}}$.

Finally, the additive error of the fidelity also determines the overhead complexity in each case, which in practice represents a fundamental limitation for the simulations. Concretely, one may consider the task of estimating the logical error rate $p_L$ up to (say) one standard deviation. For this task, the phase-sensitive method requires $O(1/p_L)$ independent samples whereas the phase-insensitive method requires $O(1/p_L^2)$ independent samples.

\section{Discussion}
\label{sec:discussion}

In this paper, we introduce a new method for simulating logical MSP protocols under circuit-level noise. Our method is applicable to a broad range of protocols, including cultivation-style methods~\cite{Goto:2016gss,Chamberland:2019ehl,Chamberland:2020axi,Itogawa_2025,Gidney:2024alh,sahay2025foldtransversalsurfacecodecultivation,claes2025cultivatingtstatessurface,Vaknin:2025pbp}, MSD at the logical level~\cite{Bravyi_2005,meier2012magic,Bravyi_2012,Jones_2013a,fowler_surface_2013,Jones_2013b,duclos_cianci_distillation_2013,Duclos_Cianci_2015,Campbell_2017,O_Gorman_2017,Haah2018codesprotocols}, and code-switching~\cite{Anderson:2014jvy,bombin2016dimensionaljumpquantumerror,Beverland2021,Butt_2024,Daguerre:2024gjd}. Thus, our work provides a flexible, accurate, and efficient method to benchmark these protocols.

We expect our method to enable the following studies, which we leave for future work. First, we may be able to address an open question concerning magic state cultivation~\cite{Gidney:2024alh}: the relationship between the fidelity of the prepared magic state and that of a closely related stabilizer state. Because there was no known method to simulate these protocols efficiently at large code distances, the ratio between the two at small code distance ($d=3$) was used to deduce the magic state fidelity at higher distances. It may be possible to rigorously test the validity of this assumption using our method. 

In order to do so, we will need to develop a software toolkit capable of executing large-scale simulations of MSP protocols. One practical complication resides in the implementation of a general-purpose, efficient error propagation algorithm (such as the one described in Appendix~\ref{app:error_prop_alg}). For instance, the Clifford error cannot naively be stored in its tableau representation since propagation rules through non-Clifford gates are cumbersome to implement without an explicit circuit representation. Another approach consists of storing the Clifford error as an ordered list of gates, where propagation rules can be defined and applied layer by layer.

Alternatively, to efficiently obtain many samples of propagated errors, one can precompute propagated errors for each error location in a protocol circuit and store them in a table as Clifford tableaus. This one-time cost, although initially expensive, will not contribute to the per-shot runtime complexity as propagated errors for each sampled error configuration can be quickly constructed via a few cheap accesses to the precomputed table. However, the implementation of such a software toolkit is beyond the scope of this article and its development is left for future work.

Second, we may be able to gain an analytic understanding of circuit-level Pauli noise in the presence of non-Clifford gates. The most striking effect of these gates is that a single-qubit Pauli error can propagate to an \emph{entangling Clifford error}; see Fig.~\ref{fig:422_propagation} for an example. Inspecting the form of the propagated errors may be helpful to further optimize quantum error correction protocols involving non-Clifford gates. In particular, it is important to understand the effectiveness of flag gadgets in the presence of these entangling Clifford errors~\cite{Yoder_2017,Chao_2018a,Chao_2018b}.

One intriguing possibility is that one may be able to decode errors better by taking the error propagation rule into account. In the presence of an entangling Clifford error, measurement of a syndrome qubit induces a projective measurement on a subset of data qubits in some known basis. For instance, if there is a CZ error between a data qubit and a syndrome qubit, measuring the syndrome qubit in the $X$-basis induces a $Z$-basis measurement of the data qubit, wherein the postmeasurement state of the data qubit is completely determined by the $X$-basis measurement. This extra information may be useful for the decoder. We leave such studies for future work.

We also leave the following open questions for future work. One question is whether our method can be generalized to alternative protocols. For instance, it is possible to apply a non-Clifford gate by viewing a three-dimensional protocol involving a transversal non-Clifford gate as a two-dimensional protocol  ~\cite{bombin20182dquantumcomputation3d,Davydova:2025ylx,kobayashi2025cliffordhierarchystabilizercodes,bauer2025planarfaulttolerantcircuitsnonclifford} evolving in time (under some additional conditions). In these methods, the syndrome measurements cannot be deferred to the end, making our current simulation technique inadequate. We also note that these protocols can be viewed as fault-tolerant gates of  recently studied Clifford stabilizer codes~\cite{Davydova:2025ylx,kobayashi2025cliffordhierarchystabilizercodes}, which are in fact generated by PSCs. Can the special properties of PSCs [Sec.~\ref{sec:controlled_clifford}] let us extend the stabilizer formalism to such Clifford stabilizers? Another family of protocols one may consider is a cultivation-style protocol involving Cliffords that do not square to a Pauli~\cite{claes2025cultivatingtstatessurface}. It is unclear if a central insight of this paper --- propagation of circuit-level Pauli noise to a Clifford --- is valid in such protocols.

Lastly, we remark that our method is useful only for methods that prepare a small amount of magic. How to efficiently simulate MSP protocols that can produce a large amount of magic \cite{wills2024constantoverheadmagicstatedistillation,golowich2024asymptoticallygoodquantumcodes,nguyen2024goodbinaryquantumcodes,golowich2024quantumldpccodestransversal} remains an interesting open question.
%
\smallskip
\vspace{1em}
{\centerline{\textbf{ACKNOWLEDGMENTS}}\par}
\vspace{1em}
We thank Michael Gullans, Dominic Williamson, Takada Yugo, Kwok Ho Wan for helpful discussions. We acknowledge support from the Advanced Scientific Computing Research program in the Office of Science of the Department of Energy (DE-SC0026109). LD is supported by a Dean's Distinguished Graduate Fellowship from the College of Letters and Science of the University of California, Davis. 

\vspace{1em}
{\centering \textbf{DATA AVAILABILITY}\par}
\vspace{1em}

All the datasets and codes can be made available upon reasonable request to the corresponding author.

\bibliography{bib_magic_simulation}

\providecommand{\href}[2]{#2}\begingroup\raggedright\begin{thebibliography}{10}

\bibitem{Shor1995}
P.~W. Shor, ``{Scheme for reducing decoherence in quantum computer memory},'' \href{http://dx.doi.org/10.1103/physreva.52.r2493}{{\em Phys. Rev. A} {\bfseries 52} no.~4, (1995) R2493--R2496}.

\bibitem{aharonov1997fault}
D.~Aharonov and M.~Ben-Or, ``{Fault-Tolerant Quantum Computation with Constant Error Rate},'' \href{http://dx.doi.org/10.1137/S0097539799359385}{{\em SIAM J. Comput.} {\bfseries 38} no.~4, (2008) 1207--1282}, \href{http://arxiv.org/abs/quant-ph/9906129}{{\ttfamily arXiv:quant-ph/9906129}}.

\bibitem{dennis2002}
E.~Dennis, A.~Kitaev, A.~Landahl, and J.~Preskill, ``Topological quantum memory,'' \href{http://dx.doi.org/10.1063/1.1499754}{{\em Journal of Mathematical Physics} {\bfseries 43} no.~9, (09, 2002) 4452--4505}. \url{https://doi.org/10.1063/1.1499754}.

\bibitem{Horsman_2012}
D.~Horsman, A.~G. Fowler, S.~Devitt, and R.~V. Meter, ``Surface code quantum computing by lattice surgery,'' \href{http://dx.doi.org/10.1088/1367-2630/14/12/123011}{{\em New Journal of Physics} {\bfseries 14} no.~12, (Dec, 2012) 123011}. \url{https://doi.org/10.1088/1367-2630/14/12/123011}.

\bibitem{Bombin2006}
H.~Bombin and M.~A. Martin-Delgado, ``Topological quantum distillation,'' \href{http://dx.doi.org/10.1103/PhysRevLett.97.180501}{{\em Phys. Rev. Lett.} {\bfseries 97} (Oct, 2006) 180501}. \url{https://link.aps.org/doi/10.1103/PhysRevLett.97.180501}.

\bibitem{Cohen2022}
L.~Z. Cohen, I.~H. Kim, S.~D. Bartlett, and B.~J. Brown, ``{Low-overhead fault-tolerant quantum computing using long-range connectivity},'' \href{http://dx.doi.org/10.1126/sciadv.abn1717}{{\em Sci. Adv.} {\bfseries 8} no.~20, (2022) abn1717}, \href{http://arxiv.org/abs/2110.10794}{{\ttfamily arXiv:2110.10794 [quant-ph]}}.

\bibitem{Quintavalle2023partitioningqubits}
A.~O. Quintavalle, P.~Webster, and M.~Vasmer, ``Partitioning qubits in hypergraph product codes to implement logical gates,'' \href{http://dx.doi.org/10.22331/q-2023-10-24-1153}{{\em {Quantum}} {\bfseries 7} (Oct., 2023) 1153}. \url{https://doi.org/10.22331/q-2023-10-24-1153}.

\bibitem{Bravyi_2005}
S.~Bravyi and A.~Kitaev, ``Universal quantum computation with ideal clifford gates and noisy ancillas,'' \href{http://dx.doi.org/10.1103/physreva.71.022316}{{\em Physical Review A} {\bfseries 71} no.~2, (Feb., 2005) }. \url{http://dx.doi.org/10.1103/PhysRevA.71.022316}.

\bibitem{meier2012magic}
A.~M. Meier, B.~Eastin, and E.~Knill, ``{Magic-state distillation with the four-qubit code},'' \href{http://arxiv.org/abs/1204.4221}{{\ttfamily arXiv:1204.4221 [quant-ph]}}.

\bibitem{Bravyi_2012}
S.~Bravyi and J.~Haah, ``Magic-state distillation with low overhead,'' \href{http://dx.doi.org/10.1103/physreva.86.052329}{{\em Physical Review A} {\bfseries 86} no.~5, (Nov., 2012) }. \url{http://dx.doi.org/10.1103/PhysRevA.86.052329}.

\bibitem{Jones_2013a}
C.~Jones, ``Low-overhead constructions for the fault-tolerant toffoli gate,'' \href{http://dx.doi.org/10.1103/physreva.87.022328}{{\em Physical Review A} {\bfseries 87} no.~2, (Feb., 2013) }. \url{http://dx.doi.org/10.1103/PhysRevA.87.022328}.

\bibitem{fowler_surface_2013}
A.~G. Fowler, S.~J. Devitt, and C.~Jones, ``Surface code implementation of block code state distillation,'' \href{http://dx.doi.org/10.1038/srep01939}{{\em Scientific Reports} {\bfseries 3} no.~1, (June, 2013) 1939}. \url{https://www.nature.com/articles/srep01939}. Publisher: Nature Publishing Group.

\bibitem{Jones_2013b}
C.~Jones, ``Multilevel distillation of magic states for quantum computing,'' \href{http://dx.doi.org/10.1103/physreva.87.042305}{{\em Physical Review A} {\bfseries 87} no.~4, (Apr., 2013) }. \url{http://dx.doi.org/10.1103/PhysRevA.87.042305}.

\bibitem{duclos_cianci_distillation_2013}
G.~Duclos-Cianci and K.~M. Svore, ``Distillation of nonstabilizer states for universal quantum computation,'' \href{http://dx.doi.org/10.1103/PhysRevA.88.042325}{{\em Physical Review A} {\bfseries 88} no.~4, (Oct., 2013) 042325}. \url{https://link.aps.org/doi/10.1103/PhysRevA.88.042325}. Publisher: American Physical Society.

\bibitem{Duclos_Cianci_2015}
G.~Duclos-Cianci and D.~Poulin, ``Reducing the quantum-computing overhead with complex gate distillation,'' \href{http://dx.doi.org/10.1103/physreva.91.042315}{{\em Physical Review A} {\bfseries 91} no.~4, (Apr., 2015) }. \url{http://dx.doi.org/10.1103/PhysRevA.91.042315}.

\bibitem{Campbell_2017}
E.~T. Campbell and M.~Howard, ``Unified framework for magic state distillation and multiqubit gate synthesis with reduced resource cost,'' \href{http://dx.doi.org/10.1103/physreva.95.022316}{{\em Physical Review A} {\bfseries 95} no.~2, (Feb., 2017) }. \url{http://dx.doi.org/10.1103/PhysRevA.95.022316}.

\bibitem{O_Gorman_2017}
J.~O’Gorman and E.~T. Campbell, ``Quantum computation with realistic magic-state factories,'' \href{http://dx.doi.org/10.1103/physreva.95.032338}{{\em Physical Review A} {\bfseries 95} no.~3, (Mar., 2017) }. \url{http://dx.doi.org/10.1103/PhysRevA.95.032338}.

\bibitem{Haah2018codesprotocols}
J.~Haah and M.~B. Hastings, ``Codes and {P}rotocols for {D}istilling {$T$}, controlled-{$S$}, and {T}offoli {G}ates,'' \href{http://dx.doi.org/10.22331/q-2018-06-07-71}{{\em {Quantum}} {\bfseries 2} (June, 2018) 71}. \url{https://doi.org/10.22331/q-2018-06-07-71}.

\bibitem{Gidney2019efficientmagicstate}
C.~Gidney and A.~G. Fowler, ``Efficient magic state factories with a catalyzed {$|CCZ\rangle$} to {$2|T\rangle$} transformation,'' \href{http://dx.doi.org/10.22331/q-2019-04-30-135}{{\em {Quantum}} {\bfseries 3} (Apr., 2019) 135}. \url{https://doi.org/10.22331/q-2019-04-30-135}.

\bibitem{Reiher2017}
M.~Reiher, N.~Wiebe, K.~M. Svore, D.~Wecker, and M.~Troyer, ``{Elucidating reaction mechanisms on quantum computers},'' \href{http://dx.doi.org/10.1073/pnas.1619152114}{{\em Proc. Nat. Acad. Sci.} {\bfseries 114} no.~29, (2017) 7555--7560}, \href{http://arxiv.org/abs/1605.03590}{{\ttfamily arXiv:1605.03590 [quant-ph]}}.

\bibitem{Li_2015}
Y.~Li, ``A magic state’s fidelity can be superior to the operations that created it,'' \href{http://dx.doi.org/10.1088/1367-2630/17/2/023037}{{\em New Journal of Physics} {\bfseries 17} no.~2, (Feb., 2015) 023037}. \url{http://dx.doi.org/10.1088/1367-2630/17/2/023037}.

\bibitem{Yoder_2017}
T.~J. Yoder and I.~H. Kim, ``The surface code with a twist,'' \href{http://dx.doi.org/10.22331/q-2017-04-25-2}{{\em Quantum} {\bfseries 1} (Apr., 2017) 2}. \url{http://dx.doi.org/10.22331/q-2017-04-25-2}.

\bibitem{Chao_2018a}
R.~Chao and B.~W. Reichardt, ``Quantum error correction with only two extra qubits,'' \href{http://dx.doi.org/10.1103/physrevlett.121.050502}{{\em Physical Review Letters} {\bfseries 121} no.~5, (Aug., 2018) }. \url{http://dx.doi.org/10.1103/PhysRevLett.121.050502}.

\bibitem{Chao_2018b}
R.~Chao and B.~W. Reichardt, ``Fault-tolerant quantum computation with few qubits,'' \href{http://dx.doi.org/10.1038/s41534-018-0085-z}{{\em npj Quantum Information} {\bfseries 4} no.~1, (Sept., 2018) }. \url{http://dx.doi.org/10.1038/s41534-018-0085-z}.

\bibitem{Chamberland:2020axi}
C.~Chamberland and K.~Noh, ``{Very low overhead fault-tolerant magic state preparation using redundant ancilla encoding and flag qubits},'' \href{http://dx.doi.org/10.1038/s41534-020-00319-5}{{\em npj Quantum Inf.} {\bfseries 6} (2020) 91}.

\bibitem{Itogawa_2025}
T.~Itogawa, Y.~Takada, Y.~Hirano, and K.~Fujii, ``Efficient magic state distillation by zero-level distillation,'' \href{http://dx.doi.org/10.1103/thxx-njr6}{{\em PRX Quantum} {\bfseries 6} no.~2, (June, 2025) }. \url{http://dx.doi.org/10.1103/thxx-njr6}.

\bibitem{Gidney:2024alh}
C.~Gidney, N.~Shutty, and C.~Jones, ``{Magic state cultivation: growing T states as cheap as CNOT gates},'' \href{http://arxiv.org/abs/2409.17595}{{\ttfamily arXiv:2409.17595 [quant-ph]}}.

\bibitem{Daguerre:2024gjd}
L.~Daguerre and I.~H. Kim, ``{Code switching revisited: Low-overhead magic state preparation using color codes},'' \href{http://dx.doi.org/10.1103/PhysRevResearch.7.023080}{{\em Phys. Rev. Res.} {\bfseries 7} no.~2, (2025) 023080}, \href{http://arxiv.org/abs/2410.07327}{{\ttfamily arXiv:2410.07327 [quant-ph]}}.

\bibitem{sahay2025foldtransversalsurfacecodecultivation}
K.~Sahay, P.-K. Tsai, K.~Chang, Q.~Su, T.~B. Smith, S.~Singh, and S.~Puri, ``{Fold-transversal surface code cultivation},'' \href{http://arxiv.org/abs/2509.05212}{{\ttfamily arXiv:2509.05212 [quant-ph]}}.

\bibitem{claes2025cultivatingtstatessurface}
J.~Claes, ``{Cultivating T states on the surface code with only two-qubit gates},'' \href{http://arxiv.org/abs/2509.05232}{{\ttfamily arXiv:2509.05232 [quant-ph]}}.

\bibitem{Vaknin:2025pbp}
Y.~Vaknin, S.~Jacoby, A.~Grimsmo, and A.~Retzker, ``{Magic State Cultivation on the Surface Code},'' \href{http://arxiv.org/abs/2502.01743}{{\ttfamily arXiv:2502.01743 [quant-ph]}}.

\bibitem{Chen:2025imz}
Z.-H. Chen, M.-C. Chen, C.-Y. Lu, and J.-W. Pan, ``{Efficient Magic State Cultivation on $\mathbb{RP}^2$},'' \href{http://arxiv.org/abs/2503.18657}{{\ttfamily arXiv:2503.18657 [quant-ph]}}.

\bibitem{Lacroix:2024vls}
N.~Lacroix {\em et~al.}, ``{Scaling and logic in the color code on a superconducting quantum processor},'' \href{http://arxiv.org/abs/2412.14256}{{\ttfamily arXiv:2412.14256 [quant-ph]}}.

\bibitem{Rodriguez:2024bhh}
P.~S. Rodriguez {\em et~al.}, ``{Experimental demonstration of logical magic state distillation},'' \href{http://dx.doi.org/10.1038/s41586-025-09367-3}{{\em Nature} {\bfseries 645} no.~8081, (2025) 620--625}, \href{http://arxiv.org/abs/2412.15165}{{\ttfamily arXiv:2412.15165 [quant-ph]}}.

\bibitem{pogorelov2024experimentalfaulttolerantcodeswitching}
I.~Pogorelov, F.~Butt, L.~Postler, C.~D. Marciniak, P.~Schindler, M.~M{\"u}ller, and T.~Monz, ``{Experimental fault-tolerant code switching},'' \href{http://dx.doi.org/10.1038/s41567-024-02727-2}{{\em Nature Phys.} {\bfseries 21} no.~2, (2025) 298--303}, \href{http://arxiv.org/abs/2403.13732}{{\ttfamily arXiv:2403.13732 [quant-ph]}}.

\bibitem{Kim:2024vmw}
Y.~Kim, M.~Sevior, and M.~Usman, ``{Magic State Injection on IBM Quantum Processors Above the Distillation Threshold},'' \href{http://arxiv.org/abs/2412.01446}{{\ttfamily arXiv:2412.01446 [quant-ph]}}.

\bibitem{Ye:2023hxg}
Y.~Ye {\em et~al.}, ``{Logical Magic State Preparation with Fidelity beyond the Distillation Threshold on a Superconducting Quantum Processor},'' \href{http://dx.doi.org/10.1103/PhysRevLett.131.210603}{{\em Phys. Rev. Lett.} {\bfseries 131} no.~21, (2023) 210603}, \href{http://arxiv.org/abs/2305.15972}{{\ttfamily arXiv:2305.15972 [quant-ph]}}.

\bibitem{Gupta:2023zei}
R.~S. Gupta {\em et~al.}, ``{Encoding a magic state with beyond break-even fidelity},'' \href{http://dx.doi.org/10.1038/s41586-023-06846-3}{{\em Nature} {\bfseries 625} no.~7994, (2024) 259--263}, \href{http://arxiv.org/abs/2305.13581}{{\ttfamily arXiv:2305.13581 [quant-ph]}}.

\bibitem{Postler_2022}
L.~Postler, S.~Heuben, I.~Pogorelov, M.~Rispler, T.~Feldker, M.~Meth, C.~D. Marciniak, R.~Stricker, M.~Ringbauer, R.~Blatt, P.~Schindler, M.~Müller, and T.~Monz, ``Demonstration of fault-tolerant universal quantum gate operations,'' \href{http://dx.doi.org/10.1038/s41586-022-04721-1}{{\em Nature} {\bfseries 605} no.~7911, (May, 2022) 675–680}. \url{http://dx.doi.org/10.1038/s41586-022-04721-1}.

\bibitem{Anderson2021}
C.~Ryan-Anderson {\em et~al.}, ``Realization of real-time fault-tolerant quantum error correction,'' \href{http://dx.doi.org/10.1103/PhysRevX.11.041058}{{\em Phys. Rev. X} {\bfseries 11} (Dec, 2021) 041058}, \href{http://arxiv.org/abs/2107.07505}{{\ttfamily arXiv:2107.07505 [quant-ph]}}.

\bibitem{dasu2025breakingmagicdemonstrationhighfidelity}
S.~Dasu, S.~Burton, K.~Mayer, D.~Amaro, J.~A. Gerber, K.~Gilmore, D.~Gresh, D.~DelVento, A.~C. Potter, and D.~Hayes, ``{Breaking even with magic: demonstration of a high-fidelity logical non-Clifford gate},'' \href{http://arxiv.org/abs/2506.14688}{{\ttfamily arXiv:2506.14688 [quant-ph]}}.

\bibitem{Rosenfeld:2025xvf}
E.~Rosenfeld {\em et~al.}, ``{Magic state cultivation on a superconducting quantum processor},'' \href{http://arxiv.org/abs/2512.13908}{{\ttfamily arXiv:2512.13908 [quant-ph]}}.

\bibitem{Daguerre:2025boq}
L.~Daguerre, R.~Blume-Kohout, N.~C. Brown, D.~Hayes, and I.~H. Kim, ``{Experimental Demonstration of High-Fidelity Logical Magic States from Code Switching},'' \href{http://dx.doi.org/10.1103/dck4-x9c2}{{\em Phys. Rev. X} {\bfseries 15} no.~4, (2025) 041008}, \href{http://arxiv.org/abs/2506.14169}{{\ttfamily arXiv:2506.14169 [quant-ph]}}.

\bibitem{Gottesman:1998hu}
D.~Gottesman, ``{The Heisenberg representation of quantum computers},'' {\em {22nd International Colloquium on Group Theoretical Methods in Physics}} (7, 1998) 32--43, \href{http://arxiv.org/abs/quant-ph/9807006}{{\ttfamily arXiv:quant-ph/9807006}}.

\bibitem{Aaronson:2004xuh}
S.~Aaronson and D.~Gottesman, ``{Improved simulation of stabilizer circuits},'' \href{http://dx.doi.org/10.1103/PhysRevA.70.052328}{{\em Phys. Rev. A} {\bfseries 70} no.~5, (2004) 052328}, \href{http://arxiv.org/abs/quant-ph/0406196}{{\ttfamily arXiv:quant-ph/0406196}}.

\bibitem{Bravyi:2016yhm}
S.~Bravyi and D.~Gosset, ``{Improved Classical Simulation of Quantum Circuits Dominated by Clifford Gates},'' \href{http://dx.doi.org/10.1103/PhysRevLett.116.250501}{{\em Phys. Rev. Lett.} {\bfseries 116} no.~25, (2016) 250501}.

\bibitem{Bravyi:2016uqx}
S.~Bravyi, G.~Smith, and J.~A. Smolin, ``{Trading Classical and Quantum Computational Resources},'' \href{http://dx.doi.org/10.1103/PhysRevX.6.021043}{{\em Phys. Rev. X} {\bfseries 6} no.~2, (2016) 021043}.

\bibitem{Bravyi:2018ugg}
S.~Bravyi, D.~Browne, P.~Calpin, E.~Campbell, D.~Gosset, and M.~Howard, ``{Simulation of quantum circuits by low-rank stabilizer decompositions},'' \href{http://dx.doi.org/10.22331/q-2019-09-02-181}{{\em Quantum} {\bfseries 3} (2019) 181}, \href{http://arxiv.org/abs/1808.00128}{{\ttfamily arXiv:1808.00128 [quant-ph]}}.

\bibitem{Bu:2019qed}
K.~Bu and D.~E. Koh, ``{Efficient Classical Simulation of Clifford Circuits with Nonstabilizer Input States},'' \href{http://dx.doi.org/10.1103/PhysRevLett.123.170502}{{\em Phys. Rev. Lett.} {\bfseries 123} no.~17, (2019) 170502}.

\bibitem{Masot-Llima:2024doz}
S.~Masot-Llima and A.~Garcia-Saez, ``{Stabilizer Tensor Networks: Universal Quantum Simulator on a Basis of Stabilizer States},'' \href{http://dx.doi.org/10.1103/PhysRevLett.133.230601}{{\em Phys. Rev. Lett.} {\bfseries 133} no.~23, (2024) 230601}, \href{http://arxiv.org/abs/2403.08724}{{\ttfamily arXiv:2403.08724 [quant-ph]}}.

\bibitem{Zhang:2024jlz}
Y.~Zhang and Y.~Zhang, ``{Classical Simulability of Quantum Circuits with Shallow Magic Depth},'' \href{http://dx.doi.org/10.1103/PRXQuantum.6.010337}{{\em PRX Quantum} {\bfseries 6} no.~1, (2025) 010337}, \href{http://arxiv.org/abs/2409.13809}{{\ttfamily arXiv:2409.13809 [quant-ph]}}.

\bibitem{Garner:2025xts}
S.~Garner, C.~Liu, M.~Wang, S.~Stein, and A.~Li, ``{STABSim: A Parallelized Clifford Simulator with Features Beyond Direct Simulation},'' \href{http://arxiv.org/abs/2507.03092}{{\ttfamily arXiv:2507.03092 [quant-ph]}}.

\bibitem{Aziz:2025cnu}
K.~Aziz, H.~Pan, M.~J. Gullans, and J.~H. Pixley, ``{Classical Simulations of Low Magic Quantum Dynamics},'' \href{http://arxiv.org/abs/2508.20252}{{\ttfamily arXiv:2508.20252 [quant-ph]}}.

\bibitem{Goto:2016gss}
H.~Goto, ``{Minimizing resource overheads for fault-tolerant preparation of encoded states of the Steane code},'' \href{http://dx.doi.org/10.1038/srep19578}{{\em Sci. Rep.} {\bfseries 6} no.~1, (2016) 19578}.

\bibitem{Chamberland_2018}
C.~Chamberland and M.~E. Beverland, ``Flag fault-tolerant error correction with arbitrary distance codes,'' \href{http://dx.doi.org/10.22331/q-2018-02-08-53}{{\em Quantum} {\bfseries 2} (Feb., 2018) 53}. \url{http://dx.doi.org/10.22331/q-2018-02-08-53}.

\bibitem{Chamberland:2019ehl}
C.~Chamberland and A.~W. Cross, ``{Fault-tolerant magic state preparation with flag qubits},'' \href{http://dx.doi.org/10.22331/q-2019-05-20-143}{{\em Quantum} {\bfseries 3} (2019) 143}, \href{http://arxiv.org/abs/1811.00566}{{\ttfamily arXiv:1811.00566}}.

\bibitem{Butt_2024}
F.~Butt, S.~Heußen, M.~Rispler, and M.~Müller, ``Fault-tolerant code-switching protocols for near-term quantum processors,'' \href{http://dx.doi.org/10.1103/prxquantum.5.020345}{{\em PRX Quantum} {\bfseries 5} no.~2, (May, 2024) }. \url{http://dx.doi.org/10.1103/PRXQuantum.5.020345}.

\bibitem{wan2025cuttingstabiliserdecompositionsmagic}
K.~H. Wan and Z.~Zhong, ``{Cutting stabiliser decompositions of magic state cultivation with ZX-calculus},'' \href{http://arxiv.org/abs/2509.01224}{{\ttfamily arXiv:2509.01224 [quant-ph]}}.

\bibitem{wan2025simulatemagicstatecultivation}
K.~H. Wan and Z.~Zhong, ``{How to simulate magic state cultivation with around $8$ Clifford terms on average},'' \href{http://arxiv.org/abs/2509.08658}{{\ttfamily arXiv:2509.08658 [quant-ph]}}.

\bibitem{Sutcliffe_2024}
M.~Sutcliffe and A.~Kissinger, ``Procedurally optimised zx-diagram cutting for efficient t-decomposition in classical simulation,'' \href{http://dx.doi.org/10.4204/eptcs.406.3}{{\em Electronic Proceedings in Theoretical Computer Science} {\bfseries 406} (Aug., 2024) 63–78}. \url{http://dx.doi.org/10.4204/EPTCS.406.3}.

\bibitem{Anderson:2014jvy}
J.~T. Anderson, G.~Duclos-Cianci, and D.~Poulin, ``{Fault-Tolerant Conversion between the Steane and Reed-Muller Quantum Codes},'' \href{http://dx.doi.org/10.1103/PhysRevLett.113.080501}{{\em Phys. Rev. Lett.} {\bfseries 113} no.~8, (2014) 080501}, \href{http://arxiv.org/abs/1403.2734}{{\ttfamily arXiv:1403.2734 [quant-ph]}}.

\bibitem{bombin2016dimensionaljumpquantumerror}
H.~Bomb{\'\i}n, ``{Dimensional jump in quantum error correction},'' \href{http://dx.doi.org/10.1088/1367-2630/18/4/043038}{{\em New J. Phys.} {\bfseries 18} no.~4, (2016) 043038}, \href{http://arxiv.org/abs/1412.5079}{{\ttfamily arXiv:1412.5079 [quant-ph]}}.

\bibitem{Beverland2021}
M.~E. Beverland, A.~Kubica, and K.~M. Svore, ``Cost of universality: A comparative study of the overhead of state distillation and code switching with color codes,'' \href{http://dx.doi.org/10.1103/PRXQuantum.2.020341}{{\em PRX Quantum} {\bfseries 2} (Jun, 2021) 020341}. \url{https://link.aps.org/doi/10.1103/PRXQuantum.2.020341}.

\bibitem{Gottesman_1999}
D.~Gottesman and I.~L. Chuang, ``Demonstrating the viability of universal quantum computation using teleportation and single-qubit operations,'' \href{http://dx.doi.org/10.1038/46503}{{\em Nature} {\bfseries 402} no.~6760, (Nov., 1999) 390–393}. \url{http://dx.doi.org/10.1038/46503}.

\bibitem{Davydova:2025ylx}
M.~Davydova, A.~Bauer, J.~C.~M. de~la Fuente, M.~Webster, D.~J. Williamson, and B.~J. Brown, ``{Universal fault tolerant quantum computation in 2D without getting tied in knots},'' \href{http://arxiv.org/abs/2503.15751}{{\ttfamily arXiv:2503.15751 [quant-ph]}}.

\bibitem{shor1997faulttolerantquantumcomputation}
P.~W. Shor, ``{Fault-tolerant quantum computation},'' \href{http://dx.doi.org/10.1109/SFCS.1996.548464}{{\em Proceedings of 37th Conference on Foundations of Computer Science} (5, 1996) }, \href{http://arxiv.org/abs/quant-ph/9605011}{{\ttfamily arXiv:quant-ph/9605011}}.

\bibitem{Bombin_2007}
H.~Bombin and M.~A. Martin-Delgado, ``Topological computation without braiding,'' \href{http://dx.doi.org/10.1103/physrevlett.98.160502}{{\em Physical Review Letters} {\bfseries 98} no.~16, (Apr., 2007) }. \url{http://dx.doi.org/10.1103/PhysRevLett.98.160502}.

\bibitem{anderson2025controlledgatescliffordhierarchy}
J.~T. Anderson and M.~Weippert, ``{Controlled Gates in the Clifford Hierarchy},'' \href{http://arxiv.org/abs/2410.04711}{{\ttfamily arXiv:2410.04711 [quant-ph]}}.

\bibitem{Moussa:2016kgp}
J.~E. Moussa, ``{Transversal Clifford gates on folded surface codes},'' \href{http://dx.doi.org/10.1103/PhysRevA.94.042316}{{\em Phys. Rev. A} {\bfseries 94} no.~4, (2016) 042316}, \href{http://arxiv.org/abs/1603.02286}{{\ttfamily arXiv:1603.02286 [quant-ph]}}.

\bibitem{Zheng:2024dpx}
Y.~Zheng and D.~E. Liu, ``{From Magic State Distillation to Dynamical Systems},'' \href{http://dx.doi.org/10.22331/q-2025-09-15-1858}{{\em Quantum} {\bfseries 9} (2025) 1858}, \href{http://arxiv.org/abs/2412.04402}{{\ttfamily arXiv:2412.04402 [quant-ph]}}.

\bibitem{Campbell:2009jko}
E.~T. Campbell and D.~E. Browne, ``{On the Structure of Protocols for Magic State Distillation},'' \href{http://dx.doi.org/10.1007/978-3-642-10698-9_3}{{\em Lect. Notes Comput. Sci.} {\bfseries 5906} (2009) 20--32}, \href{http://arxiv.org/abs/0908.0838}{{\ttfamily arXiv:0908.0838 [quant-ph]}}.

\bibitem{Bombin:2015tpp}
H.~Bomb{\'\i}n, ``{Gauge color codes: optimal transversal gates and gauge fixing in topological stabilizer codes},'' \href{http://dx.doi.org/10.1088/1367-2630/17/8/083002}{{\em New J. Phys.} {\bfseries 17} no.~8, (2015) 083002}, \href{http://arxiv.org/abs/1311.0879}{{\ttfamily arXiv:1311.0879 [quant-ph]}}.

\bibitem{Kubica:2014jue}
A.~Kubica and M.~E. Beverland, ``{Universal transversal gates with color codes: A simplified approach},'' \href{http://dx.doi.org/10.1103/PhysRevA.91.032330}{{\em Phys. Rev. A} {\bfseries 91} no.~3, (2015) 032330}, \href{http://arxiv.org/abs/1410.0069}{{\ttfamily arXiv:1410.0069 [quant-ph]}}.

\bibitem{Vasmer:2019ovd}
M.~Vasmer and D.~E. Browne, ``{Three-dimensional surface codes: Transversal gates and fault-tolerant architectures},'' \href{http://dx.doi.org/10.1103/PhysRevA.100.012312}{{\em Phys. Rev. A} {\bfseries 100} no.~1, (2019) 012312}, \href{http://arxiv.org/abs/1801.04255}{{\ttfamily arXiv:1801.04255 [quant-ph]}}.

\bibitem{Eastin:2009tem}
B.~Eastin and E.~Knill, ``{Restrictions on Transversal Encoded Quantum Gate Sets},'' \href{http://dx.doi.org/10.1103/PhysRevLett.102.110502}{{\em Phys. Rev. Lett.} {\bfseries 102} no.~11, (2009) 110502}, \href{http://arxiv.org/abs/0811.4262}{{\ttfamily arXiv:0811.4262 [quant-ph]}}.

\bibitem{Bravyi:2012rnv}
S.~Bravyi and R.~Koenig, ``{Classification of Topologically Protected Gates for Local Stabilizer Codes},'' \href{http://dx.doi.org/10.1103/physrevlett.110.170503}{{\em Phys. Rev. Lett.} {\bfseries 110} no.~17, (2013) 170503}, \href{http://arxiv.org/abs/1206.1609}{{\ttfamily arXiv:1206.1609 [quant-ph]}}.

\bibitem{Sullivan:2023qsg}
M.~Sullivan, ``{Code conversion with the quantum Golay code for a universal transversal gate set},'' \href{http://dx.doi.org/10.1103/PhysRevA.109.042416}{{\em Phys. Rev. A} {\bfseries 109} no.~4, (2024) 042416}, \href{http://arxiv.org/abs/2307.14425}{{\ttfamily arXiv:2307.14425 [quant-ph]}}.

\bibitem{Heussen:2024vtv}
S.~Heu{\ss}en and J.~Hilder, ``{Efficient fault-tolerant code switching via one-way transversal CNOT gates},'' \href{http://dx.doi.org/10.22331/q-2025-09-03-1846}{{\em Quantum} {\bfseries 9} (2025) 1846}, \href{http://arxiv.org/abs/2409.13465}{{\ttfamily arXiv:2409.13465 [quant-ph]}}.

\bibitem{gidney2021stim}
C.~Gidney, ``Stim: a fast stabilizer circuit simulator,'' \href{http://dx.doi.org/10.22331/q-2021-07-06-497}{{\em {Quantum}} {\bfseries 5} (July, 2021) 497}. \url{https://doi.org/10.22331/q-2021-07-06-497}.

\bibitem{CirqDevelopers_2025}
C.~Developers, \href{http://dx.doi.org/10.5281/ZENODO.4062499}{{\em Cirq}}.
\newblock Zenodo, Aug., 2025.
\newblock \url{https://zenodo.org/doi/10.5281/zenodo.4062499}.

\bibitem{Steane_1997}
A.~M. Steane, ``Active stabilization, quantum computation, and quantum state synthesis,'' \href{http://dx.doi.org/10.1103/physrevlett.78.2252}{{\em Physical Review Letters} {\bfseries 78} no.~11, (Mar., 1997) 2252–2255}. \url{http://dx.doi.org/10.1103/PhysRevLett.78.2252}.

\bibitem{Mahler:2013mdg}
D.~H. Mahler, L.~A. Rozema, A.~Darabi, C.~Ferrie, R.~Blume-Kohout, and A.~M. Steinberg, ``{Adaptive Quantum State Tomography Improves Accuracy Quadratically},'' \href{http://dx.doi.org/10.1103/PhysRevLett.111.183601}{{\em Phys. Rev. Lett.} {\bfseries 111} no.~18, (2013) 183601}, \href{http://arxiv.org/abs/1303.0436}{{\ttfamily arXiv:1303.0436 [quant-ph]}}.

\bibitem{Lee:2025lcs}
S.-u. Lee, M.~Yuan, S.~Chen, K.~Tsubouchi, and L.~Jiang, ``{Efficient benchmarking of logical magic state},'' \href{http://arxiv.org/abs/2505.09687}{{\ttfamily arXiv:2505.09687 [quant-ph]}}.

\bibitem{bombin20182dquantumcomputation3d}
H.~Bombin, ``2d quantum computation with 3d topological codes,'' 2018.
\newblock \url{https://arxiv.org/abs/1810.09571}.

\bibitem{kobayashi2025cliffordhierarchystabilizercodes}
R.~Kobayashi, G.~Zhu, and P.-S. Hsin, ``Clifford hierarchy stabilizer codes: Transversal non-clifford gates and magic,'' 2025.
\newblock \url{https://arxiv.org/abs/2511.02900}.

\bibitem{bauer2025planarfaulttolerantcircuitsnonclifford}
A.~Bauer and J.~C.~M. de~la Fuente, ``Planar fault-tolerant circuits for non-clifford gates on the 2d color code,'' 2025.
\newblock \url{https://arxiv.org/abs/2505.05175}.

\bibitem{wills2024constantoverheadmagicstatedistillation}
A.~Wills, M.-H. Hsieh, and H.~Yamasaki, ``Constant-overhead magic state distillation,'' 2024.
\newblock \url{https://arxiv.org/abs/2408.07764}.

\bibitem{golowich2024asymptoticallygoodquantumcodes}
L.~Golowich and V.~Guruswami, ``{Asymptotically Good Quantum Codes with Transversal Non-Clifford Gates},'' \href{http://arxiv.org/abs/2408.09254}{{\ttfamily arXiv:2408.09254 [quant-ph]}}.

\bibitem{nguyen2024goodbinaryquantumcodes}
Q.~T. Nguyen, ``Good binary quantum codes with transversal ccz gate,'' 2024.
\newblock \url{https://arxiv.org/abs/2408.10140}.

\bibitem{golowich2024quantumldpccodestransversal}
L.~Golowich and T.-C. Lin, ``Quantum ldpc codes with transversal non-clifford gates via products of algebraic codes,'' 2024.
\newblock \url{https://arxiv.org/abs/2410.14662}.

\end{thebibliography}\endgroup
\bibliographystyle{utphys}


\clearpage
\onecolumngrid

\appendix



\section{Error Model}
\label{app:error_model}

We consider an error model given by the standard circuit-level Pauli noise. That is, each quantum operation is modeled as a noiseless operation, preceded or followed by some stochastic Pauli error occurring with a given probability. Hence, this error model encompass initialization $P_I$, gate-specific $\{P_i\}$ and measurement $P_M$ errors. For the numerical simulations we consider a uniform circuit-level error model with strength $p$, where
\begin{itemize}
    \item \textbf{Initialization errors $P_I$.} Preparation of $\ket{0}$ ($\ket{+}$) noisy physical states is given by the initialization of noiseless $\ket{0}$ states followed by $X$ ($Z$) errors occurring with probability $p$. Similarly, preparation of physical magic states is followed by some Pauli error $Q$ occurring with probability $p$; see Appendix~\ref{sub:initialization} for the justification. For $|T\rangle$ and $|H\rangle$ magic states, the Pauli errors $Q$ are single-qubit Paulis $Z$ and $Y$, respectively.
    \item \textbf{Gate errors $\{P_i\}$.} Single-qubit (two-qubit) gates are given by the respective noiseless gates followed by uniform single-qubit $\mathcal{D}_1(\rho)$ (two-qubit $\mathcal{D}_2(\rho)$) depolarizing noise channels of strength $p$,
    \begin{equation}
    \quad \mathcal{D}_1(\rho)=(1-p)\rho+\frac{p}{3}\sum_{Q \in \{X,Y,Z\}} Q\rho Q \quad , \quad \mathcal{D}_2(\rho)=(1-p)\rho+\frac{p}{15}\sum_{Q \in \{I,X,Y,Z\}^{\otimes 2}/\{I\otimes I\}}\: Q\rho Q \:. 
    \end{equation}
    \item \textbf{Gate errors $\{P_i\}$.} ``Idling"  qubits (i.e., qubits that are acted on trivially) are followed by single-qubit depolarizing noise channel $\mathcal{D}_1(\rho)$ of strength $p$. Note that for $n$ consecutive idling operations the effective error channel is given by $n$-time composition of $\mathcal{D}_1^{\otimes n}(\rho)=(\mathcal{D}_1\circ \dots \circ \mathcal{D}_1)(\rho)$, where
    \begin{equation}
    \mathcal{D}_1^{\otimes n}(\rho)=(1-p_n)\rho+\frac{p_n}{3}\sum_{Q \in \{X,Y,Z\}}\: Q\rho Q \quad , \quad p_n=\frac{3}{4}\left[1-\left(1-\frac{4}{3}p\right)^n\right]\:.
    \label{eq:dep_idling}
    \end{equation}  
    \item \textbf{Measurement errors $P_M$.} Measurement of qubits in the $Z$ basis ($X$ basis) is preceded by an $X$ error ($Z$ error) occurring with probability $p$.
\end{itemize}

\subsection{Initialization Error: PSC Eigenstates}
\label{sub:initialization}

Here we justify the Pauli initialization error for eigenstates of PSCs [Definition~\ref{def:square_root_pauli}]. The main claim is that any error model acting on such a state can be converted to a Pauli-based error model. This is achieved by randomly applying the PSC. We first focus on the single-qubit PSCs.

\begin{proposition}
Let $C$ be a single-qubit order-$2$ PSC and $|\psi\rangle$ be its $+1$ eigenstate. Consider a twirling map $\mathcal{T}_{C}(\cdot) = \frac{1}{2}((\cdot) + C(\cdot)C^{\dagger})$. For any channel $\mathcal{E}$,
\begin{equation}
    \mathcal{T}_C\circ\mathcal{E}(|\psi{\rangle\langle}\psi|) = (1-p)|\psi{\rangle\langle}\psi| + p P|\psi{\rangle\langle}\psi|P^{\dagger}
    \label{eq:twirl_single_qub}
\end{equation}
for some Pauli $P$ and $0\leq p\leq 1$.
\label{prop:pauli_frame_init}
\end{proposition}

\begin{proof}

Note that the $\pm 1$ eigenstates of $C$ are $|v^{\pm}\rangle=\frac{I\pm C}{2}|s\rangle$
for a stabilizer state $|s\rangle$. Therefore, for any single-qubit density matrix $\rho=\sum_{a,b\in \{+,-\}}\alpha_{ab}|v^{a}\rangle \langle v^{b}|$ and some complex coefficients $\alpha_{ab}$, the twirling map acting on $\rho$ decoheres it onto the eigenbasis of $C$,
\begin{equation}
    \mathcal{T}_C(\rho)=\alpha_{++}|v^{+}\rangle \langle v^{+}|+\alpha_{--}|v^{-}\rangle \langle v^{-}| \:.
\end{equation}
Further, if $|\psi \rangle \equiv  |v^{+}\rangle$ we need to prove that exists a Pauli $P$ such that $|v^{-}\rangle = P |\psi\rangle$. This is equivalent to proving that given $C$, there exists a Pauli $P$ that anticommutes with $U$. Since $U^2=1$, the standard form (\ref{eq:standard_form}) of $C=Qe^{i\frac{\pi}{4}Q'}$ for Paulis $Q$, $Q'$ such that $\{Q,Q'\}=0$ and $Q^2=I$. Therefore, $P\equiv Q'$ is such that $\{P,C\}=0$. Hence, in this case twirling a physical magic state $|\psi\rangle \langle \psi |$ subjected to any noise channel $\mathcal{E}$ guarantees that only Pauli error $P$ influence it, which implies (\ref{eq:twirl_single_qub}).   
\end{proof}

The result of Proposition~\ref{prop:pauli_frame_init} can be also generalized to a larger class of multi-qubit magic state as in Proposition~\ref{prop:pauli_frame_init_diag}.

\begin{proposition}
Consider an $l$-qubit magic state $|V\rangle = V|+\rangle^{\otimes l}$, where $V$ is a diagonal gate in $\mathcal{C}^{(3)}$. Let $S_V=\langle VX_iV^{\dagger}: i\in [l]\rangle$ and define the twirling map
\begin{equation}
    \mathcal{T}_V(\cdot) = \frac{1}{|S_V|} \sum_{s\in S_V} s(\cdot) s^{\dagger}.
\end{equation}
For any channel $\mathcal{E}$,
\begin{equation}
    \mathcal{T}_V\circ \mathcal{E}(|V{\rangle\langle}V|) = \sum_{x\in \mathbb{F}_2^l} p_x Z_x |V{\rangle\langle}V| Z_x^{\dagger},
    \label{eq:twirl_diag_prop}
\end{equation}
for some probability distribution $\{p_x\}$, where $Z_x := \bigotimes_{i=1}^{l} Z_i^{x_i}$ and $x \in \mathbb{F}_2^l$.

    \label{prop:pauli_frame_init_diag}
\end{proposition}
\begin{proof}
Let $|V\rangle = V|+\rangle^{\otimes l}$ be a $l$-qubit magic state, where $V$ is a diagonal gate from the third level of the Clifford hierarchy. The stabilizer group $S_V$ of $|V\rangle$ is generated by unitaries $s_z$ such that $s_z|V\rangle=|V\rangle$, where
\begin{equation}
    S_V=\left \{ s_z: s_z=V\left(\bigotimes_{i=1}^l X^{z_i}_i\right)V^{\dagger}\:,\: z \in \mathbb{F}_2^l\right\} \:.
\end{equation}
We now show that there is a twirling operator $\mathcal{T}_V(\rho)$ that decoheres any $l$-qubit density matrix $\rho$ into a basis $\mathcal{B}$ consisting of projectors $P|V\rangle \langle V|P^{\dagger}$ for some set of Paulis $P$. Indeed, the basis is defined as $\mathcal{B}=\left\{|\tilde{x}\rangle \langle \tilde{x}|: |\tilde{x}\rangle= \left(\bigotimes_{i=1}^l Z^{x_i}_i\right)|V\rangle \:,\: x \in \mathbb{F}_2^l \right\}$, where $s_z |\tilde{x}\rangle = (-1)^{x \cdot z} |\tilde{x}\rangle$, $x,z \in \mathbb{F}_2^l$ and $x \cdot z =\sum_{i=1}^l x_iz_i$. The twirling operator is then
\begin{equation}
    \mathcal{T}_V(\rho)=\frac{1}{2^l}\sum_{g_z\in S}s_z\rho s_z^{\dagger}\:.
\end{equation}
The action $\mathcal{T}_V(\rho)$ on basis elements $|\tilde{x}\rangle \langle \tilde{y}|$ of $\rho$, represented by respective binary strings $x,y \in \mathbb{F}_2^l$, is
\begin{equation}
    \frac{1}{2^l}\sum_{s_z\in S}s_z|\tilde{x}\rangle \langle \tilde{y}| s_z^{\dagger}=\delta_{x,y} |\tilde{x}\rangle \langle \tilde{x}|\:,
\label{eq:delta_states}
\end{equation}
where the Kronecker delta $\delta_{x,y}$ is 1 if $x=y$ and 0 otherwise. Therefore, due to (\ref{eq:delta_states})
\begin{equation}
    \rho=\sum_{x,y \in \mathbb{F}_2^l} \alpha_{xy}|\tilde{x}\rangle \langle \tilde{y}| \quad \Longrightarrow \quad \mathcal{T}(\rho)= \sum_{x \in \mathbb{F}_2^l} \alpha_{xx}|\tilde{x}\rangle \langle \tilde{x}|\:,
\end{equation}
for some complex coefficients $\alpha_{xy}$. Hence, in this case twirling a physical magic state $|V\rangle \langle V|$ subjected to any noise channel $\mathcal{E}$ guarantees that only $Z$-type Pauli errors $P$ influence it, which implies (\ref{eq:twirl_diag_prop}).
\end{proof}

Propositions~\ref{prop:pauli_frame_init}~and~\ref{prop:pauli_frame_init_diag} imply that applying stochastically $C$ (or more generally, a group generated by a set of PSCs) to a physical magic state $|\psi\rangle$ stabilized by $C$, errors become at most a Pauli $P$. We conjecture that this property generalizes to any $k$-qubit magic state stabilized by $k$ independent and commuting PSCs, which is left for future work.

\section{Properties of PSC}
\label{app:proof_sec4}

In this Appendix, we provide proofs of the statements in Sec.~\ref{sec:controlled_clifford}.

\LemmaPSCHierarchy*
\begin{proof}
    The ``only if'' part of the statement follows straightforwardly from Eqs.~\eqref{eq:identity_general_1} and~\eqref{eq:identity_general_2}. For the ``if'' part of the statement, our proof closely follows Theorem A.2 from~\cite{anderson2025controlledgatescliffordhierarchy} for their $k = 3$ case. Notice that in Eq.~\eqref{eq:identity_general_2}, under our assumption that $C(U) \in \mathcal{C}^{(3)}$, the unitary $V=(X \otimes U) C(U^\dag{}^2)$ must be a Clifford gate. 
    We can recursively apply the circuit identity of Eq.~\eqref{eq:identity_general_2} too
    \begin{equation}
        \begin{adjustbox}{height=0.6cm}
            \begin{quantikz}
                & \gate{X} & \ctrl{1} & \\
                & & \gate{U^2} &
            \end{quantikz}
        \end{adjustbox}
        =
        \begin{adjustbox}{height=0.6cm}
            \begin{quantikz}
                & \ctrl{1} & \ctrl{1} & \gate{X} & \\
                & \gate{U^2} & \gate{ U^{\dag}{}^4} & \gate{U^2} &
            \end{quantikz}
        \end{adjustbox}.
    \label{eq:identity_general_recurse}
    \end{equation}
    Moreover, note that
    \begin{equation}
        V (X \otimes I) V^\dagger = C(U^2) (X \otimes I) C(U^\dag{}^2)\:.
        \label{eq:ident_V_U_quad}
    \end{equation}
    Since $V=(X \otimes U) C(U^\dag{}^2)$ is a Clifford, from Eqs. (\ref{eq:identity_general_recurse}) and (\ref{eq:ident_V_U_quad}) we see that $W=(X \otimes U^2) C(U^\dag{}^4) \in \mathcal{C}^{(1)}$. This implies $U^4=I$ (up to a global phase) because there are no controlled-gates in $\mathcal{C}^{(1)}$, as well as $U^2 \in \mathcal{C}^{(1)}$. Moreover, using the definition of $V$ we can rewrite $X \otimes U=V\cdot C(U^2)$. Since $U^2$ is a Pauli then $C(U^2)$ is a Clifford. Therefore, because $V$ is a Clifford, $V\cdot C(U^2)=X\otimes U$ is a Clifford too. By the group structure of Clifford gates, $X\otimes U$ being a Clifford implies that $U$ must be a Clifford too. To conclude, $U$ is a Clifford such that $U^2$ is a Pauli (hence, it is a PSC).
\end{proof}

\PropositionPSCNormalForm*
\begin{proof}
The ``if'' part of the statement can be shown easily from the following identity:
\begin{equation}
    U^2 = \alpha^2 Pe^{\frac{\pi i}{4}\sum_{j=1}^mQ_j} P e^{-\frac{\pi i}{4}\sum_{j=1}^mQ_j} e^{\frac{\pi i}{2}\sum_{j=1}^mQ_j}.
\end{equation}
Since $e^{\frac{\pi i}{4}\sum_{j=1}^mQ_j}$ is a Clifford, conjugating $P$ by this operator results in a Pauli [Eq.~\eqref{eq:expPi4}]. Moreover, 
\begin{equation}
    e^{\frac{\pi i}{2}\sum_{j=1}^mQ_j} =i^m \prod_{j=1}^m Q_j,
\end{equation}
which is again a Pauli. Since $\alpha^8=1$, $\alpha^2=+1, -1, +i,$ or $-i$. Thus $U^2$ is a Pauli.

Less obvious is the ``only if" statement. To that end, note that two Cliffords $U_1$ and $U_2$ are equivalent (up to a global phase) if and only if 
\begin{equation}
    U_1PU_1^{\dagger} = U_2PU_2^{\dagger}
\label{eq:equivalence}
\end{equation}
for every Pauli $P$. Thus, if we can show that a conjugation by a PSC is equivalent to a conjugation by a unitary of the form of $Pe^{\frac{\pi i}{4}\sum_{j=1}^mQ_j}$, we are done.

A convenient framework to study the action of the Cliffords up to Pauli is the symplectic representation, reviewed below. In this framework, each $n$-qubit Pauli can be written as a $2n$-dimensional binary vector $u$ and a Clifford can be expressed as a $2n\times 2n$ binary invertible matrix $C$ that preserves the symplectic form:
\begin{equation}
    C^T \omega
    C = \omega,
\label{eq:symplectic}
\end{equation}
where $\omega$ is the binary symplectic matrix
\begin{equation}
    \omega= \begin{pmatrix}
        0_{n\times n} & I_{n\times n} \\
        I_{n\times n} & 0_{n\times n}
    \end{pmatrix}.
\end{equation}
In particular, the Clifford preserves the symplectic product, i.e., $\langle Cv, Cu\rangle = \langle v, u\rangle$, where $\langle v, u\rangle := v^T \omega u$. Moreover, two Paulis $P$ and $Q$ commute if and only if their corresponding binary vector has a zero symplectic inner product. 

In the symplectic representation, a PSC is a symplectic matrix (\ref{eq:symplectic}) that squares to the identity since Paulis have trivial form in this representation. Without loss of generality, let $C=I+N$. Then, $C^2=I$ implies $N^2=0$ since $N$ is a binary matrix, i.e., $N$ is nilpotent. Furthermore, $C^2=I$ also implies $\langle Nx, Ny\rangle=0$ for any $x, y$ and that $\langle Nx, y\rangle = \langle x, Ny\rangle$. Thus, for any $u, v\in \text{Im}(N)$, $\langle u, v\rangle=0$. This implies that the Paulis associated to the elements of $\text{Im}(N)$ must necessarily commute. 

Now we claim that
\begin{equation}
    N = \omega \sum_k u_k u_k^T,
\end{equation}
where $u_k\in \text{Im}(\omega N)$. To see why, note $\omega N$ is symmetric due to the relation $\langle Nx, y\rangle = \langle x, Ny\rangle$. Let $P = VV'$ be the projector onto $\text{Im}(\omega N)$, where $V'$ is a $r\times n$ matrix for some $r\leq n$. Then,
\begin{equation}
\begin{aligned}
    \omega N &= P \omega N P^T \\
    &= V(V' \omega N (V')^T) V^T,
\end{aligned}
\end{equation}
where the matrix in the parenthesis is a symmetric $r\times r$ matrix. Because it is a symmetric matrix over $\mathbb{F}_2$, it can be written in the following form:
\begin{equation}
V' \omega N (V')^T = \sum_{k} v_k v_k^T,
\end{equation}
where $v_k \in \mathbb{F}_2^k$. Thus, 
\begin{equation}
    \omega N = N\omega = \sum_k u_k u_k^T
\end{equation}
for some $u_k \in \text{Im}(\omega N)$.

Now we decompose $N= \sum_k N_k$, where $N_k =  u_k u_k^T \omega$. Note that because $u_k\in \text{Im}(\omega N)$, $\langle u_i, u_j\rangle=0$ for all $i, j$. Therefore, $N_i N_j=0$ for all $i, j$ and in particular each $N_i$ is nilpotent. Thus we arrive at the following decomposition,
\begin{equation}
    C =I+\sum_{j=1}^m N_j= \prod_{j=1}^m (I+ N_j).
\label{eq:C_product_rep}
\end{equation}
Note that $(I+N_j)(x)$ is in one-to-one correspondence with $e^{\frac{\pi i}{4}Q_j}Pe^{-\frac{\pi i}{4}Q_j}$ (up to a Pauli), where $Q_j$ is a Pauli associated with $u_j$ and $P$ is an arbitrary Pauli with binary representation $x$. Indeed, it can be checked that the equivalence in terms of Eq.~(\ref{eq:equivalence}) holds since $e^{\frac{\pi i}{4}Q_j}$ maps $P$ as in Eq.~(\ref{eq:expPi4}), and $(I+N_j)(x)=x+\langle u_j,x\rangle u_j$, which is either $x$ if $\langle u_j,x\rangle =0$ or $x+u_j$ if $\langle u_j,x\rangle=1$. Moreover, because any two elements in $\text{Im}(\omega N)$ have zero symplectic inner product, $[Q_j, Q_k]=0$ for all $j$ and $k$. This completes the proof. 
\end{proof}

\PropositionPSCPropagationCaseOne*
\begin{proof}
    Note Eq.~\eqref{eq:propagation_psc} follows from the identity~(\ref{eq:identity_general_1}) when $U=P$, $V=C$ and $Q= CPC^{\dagger}P^{\dagger}$. Because $C$ is a Clifford, it follows that $Q$ is a Pauli. Below we show that $Q$ either commutes or anticommutes with $C$.

    From Proposition~\ref{prop:psc_normal_form}, it follows that $C$ can be written in the following form
    \begin{equation}
        C = \alpha \tilde{P} \exp\left(\frac{\pi i}{4}\sum_{j=1}^m Q_j \right),
    \end{equation}
    where $\tilde{P}$ is a Pauli and $\{Q_j\}$ is a set of commuting Paulis and $\alpha$ is a complex number with a unit modulus. Recall the following standard identity:
    \begin{equation}
        e^{-\frac{\pi i}{4}P' }P e^{\frac{\pi i}{4}P' } = 
        \begin{cases}
            P & \,\, \text{if} \,\, [P,P']=0, \\ 
            -iP'P & \,\, \text{if} \,\, \{P,P' \}=0.
        \end{cases}
    \label{eq:expPi4}
    \end{equation}
    Repeatedly applying this identity, we obtain
    \begin{equation}
        Q = \alpha \prod_{j\in \mathcal{A}_P} Q_j,
    \end{equation}
    where $\alpha$ is a scalar of modulus $1$ and $\mathcal{A}_P$ is a subset of $\{Q_j:j=1,\ldots m\}$ that anticommute with $P$. 

    By construction, the following commutation relation holds:
    \begin{equation}
        \left[\exp\left(\frac{\pi i}{4}\sum_{j=1}^m Q_j \right), Q \right]=0.
    \end{equation}
    Therefore, the commutation relation of $C$ and $Q$ is determined by that of $\tilde{P}$ and $Q$. Since both are Paulis, it follows that $C$ and $Q$ either commute or anticommute. 
\end{proof}

\section{Error Propagation}
\label{app:proof_sec5}

In Appendix~\ref{app:proof_theo2}, we provide a proof of the asymptotic gate complexity of a Clifford error resulting from error propagation through a standard logical PSC measurement-based protocol from Theorem~\ref{thm:standard-protocol-error-prop} in Sec.~\ref{subsec:logical-clifford-meas}. In Appendix~\ref{app:structure_PSC}, we provide supporting results regarding the structure of propagated errors for such protocols.

\subsection{Proof of Theorem \ref{thm:standard-protocol-error-prop}}
\label{app:proof_theo2}

We consider the most general setting of MSP protocols consisting of repeated stabilizer and logical PSC measurements that are interleaved in an arbitrary way. Using Fig.~\ref{fig:propagation_rules_visual}, we can guarantee that circuit-level errors will \textit{only} ever propagate to a Clifford. What remains to be studied is the gate complexity of the resulting Clifford error circuit. We show this grows modestly with various parameters of the protocol.

\StandardProtocolErrorProp*

\begin{proof}
    Without loss of generality, it suffices to consider single-qubit errors occurring on either data or ancilla qubits. Furthermore, at any given time errors will be propagated through a round of stabilizer measurements or a round of logical PSC measurements.

    \textit{Propagation through a data qubit.} Given a data qubit is involved in at most $w_q$ stabilizer checks, propagation through a round of stabilizer measurements yields at most an additional $w_q$ single-qubit Pauli errors on ancillae. Hence across $r_S$ rounds, at most $r_Sw_q$ additional single-qubit Pauli errors are propagated to ancillae. Moreover, since we do not consider ancilla qubit resets, these Pauli errors will not propagate any further in a nontrivial way. 

    For propagation through a logical PSC measurement round, a given data qubit is in the support of a \textit{single} PSC $V_i$ [Definition~\ref{def:transversal-psc}] and, therefore, in the support of $C(V_i)$'s target. By Proposition~\ref{prop:psc_propagation_case1}, a single-qubit Pauli error will propagate an additional controlled-Pauli $C(Q)$ for $Q$ a Pauli that either commutes or anticommutes with $V_i$ and satisfies $\mathrm{supp}(Q) \subseteq \mathrm{supp}(V_i)$. As per our assumption, this latter fact implies that $Q$ must be an $O(1)$-qubit Pauli. Across $r_L$ rounds of logical PSC measurement, $r_L$ $C(Q)$ errors are propagated. Further, each $C(Q)$ will propagate additional $C(Z)$ errors [Eq.~\eqref{eq:same-target-diff-control-commutation}]. In particular, each data qubit in the target of $C(Q)$ is involved in at most $w_q$ stabilizer checks. Hence, each of the $r_L$ $C(Q)$s will propagate at most $r_Sw_q$ additional $C(Z)$ gates. Further, each of the $r_L$ $C(Q)$ errors will encounter $r_L$ $C(V_i)$s from logical PSC measurement rounds, resulting in at most an additional $r_L^2$ $C(Z)$. 

    Hence, propagation of single-qubit circuit-level errors from data qubits results in a Clifford error with $O(r_Sr_Lw_q + r_L^2)$ one- and two-qubit gates.

    \textit{Propagation through an ancilla qubit.} Ancilla qubits may be used for stabilizer measurement of PSC measurement. We treat each case separately.

    We start with the former case. Given each stabilizer check involves at most $w_c$ data qubits, a single-qubit Pauli error will propagate to at most $w_c$ single-qubit Pauli errors on data qubits. Unlike ancillae, data qubits are used throughout the protocol and can, therefore, continue to propagate nontrivial errors. From the analysis of the error propagation on data qubits, it follows that each of the (at most) $w_c$ single-qubit Pauli errors propagates to a Clifford error with $O(r_Sr_Lw_q + r_L^2)$ gates.

    We now consider the latter case. Each such logical PSC measurement ancilla will be the control for $\ell$ $C(V_i)$ gates [Definition~\ref{def:transversal-psc}]. Thus across $r_L$ rounds of logical PSC measurement, a single-qubit circuit-level error on such an ancilla propagates to an additional (at most) $\ell r_L$ $C(Q)$ gates and $\ell r_L$ $V_i$ gates [Eq.~\eqref{eq:identity_general_2}]. From above, we know that $\ell r_L$ $C(Q)$ gates will propagate an additional $O(\ell r_Sr_Lw_q + \ell r_L^2)$ $C(Z)$ gates. Then, each of the $\ell r_L$ PSC gates $V_i$ will trivially propagate past logical PSC measurements as each $V_i$ has disjoint support and $V_i$ necessarily commutes with $V_j$ for $i = j$. For stabilizer measurements, each of the $O(1)$ qubits in the support of $V_i$ are involved in at most $w_q$ stabilizer checks. Thus, each of $\ell r_L$ $V_i$ gates will propagate an additional $C(R)$ gate, defined similarly to $C(Q)$ above. Finally, each of the $\ell r_L$ $C(R)$ gates will also propagate to additional $O(\ell r_Sr_Lw_q + \ell r_L^2)$ $C(Z)$ gates. 

    Hence, propagation of single-qubit circuit-level errors from ancilla qubits results in a Clifford error with $O\left( (w_c + \ell)(r_Sr_Lw_q + r_L^2)\right)$ one- and two-qubit gates.

    In summary, any single-qubit error propagates to a Clifford error with $O\left(r^2w_q(w_c + \ell)\right)$ one- and two-qubit gates as desired.
\end{proof}

\subsection{Structure of propagated errors}

\label{app:structure_PSC}

The main purpose of this section is to elucidate the structure of propagated errors in PSC measurement protocols. While the rules for propagating circuit-level errors through such protocols are already summarized in Fig.~\ref{fig:propagation_rules_visual}, we note that the propagated Clifford error has an additional internal structure. Although we do not use this structure in our simulation, we envision this to be useful in optimizing the error propagation algorithm, which may be useful for future work.

We begin by showing how circuit-level noise propagates through a single round [Eq.~\eqref{eq:non-FT-noisy-controlled-logical-U-circ}] and multiple rounds [Eq.~\eqref{eq:non-FT-r-rounds-noisy-controlled-logical-U-circ}] of PSC measurements. We then show how noise from a single round of PSC measurements propagates through a single round of stabilizer measurements [Figs.~\ref{fig:pro_stab_meas} and ~\ref{fig:corr-propagated-stab-meas-errors}]. Finally, we discuss how to understand the structure of propagated errors for standard protocols [Sec.~\ref{subsec:logical-clifford-meas}] with $r_L$ contiguous logical PSC measurement rounds followed by $r_S$ contiguous stabilizer measurement rounds. In greater generality, we may consider error propagation through standard protocols with $r_L' = r_L^{(1)} + \dots + r_L^{(k)}$ PSC measurement rounds and $r_S' = r_S^{(1)} + \dots + r_S^{(k)}$ stabilizer measurement rounds where $r_L^{(i)}$ logical PSC measurement rounds are alternated with $r_S^{(i)}$ stabilizer measurement rounds. Error propagation through Shor-style FT versions [Sec.~\ref{subsec:logical-clifford-meas}] of these measurements generalizes straightforwardly from the non-FT case.

Let $\bar{U}$ be a transversal logical PSC of a stabilizer code. A noisy circuit implementing the PSC measurement $C(\bar{U})$ and its equivalent propagated Clifford error look as follows:

\begin{equation}
    \begin{adjustbox}{height=2.5cm}
            \begin{quantikz}[column sep=0.15cm,wire types={q,q,n,q,q,q,n},font=\large]
            \lstick{$\ket{+}$} & &&&&& \gate[1, style={noisy_1}]{\text{$P_{\ket{+}}$}} & \ctrl{1} & \gate[style={noisy_1}]{\text{$D$}} & \ \ldots \ &  \ctrl{3} & \gate[style={noisy_1}]{\text{$D$}} & \ctrl{4} & \gate[style={noisy_1}]{\text{$D$}} & \ \ldots \ & \gate[1, style={noisy_1}]{\text{$P_M$}} & \meter[1]{x}\\
            \lstick{1} & \qwbundle{n_1} &&&&&  \gate[1, style={noisy_1}]{\text{$P_1$}} & \gate{V_1} & \gate[style={noisy_1}]{\text{$D$}} & \ \ldots \ && \gate[1, style={noisy_2}]{\text{$Idl$}} && \gate[1, style={noisy_2}]{\text{$Idl$}} &\\
            \lstick{$\vdots$} \\
            \lstick{$i$} & \qwbundle{n_i} &&&&& \gate[1, style={noisy_1}]{\text{$P_i$}} && \gate[1, style={noisy_2}]{\text{$Idl$}} & \ \ldots \ & \gate{V_i} & \gate[style={noisy_1}]{\text{$D$}} && \gate[1, style={noisy_2}]{\text{$Idl$}} &\\
            \lstick{$i+1$} & \qwbundle{n_{i+1}} &&&&& \gate[1, style={noisy_1}]{\text{$P_{i+1}$}} && \gate[1, style={noisy_2}]{\text{$Idl$}} & \ \ldots \ && \gate[1, style={noisy_2}]{\text{$Idl$}} & \gate{V_{i+1}} & \gate[style={noisy_1}]{\text{$D$}} &\\
            \lstick{$\vdots$} 
        \end{quantikz}
    \end{adjustbox} = 
    \begin{adjustbox}{height=2.5cm}
            \begin{quantikz}[column sep=0.15cm,wire types={q,q,n,q,q,q,n},font=\large]
            \lstick{$\ket{+}$} & \gate[5]{\text{\shortstack{PSC $ \bar{U}$ \\ Meas.}}} & \gate[5, style={noisy_1}]{P} & \ctrl{1} & \ \ldots \ & \ctrl{3} & \ctrl{4} && \ \ldots \ & \meter{x}\\
            \lstick{1} & &&\gate[style={noisy_1}]{Q_1} & \ \ldots \ & & & \gate[style={noisy_1}]{V_1} & \\
            \lstick{$\vdots$} \\
            \lstick{$i$} & & & & \ \ldots \ & \gate[style={noisy_1}]{Q_i} & & \gate[style={noisy_1}]{V_i} & \\
            \lstick{$i+1$} & & & & \ \ldots \ & & \gate[style={noisy_1}]{Q_{i+1}} & \gate[style={noisy_1}]{V_{i+1}} & \\
            \lstick{$\vdots$} 
        \end{quantikz}\:,
    \end{adjustbox}
    \label{eq:non-FT-noisy-controlled-logical-U-circ}
\end{equation}
where the first row of Pauli errors $P_I = P_{\ket{+}} \otimes P_1 \otimes \dots \otimes P_i \otimes \dots $ is due to logical state initialization and the remaining errors are idling errors $Idl$, depolarizing errors $D$, and measurement errors $P_M$ [Appendix~\ref{app:error_model}]. Additionally, all Pauli errors can all be pushed to the left side of the propagated Clifford error circuit. The above equivalence is obtained by a repeated application of Theorem~\ref{thm:propagation_psc}, leaving a noiseless circuit implementing $C(\bar{U})$ followed by a Clifford error circuit.

Now, consider propagating circuit-level errors through a circuit measuring a PSC $\bar{U}$ for $r$ rounds. A noisy circuit implementing this looks as follows:

\begin{equation}
        \begin{adjustbox}{height=2cm}
                \begin{quantikz}[column sep=0.25cm,font=\large,wire types={q,q,n,q,q}]
                \lstick{$\ket{+}_1$} && \gate[1, style={noisy_1}]{\text{$P_{I_1}$}} &  \ctrl{4} & \gate[1, style={noisy_2}]{\text{$Idl$}} & \ \ldots \ & \gate[1, style={noisy_2}]{\text{$Idl$}} & \gate[1, style={noisy_1}]{\text{$P_{M_1}$}} & \meter[1]{x}\\
                \lstick{$\ket{+}_2$} && \gate[1, style={noisy_1}]{\text{$P_{I_2}$}} & \gate[1, style={noisy_2}]{\text{$Idl$}} &  \ctrl{3} & \ \ldots \ & \gate[1, style={noisy_2}]{\text{$Idl$}} & \gate[1, style={noisy_1}]{\text{$P_{M_2}$}} & \meter[1]{x}\\
                \lstick{$\vdots$} \\
                \lstick{$\ket{+}_r$} && \gate[1, style={noisy_1}]{\text{$P_{I_r}$}} & \gate[1, style={noisy_2}]{\text{$Idl$}} & \gate[1, style={noisy_2}]{\text{$Idl$}} & \ \ldots \ & \ctrl{1} & \gate[1, style={noisy_1}]{\text{$P_{M_r}$}} & \meter[1]{x}\\
                & \qwbundle{n} & \gate[1, style={noisy_1}]{\text{$\bigotimes_{i=1}^n P_i$}} & \gate{\text{$\bar{U}_1$}} & \gate{\text{$\bar{U}_2$}} & \ \ldots \ &  \gate{\text{$\bar{U}_r$}}
            \end{quantikz}\:
        \end{adjustbox}= 
        \begin{adjustbox}{height=2cm}
                \begin{quantikz}[column sep=0.25cm,wire types={q,q,n,q,q},font=\large]
                \lstick{$\ket{+}_1$} & & \gate[5]{\text{\shortstack{PSC $ \bar{U}$ \\ Meas. \\ $(\times r)$}}} & \gate[5,style={noisy_1}]{\text{\shortstack{Eq.~\eqref{eq:non-FT-noisy-controlled-logical-U-circ} \\ Error}}} & \gate[4,style={noisy_1}]{C(Z)_{a \rightarrow b}} & \ \ldots \ & \gate[5,style={noisy_1}]{\text{\shortstack{Eq.~\eqref{eq:non-FT-noisy-controlled-logical-U-circ} \\ Error}}} & \gate[4,style={noisy_1}]{C(Z)_{a \rightarrow b}} & \meter{x}\\
                \lstick{$\ket{+}_2$} &&&&& \ \ldots \ & & & \meter{x}\\
                \lstick{$\vdots$} \\
                \lstick{$\ket{+}_r$} &&&&& \ \ldots \ & & &  \meter{x}\\
                 & \qwbundle{n} &&&& \ \ldots \ & &  \\
            \end{quantikz}\:,
        \end{adjustbox}
    \label{eq:non-FT-r-rounds-noisy-controlled-logical-U-circ}
\end{equation}
where errors for each PSC measurement $C(\bar{U}_i)$ are as defined in Eq.~\eqref{eq:non-FT-noisy-controlled-logical-U-circ}. We can propagate errors for each of the $r$ rounds of PSC measurement in the same way as Eq.~\eqref{eq:non-FT-noisy-controlled-logical-U-circ} [Theorem~\ref{thm:propagation_psc}]. All PSC errors $V_i$ can be trivially propagated past as they only occur on the target qubit of future $C(V)_{anc \rightarrow i}$ gates. On the other hand, errors of the form $C(Q)_{a \rightarrow i}$ may incur additional $C(Z)_{a \rightarrow b}$ errors between $\ket{+}_a$ and $\ket{+}_b$ ancilla qubits when propagated past $C(V)_{b \rightarrow i}$ [Eq.~\eqref{eq:same-target-diff-control-commutation}]. Intuitively, we will have $r$ rounds of noiseless PSC measurement followed by many repeated blocks of errors as in the right-hand-side of Eq.~\eqref{eq:non-FT-noisy-controlled-logical-U-circ} and $C(Z)_{a \rightarrow b}$ errors between PSC measurement ancilla qubits. 

Consider the interim Clifford error due to propagation through $r$ rounds of measuring a PSC $\bar{U}$ as described above. We now discuss further propagation of this error through multiple rounds of stabilizer measurement, starting with understanding propagation through a (possibly partial) stabilizer generator measurement $C(M_i)$,

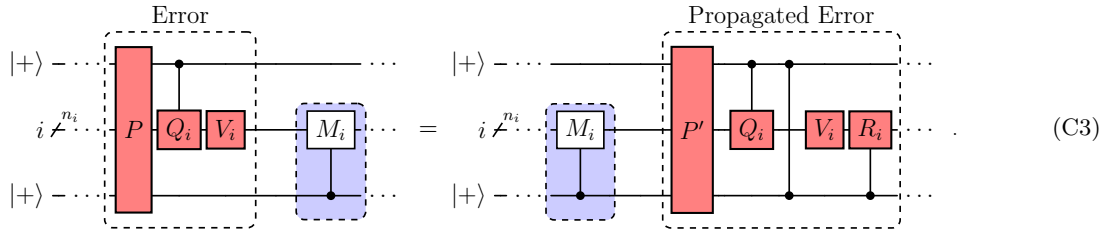
\begin{figure}[!ht]
\begin{gather}
        \begin{adjustbox}{height=1.8cm}
        \begin{quantikz}[column sep = 0.1cm, font=\large, wire types = {q,q,q}]
                \lstick{$\ket{+}$} & & \ \ldots \ & \gate[3, style={noisy_1}]{P} \gategroup[wires=3,steps=3,background,style={dashed, rounded corners, inner xsep=2pt}]{Error} & \ctrl{1} & &&& \ \ldots \ \\
                \lstick{$i$} & \qwbundle{n_i} & \ \ldots \ & & \gate[style={noisy_1}]{Q_i} & \gate[style={noisy_1}]{V_i} & \hphantomgate{wide} & \gate{M_i} \gategroup[wires=2,steps=1,background,style={dashed,rounded corners, fill=blue!20,inner xsep=2pt, inner ysep=0.1pt}, label style={label position = below,yshift=0.2cm}]{} & \ \ldots \ \\
                \lstick{$|+\rangle$} & & \ \ldots \ & & & & & \ctrl{-1} & \ \ldots \
        \end{quantikz}
        \end{adjustbox} =
        \begin{adjustbox}{height=1.8cm}
        \begin{quantikz}[column sep = 0.1cm, font=\large, wire types = {q,q,q}]
                \lstick{$\ket{+}$} & & \ \ldots \ & & & & \gate[3,style={noisy_1}]{\text{$P'$}} \gategroup[wires=3,steps=8,background,style={dashed, rounded corners, inner xsep=1pt}]{Propagated Error} & & & \ctrl{1} & & \ctrl{2} & & & \ \ldots \ \\
                \lstick{$i$} & \qwbundle{n_i} & \ \ldots \ & \gate{M_i} \gategroup[wires=2,steps=1,background,style={dashed,rounded corners, fill=blue!20,inner xsep=2pt, inner ysep=0.1pt}, label style={label position = below,yshift=0.2cm}]{}  & & \hphantomgate{wide} & & & & \gate[style={noisy_1}]{Q_i} & & & \gate[style={noisy_1}]{V_i} & \gate[style={noisy_1}]{R_i} & \ \ldots \ \\
                \lstick{$|+\rangle$} & & \ \ldots \ & \ctrl{-1} & & & & & & & & \phase{} & & \ctrl{-1} & \ \ldots \ 
        \end{quantikz}\:.
        \end{adjustbox} 
    \label{eq:lemma-propagated-stab-meas-CP-errors}   
\end{gather}
\caption{\textbf{Error propagation through a stabilizer measurement gate $C(M_i)$.} Here, $M_i \in \{I,X,Y,Z\}^{\otimes n_i}$, for an error comprised of a Pauli $P$, a Pauli-square-root Clifford $V_i$ and Pauli $Q_i$ that either commutes or anticommutes with $V_i$. Note that in this case, a $C(Z)$ error \textit{only} occurs if $Q_i$ is a Pauli that anticommutes with $M_i$ [Eq.~\eqref{eq:same-target-diff-control-commutation}].}
\label{fig:pro_stab_meas}
\end{figure}

The equivalence in Fig.~\ref{fig:pro_stab_meas} is obtained by applying Proposition~\ref{prop:Cliff_prop} and Eq.~\eqref{eq:same-target-diff-control-commutation}; the above equivalence can then be generalized to a full round of stabilizer measurements by repeatedly applying these rules [Fig.~\ref{fig:corr-propagated-stab-meas-errors}]. In particular, each error $C(R)_{\mathrm{anc} \rightarrow i}$ and $C(Q)_{\mathrm{anc} \rightarrow i}$ may incur additional $C(Z)_{\mathrm{anc}_a \rightarrow \mathrm{anc}_b}$ between different stabilizer measurement ancillas and between stabilizer measurement and PSC control ancillas, respectively [Eq.\eqref{eq:same-target-diff-control-commutation}], each PSC error $V_i$ will incur additional $C(R)_{\mathrm{anc} \rightarrow i}$ errors [Proposition~\ref{prop:Cliff_prop}], and $C(Z)$ errors between ancillas are trivially commuted as ancillas are never reused. Finally, Pauli errors $P'$ will propagate to further Pauli errors as stabilizer measurements are comprised entirely of Clifford gates. Thus, we recover Fig.~\ref{fig:corr-propagated-stab-meas-errors}.

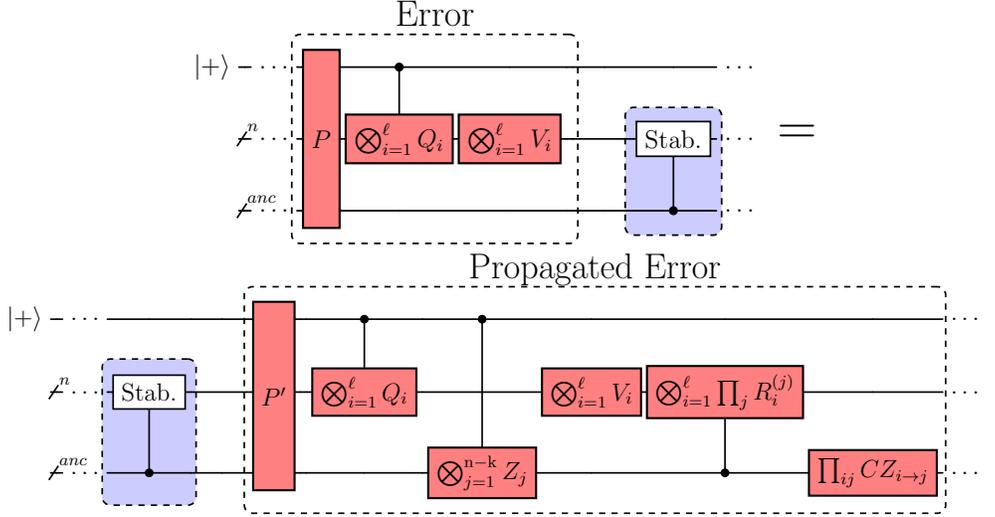
\begin{figure*}[!ht]
        \begin{adjustbox}{height=2cm}
        \begin{quantikz}[column sep = 0.1cm, font=\large, wire types = {q,q,q}]
                \lstick{$\ket{+}$} && \ \ldots \ & \gate[3, style={noisy_1}]{P} \gategroup[wires=3,steps=4,background,style={dashed, rounded corners, inner xsep=2pt}]{{\LARGE Error}} & \ctrl{1} & &&&&& \ \ldots \ \\
                & \qwbundle{n} & \ \ldots \ & & \gate[style={noisy_1}]{\bigotimes_{i=1}^\ell Q_i} & \gate[style={noisy_1}]{\bigotimes_{i=1}^\ell V_i} & & & \hphantomgate{wide} & \gate{\text{\shortstack{Stab.}}} \gategroup[wires=2,steps=1,background,style={dashed,rounded corners, fill=blue!20,inner xsep=2pt, inner ysep=0.1pt}, label style={label position = below,yshift=0.2cm}]{} & \ \ldots \ \\
                & \qwbundle{anc} & \ \ldots \ & & & & & & & \ctrl{-1} & \ \ldots \
        \end{quantikz}
        \end{adjustbox}
        {\huge =} 
        \begin{adjustbox}{height=2cm}
        \begin{quantikz}[column sep = 0.1cm, font=\large, wire types = {q,q,q}]
                \lstick{$\ket{+}$} && \ \ldots \ & & & & \gate[3,style={noisy_1}]{\text{$P'$}}\gategroup[wires=3,steps=9,background,style={dashed, rounded corners, inner xsep=1pt}]{{\LARGE Propagated Error}} & & & \ctrl{1} & & \ctrl{2} & & & & \ \ldots \ \\
                & \qwbundle{n} &\ \ldots \ & \gate{\text{\shortstack{Stab.}}} \gategroup[wires=2,steps=1,background,style={dashed,rounded corners, fill=blue!20,inner xsep=2pt, inner ysep=0.1pt}, label style={label position = below,yshift=0.2cm}]{}  & & \hphantomgate{wide} & & & & \gate[style={noisy_1}]{\bigotimes_{i=1}^\ell Q_i} & & & \gate[style={noisy_1}]{\bigotimes_{i=1}^\ell V_i} & \gate[style={noisy_1}]{\bigotimes_{i=1}^\ell \prod_{j} R_i^{(j)}} &  & \ \ldots \ \\
                & \qwbundle{anc} & \ \ldots \ & \ctrl{-1} & & & & & & & & \gate[style={noisy_1}]{\bigotimes_{j=1}^{\mathrm{n-k}}Z_j} & & \ctrl{-1} & \gate[style={noisy_1}]{\prod_{ij} CZ_{i \rightarrow j}} & \ \ldots \ 
        \end{quantikz}\:.
        \end{adjustbox}
    \caption{\textbf{Error propagation of Clifford errors produced by a previous round of Pauli-square-root Clifford measurements through a round of stabilizer measurements.} This figure generalizes error propagation depicted in Fig.~\ref{fig:pro_stab_meas} to all data qubits and all controlled-Pauli gates present in stabilizer measurement circuits. $P$ and $P'$ are Pauli errors, $C\left(\bigotimes_{i=1}^\ell Q_{\mathrm{anc} \rightarrow i}\right)$ are controlled-$Q_i$ errors for $Q_i$ either commuting or anticommuting with Pauli-square-root Clifford $V_i$ [Theorem~\ref{thm:propagation_psc}], $C\left(\bigotimes_{j=1}^\mathrm{n-k}Z_{\mathrm{anc} \rightarrow j}\right)$ are controlled-$Z$ errors that may occur between Pauli-square-root Clifford measurement control ancillae ($|+\rangle$ state in the circuit) and stabilizer measurement ancillae (denoted by index $j$), $C(\bigotimes_{i=1}^\ell \prod_{j} R_{j \rightarrow i})$ are controlled-$R_i$ errors between stabilizer measurement ancillae (indexed by superscript $j$) and data qubits (indexed by subscript $i$) for $R_i$ either commuting or anticommuting with $V_i$ [Proposition~\ref{prop:Cliff_prop}], and $\prod_{ij} CZ_{i\rightarrow j}$ represents controlled-$Z$ errors that occur \textit{within} stabilizer measurement ancillae. As a note, $C(Z)$ errors only occur when the conditions of Eq.~\eqref{eq:same-target-diff-control-commutation} are met.}
    
    \label{fig:corr-propagated-stab-meas-errors}    
    \end{figure*}

We note that the Clifford error depicted in the bottom circuit of Fig.~\ref{fig:corr-propagated-stab-meas-errors} corresponds to propagation through a standard protocol for $r_L = 1$ and $r_S = 1$. Generalizing to repeated and interleaved blocks of $\{r_L^{(i)}\}_{i \in [k]}$ rounds of logical PSC measurements and $\{r_S^{(i)}\}_{i \in [k]}$ rounds of stabilizer measurements can be understood as repeated blocks of the propagated error from Fig.~\ref{fig:corr-propagated-stab-meas-errors} (up to additional $C(Z)$ errors strictly on ancilla qubits) with ancilla control qubits corresponding to the ancillae used in respective PSC and stabilizer measurement rounds.

\section{Magic state $|\bar{H}\rangle$ simulation details}
\label{app:prep_H}
In this Appendix we give additional details about the simulation technique for the specific magic $|\bar{H}\rangle$ state preparation protocol of Fig.~\ref{fig:CH-meas-circ} from Sec.~\ref{sub:sim-application-example}. 


We recall the propagation rules for Pauli errors (and phase gates) through $CH$ gates introduced in Sec.~\ref{sec:example},

\begin{equation}
 \begin{adjustbox}{width=0.3\columnwidth}
\begin{quantikz}
& \gate{X} & \ctrl{1} & \\ & & \gate{H} &  \end{quantikz} = \begin{quantikz}
& \ctrl{1} & \gate{X} & \\ &  \gate{H} &  \gate{H} &\end{quantikz}\quad , \quad
\end{adjustbox}
\begin{adjustbox}{width=0.3\columnwidth}
\begin{quantikz}
& \gate{S^a} & \ctrl{1} & \\ & & \gate{H} &  \end{quantikz} = \begin{quantikz}
& \ctrl{1} & \gate{S^a} & \\ &  \gate{H} &   &\end{quantikz}\:,
\end{adjustbox}
\end{equation}

\begin{equation}
 \begin{adjustbox}{width=0.4\columnwidth}
    \begin{quantikz}
&  & \ctrl{1} & \\ &\gate{X} & \gate{H} &  \end{quantikz} =
\begin{quantikz}
 & \ctrl{1} &  & \ctrl{1} & \gate{S} & \\ & \gate{H} & \gate{X} & \gate{Y}& &
\end{quantikz}\quad , \quad
\end{adjustbox}
    \begin{adjustbox}{width=0.4\columnwidth}
    \begin{quantikz}
&  & \ctrl{1} & \\ &\gate{Z} & \gate{H} &  \end{quantikz} =
\begin{quantikz}
 & \ctrl{1} &  & \ctrl{1} & \gate{S^{\dagger}} & \\ & \gate{H} & \gate{Z} & \gate{Y}& &
\end{quantikz}\:.
\end{adjustbox}
\end{equation}
Note that $a \in \{0,1,2,3\}$ and that $S^2=Z$ and $S^3=S^{\dagger}$. To propagate errors through the measurement of $\bar{H}$ of Fig.~\ref{fig:CH-meas-circ} we label error locations with a number and a color; see Fig.~\ref{fig:error_prop_sim}.

\begin{figure}[!h]
\centering
    \begin{tikzpicture}[scale=0.8, transform shape]
        \foreach \x in {1,...,9}
            \filldraw[fill=white, draw=white] (-2,0-\x+1) rectangle (0, 0-\x) node[pos=0.5] {Qubit \x};

        \foreach \x in {1,...,10}
            \filldraw[fill=white, draw=white] (\x-1,1) rectangle (\x,0) node[pos=0.5] {$t_{\text{\x}}$};
        
        \filldraw[fill=blue!40!white, draw=black] (0,0) rectangle (1,-1) node[pos=0.5] {39};
        \filldraw[fill=blue!40!white, draw=black] (0,-1) rectangle (1,-2) node[pos=0.5] {31};
        \filldraw[fill=blue!40!white, draw=black] (0,-2) rectangle (1,-3) node[pos=0.5] {26};
        \filldraw[fill=blue!40!white, draw=black] (0,-3) rectangle (1,-4) node[pos=0.5] {21};
        \filldraw[fill=blue!40!white, draw=black] (0,-4) rectangle (1,-5) node[pos=0.5] {16};
        \filldraw[fill=blue!40!white, draw=black] (0,-5) rectangle (1,-6) node[pos=0.5] {11};
        \filldraw[fill=blue!40!white, draw=black] (0,-6) rectangle (1,-7) node[pos=0.5] {2};
        \filldraw[fill=orange!40!white, draw=black] (0,-7) rectangle (1,-8) node[pos=0.5] {1};
        \filldraw[fill=orange!40!white, draw=black] (0,-8) rectangle (1,-9) node[pos=0.5] {6};
        
        \draw (1,0) rectangle (9,-1) node[pos=0.5] {40};
        \filldraw[fill=green!40!white, draw=black] (9,0) rectangle (10,-1) node[pos=0.5] {42};
        
        \draw (1,-1) rectangle (7,-2) node[pos=0.5] {32};
        \filldraw[fill=green!40!white, draw=black] (7,-1) rectangle (8,-2) node[pos=0.5] {34};
        \draw (8,-1) rectangle (10,-2) node[pos=0.5] {35};
        
        \draw (1,-2) rectangle (6,-3) node[pos=0.5] {27};
        \filldraw[fill=green!40!white, draw=black] (6,-2) rectangle (7,-3) node[pos=0.5] {29};
        \draw (7,-2) rectangle (10,-3) node[pos=0.5] {30};

        \draw (1,-3) rectangle (5,-4) node[pos=0.5] {22};
        \filldraw[fill=green!40!white, draw=black] (5,-3) rectangle (6,-4) node[pos=0.5] {24};
        \draw (6,-3) rectangle (10,-4) node[pos=0.5] {25};

        \draw (1,-4) rectangle (4,-5) node[pos=0.5] {17};
        \filldraw[fill=green!40!white, draw=black] (4,-4) rectangle (5,-5) node[pos=0.5] {19};
        \draw (5,-4) rectangle (10,-5) node[pos=0.5] {20};

        \draw (1,-5) rectangle (3,-6) node[pos=0.5] {12};
        \filldraw[fill=green!40!white, draw=black] (3,-5) rectangle (4,-6) node[pos=0.5] {14};
        \draw (4,-5) rectangle (10,-6) node[pos=0.5] {15};

        \filldraw[fill=green!40!white, draw=black] (1,-6) rectangle (2,-7) node[pos=0.5] {4};
        \draw (2,-6) rectangle (10,-7) node[pos=0.5] {5};
        
        \filldraw[fill=green!40!white, draw=black] (1,-7) rectangle (2,-8) node[pos=0.5] {3};
        \filldraw[fill=green!40!white, draw=black] (2,-7) rectangle (3,-8) node[pos=0.5] {8};
        \filldraw[fill=green!40!white, draw=black] (3,-7) rectangle (4,-8) node[pos=0.5] {13};
        \filldraw[fill=green!40!white, draw=black] (4,-7) rectangle (5,-8) node[pos=0.5] {18};
        \filldraw[fill=green!40!white, draw=black] (5,-7) rectangle (6,-8) node[pos=0.5] {23};
        \filldraw[fill=green!40!white, draw=black] (6,-7) rectangle (7,-8) node[pos=0.5] {28};
        \filldraw[fill=green!40!white, draw=black] (7,-7) rectangle (8,-8) node[pos=0.5] {33};
        \filldraw[fill=green!40!white, draw=black] (8,-7) rectangle (9,-8) node[pos=0.5] {36};
        \filldraw[fill=green!40!white, draw=black] (9,-7) rectangle (10,-8) node[pos=0.5] {41};

        \draw (1,-8) rectangle (2,-9) node[pos=0.5] {7};
        \filldraw[fill=green!40!white, draw=black] (2,-8) rectangle (3,-9) node[pos=0.5] {9};
        \draw (3,-8) rectangle (8,-9) node[pos=0.5] {10};
        \filldraw[fill=green!40!white, draw=black] (8,-8) rectangle (9,-9) node[pos=0.5] {37};
        \draw (9,-8) rectangle (10,-9) node[pos=0.5] {38};
    \end{tikzpicture}
    \caption{\textbf{Error locations for the flag-based measurement of $\bar{H}$ from Fig.~\ref{fig:CH-meas-circ} labeled with numbers and colors}. Blue boxes represent Pauli errors due to the non-fault-tolerant encoding circuit. White boxes represent idling errors, orange boxes represent initialization errors, and green boxes represent 2-qubit depolarizing errors. }
    \label{fig:error_prop_sim}
\end{figure}
The numerical simulation of the MSP protocol considers a circuit-level noise model that we describe in Appendix~\ref{app:error_model}. The different types of errors arising during the measurement of $\bar{H}$ in Fig.~\ref{fig:CH-meas-circ} can be labeled as follows: 
\begin{itemize}
    \item \textbf{Blue boxes} correspond to Pauli errors coming from the non fault-tolerant preparation of $|\bar{H}\rangle$. These errors have as binary representation for $a_i,b_i \in \{0,1\}$,
    \begin{equation}
        X_1^{a_{39}}Z_1^{b_{39}}\otimes X_2^{a_{31}}Z_2^{b_{31}}\otimes X_3^{a_{26}}Z_3^{b_{26}}\otimes X_4^{a_{21}}Z_4^{b_{21}}\otimes X_5^{a_{16}}Z_5^{b_{16}}\otimes X_6^{a_{11}}Z_6^{b_{11}}\otimes 
        X_7^{a_{2}}Z_7^{b_{2}}\:.
        \label{eq:pauli_init_frame_H}
    \end{equation}
    \item \textbf{Orange boxes} correspond to ancilla qubits initialization errors with $a_i,b_i \in \{0,1\}$,
    \begin{equation}
        X_8^{a_{1}}Z_8^{b_{1}}\otimes X_9^{a_{6}}Z_9^{b_{6}}\:.
    \end{equation}
    For the error model being under consideration $a_1=b_6=0$.
  
    \item \textbf{Blue boxes} correspond to two-qubit depolarizing channel errors applied after each two-qubit $\mathrm{CNOT}$ or $CH$ gate. The set of errors is given by
    \begin{equation}
        X_7^{a_4}Z_7^{b_4} \otimes X_8^{a_3}Z_8^{b_3}\quad ,\quad X_8^{a_8}Z_8^{b_8} \otimes X_9^{a_9}Z_9^{b_9} \quad ,\quad X_8^{a_{13}}Z_8^{b_{13}} \otimes X_6^{a_{14}}Z_6^{b_{14}} \quad , \quad  \cdots \quad,
    \end{equation}
    for $a_i,b_i \in \{0,1\}$.
    \item \textbf{White boxes} correspond to idling errors given implemented via single-qubit depolarizing channels. Note that for idling errors starting at step $t_i$ and ending at step $t_j$, they can effectively be sampled from a depolarizing channel $\mathcal{D}_1^{\otimes n}(\rho)$ with $n=j-i+1$, see Eq.~(\ref{eq:dep_idling}). The set of errors is given by
    \begin{equation}
        X_1^{a_{40}}Z_1^{b_{40}} \quad , \quad  X_2^{a_{32}}Z_2^{b_{32}} \quad , \quad X_3^{a_{35}}Z_3^{b_{35}}  \quad , \quad X_3^{a_{27}}Z_3^{b_{27}} \quad , \quad \cdots \quad,
    \end{equation}  
    for $a_i,b_i \in \{0,1\}$.
\end{itemize}
The final propagated Clifford error $C_\mathrm{prop}$ has an analytic form that can be written as the product
\begin{equation}
C_\mathrm{prop}=E_{9}E_{8}\cdots E_{1}\:.
\label{eq:Cprop_app}
\end{equation} 
Each of the errors $E_{i}$ for $i=1,\dots,7$ is given by
\begin{equation}
    E_7=(X_7^{a_4 + a_5}Z_7^{b_4 + b_5}H_7^{a_1}X_7^{a_2}Z_7^{b_2})(C_8Y_7)^{a_2+b_2} \:,
\end{equation}
\begin{equation}
    E_6=(X_6^{a_{14} + a_{15}}Z_6^{b_{14}+b_{15}}H_6^{a_1 + a_3 + a_8}X_6^{a_{11} + a_{12}}Z_6^{b_{11} + b_{12}}\big)\big(C_8Y_6\big)^{a_{11} + a_{12} + b_{11} + b_{12}} \:,
\end{equation}
\begin{equation}
    E_5=\big(X_5^{a_{19} + a_{20}}Z_5^{b_{19} + b_{20}}H_5^{a_1 + \sum_{i = 1}^3 a_{5i - 2}}X_5^{a_{16} + a_{17}}Z_5^{b_{16} + b_{17}}\big)(C_8Y_5)^{a_{16} + a_{17} + b_{16} + b_{17}} \:,
\end{equation}
\begin{equation}
    E_4= (X_4^{a_{24} + a_{25}}Z_4^{b_{24} + b_{25}}H_4^{a_1 + \sum_{i = 1}^4 a_{5i - 2}}X_4^{a_{21} + a_{22}}Z_4^{b_{21} + b_{22}}\big)(C_8Y_4)^{a_{21} + a_{22} + b_{21} + b_{22}} \:,
\end{equation}
\begin{equation}
    E_3=(X_3^{a_{29} + a_{30}}Z_3^{b_{29} + b_{30}}H_3^{a_1 + \sum_{i = 1}^5 a_{5i - 2}}X_3^{a_{26} + a_{27}}Z_3^{b_{26} + b_{27}}\big)(C_8Y_3)^{a_{26} + a_{27} + b_{26} + b_{27}} \:,
\end{equation}
\begin{equation}
    E_2=(X_2^{a_{34} + a_{35}}Z_2^{b_{34} + b_{35}}H_2^{a_1 + \sum_{i = 1}^6 a_{5i - 2}}X_2^{a_{31} + a_{32}}Z_2^{b_{31} + b_{32}}\big)(C_8Y_2)^{a_{31} + a_{32} + b_{31} + b_{32}} \:,
\end{equation}
\begin{equation}
    E_1=(X_1^{a_{42}}Z_1^{b_{42}}H^{a_1 + \sum_{i = 1}^6a_{5i-2} + a_{36}}X_1^{a_{39} + a_{40}}Z_1^{b_{39} + b_{40}})(C_8Y_1)^{a_{39} + a_{40} + b_{39} + b_{40}} \:.
\end{equation}
The contribution for the ancillas $E_8$ and $E_9$ is
\begin{equation}
    E_9=X_9^{a_6 + a_7 + a_9 + a_{10} + \sum_{i = 2}^6a_{5i-2} + a_{37} + a_{38}}Z_9^{b_6 + b_7+b_9+b_{10} + b_{37} + b_{38}} \:,
\end{equation}
\begin{equation}
\begin{split}
   E_8 &= X_8^{a_1 + \sum_{i = 1}^6a_{5i-2} + a_{36} + a_{41}}Z_8^{b_1 + a_2 + a_3 + b_8 + b_9 + b_{10} \sum_{i = 1}^5 a_{5i + 6} + a_{5i + 7} + b_{5i + 8} + b_{36} + a_{39} + a_{40} + b_{41}} \times\\
   &\quad \quad \times S_8^{a_2 + b_2 + \sum_{i = 1}^5 a_{5i + 6} + a_{5i + 7} + b_{5i + 6} + b_{5i + 7} + a_{39} + a_{40} + b_{39} + b_{40}}\:.
\end{split}
\end{equation}
In a numerical simulation, given a sampled initialization Pauli error frame (\ref{eq:pauli_init_frame_H}), the remaining of the coefficients $a_i,b_i$ can be sampled from the respective probability distributions dictated by their type (color and position following Fig.~\ref{fig:error_prop_sim}). Therefore, since the measurement $\bar{H}$ gadget acts trivially on the noiseless encoded state, all the non-Clifford gates can be effectively removed from the circuits at the expense of introducing stochastic Clifford errors $C_\mathrm{prop}$ acting on a noiseless logical magic state. The numerical results using this technique in a stabilizer-based emulator matches the numerical results obtained via state-vector simulations; see Figs.~\ref{fig:stim-vs-cirq-numerical-results1} and~\ref{fig:stim-vs-cirq-numerical-results2}.


\section{Error propagation algorithm}
\label{app:error_prop_alg}

We describe an algorithm to propagate sampled circuit-level errors to an end-of-circuit Clifford error in Monte Carlo simulations for MSP protocols. Implementation of such an algorithm becomes necessary when explicitly computing an end-of-circuit Clifford error distribution (as in Appendix~\ref{app:prep_H} or Sec.~\ref{subsec:transv-non-Cliff}) is inconvenient.

We consider a standard circuit-level noise model as described in Appendix~\ref{app:error_model}. Intuitively, our propagation algorithm works by pushing circuit-level errors through each round of (PSC or stabilizer) measurement according to predefined propagation rules (see Fig.~\ref{fig:propagation_rules_visual} as an example) to iteratively build a propagated Clifford error circuit  [Fig.~\ref{fig:propagation_alg_visual}]. For instance, propagating errors through $t$ measurement rounds of the protocol circuit will result in some intermediate propagated error circuit $C_\mathrm{prop}(t)$ with gate count $g_t$. Then, to compute $C_\mathrm{prop}(t+1)$, each of the $g_t$ gates of $C_\mathrm{prop}(t)$ are propagated past the $(t+1)$st measurement round. Subsequently, any sampled circuit-level errors in each measurement round are appended to the propagated error circuit. This process is repeated until the last round of measurements is performed. For protocols that measure transversal logical PSCs, this is accomplished  by applying propagation rules as per Theorem~\ref{thm:propagation_psc}, Proposition~\ref{prop:Cliff_prop}, and Eq.~\eqref{eq:same-target-diff-control-commutation}. We can proceed similarly for protocols based on code-switching by using the appropriate propagation rules for an available transversal logical non-Clifford gate.

We now give a formal presentation of our error propagation algorithm followed by an analysis of its time and space complexity; see Algorithm~\ref{alg:error_prop_alg} for the pseudocode. 

\begin{algorithm}[H]
    \DontPrintSemicolon
    \SetKwInOut{Input}{Input} 
    \SetKwInOut{Output}{Output}
    \SetKwFunction{PropError}{PropagateCliffordError}

    \Input{
        $M_\mathrm{sched} = \{R_i\}_{i \in [r]}$ - measurement schedule \\
        $P_I$ - logical state initialization error \\
        $\{P_j^{(i)}\}$ - circuit-level errors \\
        $P_M$ - measurement error
    }
    \Output{
        An end-of-circuit propagated error $C_\mathrm{prop}$, stored as an ordered list of Clifford gates.
    }
    \caption{Error Propagation Algorithm}
    \label{alg:error_prop_alg}

    \BlankLine
    \tcc{Initialize $C_\mathrm{prop}$ to $P_I$, noise due to logical state initialization [Appendix~\ref{sub:initialization}]}
    $C_\mathrm{prop} \gets P_I$\;
    \tcc{Iterate through logical PSC/stabilizer measurement rounds of measurement schedule}
    \For{$R_i$ in $M_\mathrm{sched}$} { \label{line:outer_loop}
        \tcc{Each $R_i = (U_1^{(i)},P_1^{(i)},U_2^{(i)},P_2^{(i)},\dots)$}
        \For{$G$ in $R_i$} { \label{line:inner_loop}
            \tcc{Propagate and append errors, alternatingly}
            \If{$G$ noiseless gate} {
                \tcc{Propagate current Clifford error $C_\mathrm{prop}(t)$ circuit through noiseless gate $U_j^{(i)}$}
                $C_\mathrm{prop} \gets $ \FuncSty{PropagateCliffordError}($C_\mathrm{prop}$, $G$)\;
            } 
            \ElseIf {$G$ noisy layer} {
                \tcc{Append sampled circuit-level noise $G$ to the end of $C_\mathrm{prop}$}
                $C_\mathrm{prop} \gets C_\mathrm{prop}$ + $G$\;
            }
        }
    }
    \Return {$C_\mathrm{prop}$}
    \BlankLine
    \BlankLine
    \SetKwProg{Fn}{Function}{:}{}
    \Fn{\PropError{$C_\mathrm{prop}$,$G$}}{
        $C_\mathrm{new} \gets \{\}$\;
        \For{$C_i$ in $C_\mathrm{prop}$} {
            \tcc{Propagate Clifford error gate $C_i$ past a protocol circuit gate $G$}
            $C_\mathrm{new} \gets $ \FuncSty{PropagateError}($C_i$, $G$) $ + \ C_\mathrm{new}$
        }
    }
    \Return {$C_\mathrm{new}$}\;
\end{algorithm}

\vspace{2mm}

\begin{figure}[!h]
    \centering
    \includegraphics[width=0.9\textwidth]{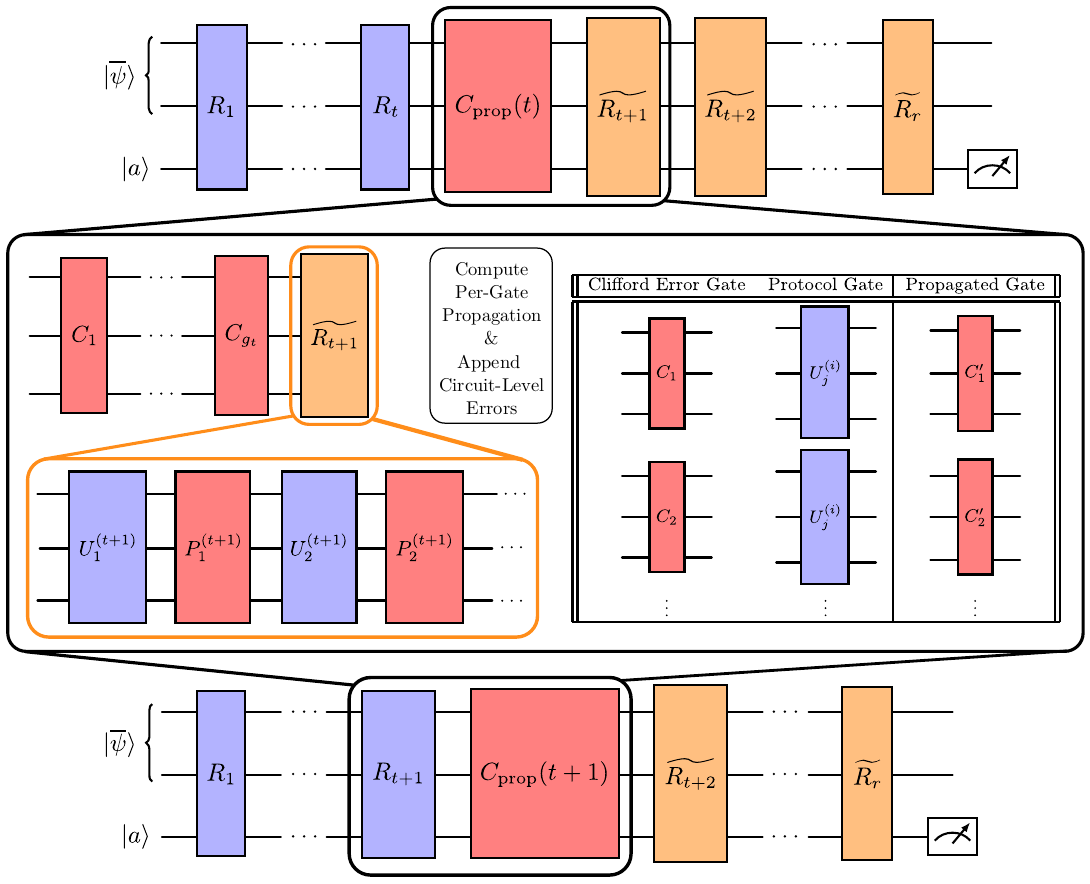}
    \caption{\textbf{Error propagation algorithm}. Circuit-level errors are pushed from left to right through the protocol circuit. For a given round of measurement $t$, an interim Clifford error circuit $C_\mathrm{prop}(t)$ is propagated past a subsequent noisy measurement round $\widetilde{R_{t+1}}$, where each such measurement round consists of noiseless one- or two-qubit gates $U_j^{(t+1)}$ and circuit-level noise $P_j^{(t+1)}$. For each of the gates $C_i$ in $C_\mathrm{prop}(t)$, a propagated error is retrieved from a precomputed look-up table that takes in as input $C_i$ and $U_j^{(t+1)}$ and returns the propagated gate $C_i'=U_j{^{(t+1)}}^\dag C_i U_j^{(t+1)}$. Then, sampled circuit-level errors $P_j^{(t+1)}$ are appended. Repeating this process for all $U_j^{(t+1)}$ and $P_j^{(t+1)}$ comprising a measurement round $\widetilde{R_{t+1}}$ gives the propagated Clifford error $C_\mathrm{prop}(t+1)$ and leaves a noiseless measurement round $R_{t+1}$.}
    \label{fig:propagation_alg_visual}
\end{figure}

We assume the protocol circuit acts on $n_\mathrm{tot}$ qubits with a measurement schedule $M_\mathrm{sched}$ consisting of $r = r_L + r_S$ rounds of logical PSC ($r_L$) and stabilizer ($r_S$) measurement and has $M$ error locations. We denote each of these $r$ measurement rounds by $R_i$ when \textit{noiseless} and by $\widetilde{R_i}$ when \textit{noisy} [Fig.~\ref{fig:propagation_alg_visual}]. Additionally, we assume the parity check matrix $H$ for the QEC code on which the logical magic state is being prepared has maximum row weight $w_c$ and maximum column weight $w_q$. We iteratively propagate past measurement rounds $R_i$ in the measurement schedule $M_\mathrm{sched}$, where each round $R_i$ consists of noiseless gate layers $U_j^{(i)}$ and noise layers $P_j^{(i)}$.

If a given measurement round corresponds to stabilizer measurements, the number of gates used is proportional to the number of nonzero elements of the parity check matrix, which is at most $n_{\mathrm{tot}}w_q$. If the measurement round corresponds to a PSC measurement, by the assumption that each $V_i$ has weight of $O(1)$, the number of gates we need is at most $n_{\mathrm{tot}}$. Either way, the number of gates for each round is bounded by $n_{\mathrm{tot}}w_q$.

At any given time, we are propagating a Clifford error $C_\mathrm{prop}(t)$ past a noisy measurement round $\widetilde{R_{t+1}}$. First, given there are $M$ error locations in the protocol circuit, we assume it takes $O(M)$ time to sample these errors and adds at most $M$ single-qubit Pauli gates to the propagated Clifford error. Then, from Corollary~\ref{corr:standard-protocol-full-error-prop} we know that the gate complexity of $C_\mathrm{prop}(t)$ is $g_t = O\left(Mt^2w_q(w_c + \ell)\right)$ (for $\ell=O(n_\mathrm{tot})$ the number of tensor factors in a logical PSC $\bar{U}$ [Definition~\ref{def:transversal-psc}]). Hence, to compute $C_\mathrm{prop}(t+1)$, we must propagate each of these $g_t$ gates past the measurement round $\widetilde{R_{t+1}}$ [Alg.~\ref{alg:error_prop_alg}, ln.~\ref{line:outer_loop}]; this is accomplished by propagating $C_\mathrm{prop}(t)$ past noiseless gates $U_{t+1}^{(i)}$ and appending (potentially trivial) sampled errors $P_{t+1}^{(i)}$ that comprise $\widetilde{R_{t+1}}$ [Alg.~\ref{alg:error_prop_alg}, ln.~\ref{line:inner_loop}].

Propagating each error gate $C_i$ of the $g_t$ gates of $C_\mathrm{prop}(t)$ past a noiseless gate $U_{t+1}^{(i)}$ is performed by accessing a precomputed look-up table [Fig.~\ref{fig:propagation_alg_visual}]. This look-up table takes in as input $C_i$ and $U_{t+1}^{(i)}$ and retrieves a propagated error $C_i' = U_{t+1}{^{(i)}}^{\dag} C_i U_{t+1}^{(i)}$ in $O(1)$ time. The propagated error $C_i'$ contributes at most $O(1)$ additional gates to the propagated error. We then append the sampled error $P_{t+1}^{(i)}$ which also takes $O(1)$ time. Across all $O(n_{\mathrm{tot}}w_q)$ gates in $\widetilde{R_{t+1}}$, propagation takes $O\left(Mt^2n_{\mathrm{tot}}w_q^2(w_c + \ell)\right)$ time. Extrapolating to all $r$ measurement rounds, propagation takes time 
\begin{equation}
    T_\mathrm{prop}(n_\mathrm{tot},r,w_q,w_c,\ell) = \sum_{t=1}^{r}O\left(Mt^2w_q^2n_\mathrm{tot}(w_c+\ell)\right) = O\left(Mr^3w_q^2n_\mathrm{tot}(w_c + \ell)\right)\:,
\end{equation}
and results in an end-of-circuit Clifford error taking space 
\begin{equation}
    S_\mathrm{prop}(n_\mathrm{tot},r,w_q,w_c,\ell) = O\left(Mr^2w_q(w_c + \ell)\right)\:,
\end{equation}
if stored explicitly as a list of one- and two-qubit gates and the qubits they act on. 

Alternatively, if a tableau representation for the propagated Clifford error is maintained, the space complexity may be reduced to $O(n_\mathrm{tot}^2)$. However, we found that storing the propagated Clifford error as a tableau representation makes it cumbersome to propagate past the non-Clifford gates in a systematic way. This is because our propagation rules are defined with respect to a set of Clifford gates. In order to apply these rules, the tableau representation would need to be compiled into these specific gates, apply the propagation rule, and then combined back into a tableau representation. A more direct way of updating the tableau representation past the non-Clifford gate would be thus desirable.


\section{Simulation via stabilizer rank decomposition}
\label{app:stab_rank_sim}

In this Appendix, we revisit the \textit{phase-sensitive} Clifford simulation method~\cite{Bravyi:2018ugg}, providing an alternative to the phase-insensitive simulation method introduced in Sec.~\ref{sec:fid_method} (which is based on the Pauli rank decomposition of the magic state).

From the discussion in Sec.~\ref{sec:canonical_family}, it follows that any circuit-level Pauli noise propagates to a Clifford error at the end of the circuit. More precisely, the state obtained after the error propagation is of the following form:
\begin{equation}
    C_{\textrm{prop}}(|\bar{\phi}\rangle|+\ldots +\rangle),
\end{equation}
where $|\bar{\phi}\rangle$ is the magic state encoded in the code subspace and $|+\ldots +\rangle$ represent the syndrome register, whose measurement is deferred to the end. We envision measuring the syndrome register, applying a correction (or postselection) to the noisy code state, and then computing the overlap of the resulting state with $|\bar{\phi}\rangle$. (We note that the syndrome register may contain additional qubits that are part of the flag gadgets.)

Without loss of generality, let the measurement basis be $\{|\tilde{x}\rangle: x\in \mathbb{F}_2^m \}$, where $|\tilde{x}\rangle = (\otimes_{i=1}^m H_i)|x\rangle$ and $U_x$ be the corresponding Pauli correction operation.\footnote{In principle, $U_x$ can be chosen as a Clifford correction. However, it is unclear if that leads to an advantage over Pauli correction.} (Here we are assuming that the syndrome register has $m$ qubits.) The state obtained after the correction is a density matrix of the following form
\begin{equation}
    \rho = \sum_x (U_x\otimes \langle\tilde{x}|) C_{\textrm{prop}}(|\bar{\phi}{\rangle\langle} \bar{\phi}| \otimes |+\ldots +{\rangle \langle} +\ldots +|) C_{\textrm{prop}}^{\dagger} (U_x^{\dagger}\otimes |\tilde{x}\rangle).
\end{equation}
Thus, the overlap between $\rho$ and $|\phi\rangle$ becomes 
\begin{equation}
    \langle \bar{\phi}|\rho |\bar{\phi}\rangle = \sum_x p_x |\langle \bar{\phi}|\tilde{\phi}_x\rangle|^2, \label{eq:overlap_stab_rank}
\end{equation}
where
\begin{equation}
    \begin{aligned}
        p_x &= \|(U_x\otimes \langle \tilde{x}|) C_{\textrm{prop}}(|\bar{\phi}\rangle\otimes |+\ldots +\rangle)\|^2, \\
        |\tilde{\phi}_x\rangle &= (U_x\otimes \langle \tilde{x}|) C_{\textrm{prop}}(|\bar{\phi}\rangle\otimes |+\ldots +\rangle) / \sqrt{p_x}.
    \end{aligned}
\end{equation}

Therefore, estimation of the overlap is reduced the problem of (i) sampling over the probability distribution $\{p_x: x\in \mathbb{F}_2^m \}$ and (ii) estimating the overlap $|\langle \bar{\phi}| \tilde{\phi}_x\rangle|^2$. Both problems can be solved using the stabilizer rank decomposition of $|\phi\rangle$. Without loss of generality, let $|\bar{\phi}\rangle = \sum_{i=1}^q \alpha_i |s_i\rangle$, where $\{|s_i\rangle \}$ is a set of normalized stabilizer states, $\{\alpha_i\}$ is a set of complex amplitudes, and $q$ is the stabilizer rank. Importantly, the stabilizer rank remains unchanged under stabilizer operations, e.g., Clifford unitary and measurement of Paulis. Moreover, the state can be updated efficiently under these operations. As we show below, these facts let us solve the aforementioned problems.

Let us first briefly note that the stabilizer rank decomposition lets us efficiently compute overlap of two stabilzer states, even including the phase. This is possible because any stabilizer state can be written in the so called CH-form~\cite{Bravyi:2018ugg}. Let $n_{\mathrm{tot}}$ be the total number of qubits. The CH form of the stabilizer state is
\begin{equation}
    |\Omega \rangle =\omega U_CU_H|s\rangle\:,
\end{equation}
with a complex global phase $\omega$, a `control'-type Clifford $U_C$ generated by $\{S,CX,CZ\}$, a `Hadamard'-type Clifford $U_H$ generated by $\{I,H\}$, and a computational basis vector $|s\rangle \in \{0,1\}^{\otimes n_\mathrm{tot}}$. Upon  application of a `control'-type Clifford gate, the CH-form of any state $|\Omega\rangle$ can be updated in $O(n_\mathrm{tot})$ time. Moreover, upon the application of a `Hadamard'-type gate or a projective Pauli measurement, the CH-form of any state $|\Omega\rangle$ can be updated in $O(n_\mathrm{tot}^2)$ time; see \cite[Proposition 4]{Bravyi:2018ugg}. Using this update rule, we can compute the overlap as follows. Let $|\Omega'\rangle = \omega' U_C' U_H'|s\rangle$ be another stabilizer state in a CH-form. Using the update rule, $\langle \Omega'| \Omega\rangle = \tilde{\omega} \langle s' | \tilde{U}_C \tilde{U}_H|s\rangle$ for some complex number $\tilde{\omega}$ and Cliffords $\tilde{U}_C \tilde{U}_H$. $\tilde{U}_C$ can be absorbed into $\langle s'|$, updating it into a new basis vector and a phase. The matrix element of the Hadamard-type Clifford can be computed straightforwardly at this point. Thus $\langle \Omega'|\Omega\rangle$ can be computed in time polynomial in $n_{\mathrm{tot}}$.

Using this overlap estimation as a subroutine, we now describe how to (i) sample over $\{p_x\}$ and (ii) compute $|\langle \bar{\phi}|\tilde{\phi}_x\rangle|^2$. The sampling can be done by recursively sampling each bit in the following way. Without loss of generality, consider measuring one of the qubits in the syndrome register. The postmeasurement state for both outcomes can be written as a $q$-fold linear combination of stabilizer states. In particular, the norm of each postmeasurement state can be estimated from $O(q^2)$ overlap estimations. Once these norms are estimated, one can sample over this binary probability distribution and pick the corresponding postmeasurement state. Since the stabilizer rank of the postmeasurement state is still at most $q$, this process can be repeated with the same complexity until the entire syndrome register is measured. Computing $|\langle \bar{\phi}|\tilde{\phi}_x\rangle|^2$ can be done in a similar way. Note that the stabilizer rank of the postmeasurement state is still at most $q$. Therefore, the overlap can be again computed using $O(q^2)$ overlap estimations.

To summarize, using the stabilizer rank decomposition, we can estimate the logical error rate in the following way. First, sample over the circuit-level noise and propagate it to an end-of-circuit Clifford error using the method in Sec.~\ref{sec:controlled_clifford}. Next, estimate Eq.~\eqref{eq:overlap_stab_rank} by averaging $|\langle\bar{\phi}|\tilde{\phi}_x\rangle|^2$ over the probability distribution $\{p_x\}$. Since all the sampling, error propagation, and overlap estimation can be done efficiently (insofar as $q$ is small), the logical error rate can be estimated efficiently.

One advantage of the stabilizer rank decomposition method, compared to the Pauli rank method [Section~\ref{sec:fid_method}], is its smaller variance. The Pauli rank-based method has $\Theta(1)$ variance. However, for the stabilizer rank based method, the logical error rate is an average of 1-$|\langle\bar{\phi}|\tilde{\phi}_x\rangle|^2 \in [0, 1]$. If the mean value is $p_L$, the variance is bounded from above by $p_L(1-p_L)$, which can be substantially smaller than that of the Pauli rank-based method when $p_L \ll 1$. Indeed, if $X$ is a random variable in $[0,1]$, the variance is $\mathbb{E}[X^2] - \mathbb{E}[X]^2 \leq \mathbb{E}[X] - \mathbb{E}[X]^2=p_L(1-p_L)$ (Here $\mathbb{E}$ is the expected value). However, from an implementation perspective, the stabilizer rank method is arguably more complicated. For future work, it will be desirable to compare the practical efficiencies of the two methods.

\subsection{Pauli rank and stabilizer rank are not related}
Although Pauli rank and stabilizer rank are both suitable nonstabilizerness metrics, it was shown in Ref.~\cite{Bu:2019qed} that one can always find an $n$-qubit state for which these nonstabilizerness metrics are exponentially-separated and, therefore, not related. In particular, consider the state $\ket{\psi_n} \propto \ket{0}^{\otimes n} + \ket{+}^{\otimes n}$. Clearly this state has stabilizer rank $q=2$ for all such $n$-qubit states. Additionally, we know that the Pauli rank of a stabilizer state is exactly $2^n$ (this can be shown by writing the density matrix of a stabilizer state $\ket{\phi}$ as a product of $n$ projectors onto +1-eigenspaces of stabilizer operators for $\ket{\phi}$). Since $\ket{\psi_n}$ is not a stabilizer state, its Pauli rank must be strictly greater than $2^n$. In fact, one can show that the Pauli rank of $\ket{\psi_n}$ is at least $p\geq 3^n-1$, hence providing a nontrivial example of a state with an exponential gap between its Pauli rank and stabilizer rank.


\end{document}